\documentclass[prodcode,acmtocl]{acmsmall}

\usepackage{url}
\usepackage{amsmath,amssymb}
\usepackage{cleveref}
\usepackage{macros}
\author{Michael Benedikt \affil{University of Oxford, UK} and Chia-Hsuan Lu \affil{University of Oxford, UK}  and Tony Tan \affil{University of Liverpool, UK}}
\begin{document}
\title{Decidability of Graph Neural Networks via Logical Characterizations}

\begin{abstract}
We present results concerning the expressiveness and decidability of a popular graph learning formalism, graph  neural networks (GNNs), exploiting  connections with  logic. We use  a family of recently-discovered decidable logics involving ``Presburger quantifiers''. We show how to use these logics to 
measure the expressiveness of classes of GNNs, in some cases getting exact correspondences between the expressiveness of logics and GNNs. We also employ the logics, and the techniques used to analyze them, to obtain decision procedures
for verification problems over GNNs. We complement this with undecidability results for static analysis problems involving the logics, as well as for GNN verification problems.   
\end{abstract}
\maketitle
\section{Introduction}

Graph Neural Networks (GNNs) have become the most common model for learning functions that work on graph data.
Like traditional neural networks, GNNs consist of a layered architecture where layer $k+1$ takes as input the output
of layer $k$. Each layer computes a function from graph vertices to a vector of numerical values -- the
\emph{feature vector}. Computation of the feature vector at layer $k+1$ for a node $u$ is based on aggregating vectors 
for layer $k$ of nodes $v$ that are related to $u$ in the source graph:
for example aggregating vectors associated
to nodes adjacent to $u$ in the graph. In an aggregation, the vectors of the previous vectors may be transformed using linear functions.
A layer can perform multiple aggregations -- corresponding to different linear functions --  and then combine them to get the feature vector for the next layer. The use of graph structure ensures that the computation of the network is \emph{invariant}: depending only on the input graph and the node up to isomorphism.
There are many variations of GNN. One key design choice is the kind of aggregation used - one can use ``local aggregation'', over the neighbors of a node, or aggregation over all nodes in the graph. 
A second design choice is the kind of numerical functions that can be applied to vector components, in particular the kind of \emph{activation functions} that can be applied at each layer: e.g. $\relu$, sigmoid, piecewise linear functions.

An important issue in the study of graph learning is the \emph{expressiveness} of a learning model. What kinds of computations can a given type of GNN express? 
The first results in this line were about the \emph{separating power of a graph learning model}: what pairs of nodes can be distinguished using GNNs within a certain class. For example, it is known that the separating power of standard GNN models is limited by the Weisfeiler-Leman (WL) test \cite{howpowerful}.

A finer-grained classification would  characterize  the functions computed by GNNs within a certain class, in terms of some formalism that is easier to analyze.  Such characterizations are referred to as \emph{uniform expressiveness results} and there has been much less work in this area.
\cite{barceloetallogical} provides a classification of a class of GNNs in terms of \emph{modal logic}. The main result in \cite{barceloetallogical} is a characterization of the classifiers expressible in first-order logic that can be performed with a GNN having only \emph{local aggregation} and \emph{truncated $\relu$ activations} over \emph{undirected graphs}.
They also provide a lower bound on the expressiveness of GNNs having in addition a ``global aggregator'', that sums over all nodes in the graph.

In this work we continue the line of work on uniform expressiveness.
Our work improves on the state of the art in a number of directions:
\begin{itemize}
\item \emph{From first order expressiveness to general expressiveness} In contrast to \cite{barceloetallogical}, we provide logical characterizations of \emph{all} the functions that can be computed by certain GNN formalisms, not just the intersection with first-order logic. To do this we utilize logics that go beyond first order, but which are still amenable to analysis.
\item \emph{From expressiveness to verification} While we deal with GNNs that go beyond first-order logic, we can still obtain characterizations in a logic where the basic satisfiability problems are decidable. This provides us with decidability of a number of natural verification problems related to GNNs. In doing this, we show a surprising link between GNNs and recently-devised decidable logics going beyond first-order logic, so-called \emph{Presburger logics}. 
\item \emph{From undirected graphs to directed graphs} While prior work focused on undirected graphs, we explore how the expressiveness characterizations vary with GNNs that can recognize directionality of graph edges. The aim is to show that the logical characterizations and decidability are often independent of the restriction to undirected graphs.
\item \emph{From bounded to unbounded activations} We explore the impact of the activation functions.
We begin with the case of \emph{bounded activation functions}, like the truncated $\relu$ of \cite{barceloetallogical}, and establish characterizations and decidability results for GNNs using this function. We show both some contrasts and some similarity to the case of  \emph{unbounded activation functions}, including the standard $\relu$. Here some, but not all, of the corresponding decidability results fail.
\end{itemize}

\myparagraph{Organization} After overviewing related work in Section~\ref{sec:related}, we formalize our GNN model and the basic logics we study in Section~\ref{sec:prelims}. In Section~\ref{sec:eventually_constant}, we start with results on logical characterizations of GNNs with ``bounded activation functions'' -- like the truncated $\relu$ of \cite{barceloetallogical}. We apply these characterizations to get decidability results. In Section~\ref{sec:pspace} we obtain tight complexity bounds for verification problems specifically for the case of GNNs with local aggregation and truncated $\relu$ activations.

Section~\ref{sec:unbounded} turns to the case of unbounded activation functions, which includes the traditional $\relu$ function. Here we provide lower bounds for expressiveness, and then turn to the implications for decidability. 
Section~\ref{sec:discuss} gives conclusions and discusses several open issues. 
\paragraph{Acknowledgements}
This paper elaborates on results announced in the conference paper \cite{usicalp24}. We thank the reviewers of ICALP for their feedback on the initial version.
\section{Related work} \label{sec:related}
Our work relates to three areas: logical characterizations of GNNs, verification of GNNs, and
decidable logics. We situate our contributions with respect to prior work in these areas below.

\myparagraph{Logical characterizations of GNN expressiveness} Logics have been used to characterized the separating power of GNN languages (``non-uniform expressiveness'') for a number of years: that is, graphs are distinguishable by a class of GNNs if and only if they are distinguisable in a certain logic. See \cite{grohegnnlogic} for an overview. There is much less work on ``uniform expressiveness'': characterizing the behaviour of GNNs on all graphs. The direct inspiration for our work is \cite{barceloetallogical}, which provides results characterizing the formulas within first-order logic that are expressible with GNNs. The recent \cite{grohedescriptivegnn}  provides logical characterizations of GNNs with piecewise linear activations. The logic is not decidable; indeed our undecidability results imply that one cannot capture such GNNs with a decidable logic.
The recent \cite{realcircuitgnn} gives characterizations of rich classes of GNNs using real circuits -- again, this does not yield any decidability.
.% \michael{check that this is true at the end}

\myparagraph{Verification of GNNs} We employ logical characterizations to gain insight on two basic verification problems -- whether a given classification can be achieved on some nodes or on all nodes. There is prior work on verification of GNNs, but it focuses on more complex (but arguably more realistic) problems, adversarial robustness. The closest paper to ours is the recent  
\cite{langeverifygnn}, which formalizes a broad set of problems related to verifying that the output is in a certain region in Euclidean space. \cite{langeverifygnn} provides both decidability and undecidability theorems, but they are incomparable to ours both in the results and in the techniques. For example Theorem 1 of \cite{langeverifygnn} shows undecidability of a satisfiability problem where we verify that certain nodes output a particular value, over GNNs which always distinguish a node from its neighbor. Theorem 2 of \cite{langeverifygnn} shows a decidability result with a different kind of specification, where the degree of input graphs is bounded.

Independently from our work, \cite{ijcai24gnnverification} looked at similar verification problems to ours, but for a more restricted class of GNNs: those with truncated $\relu$ and integer coefficients. They prove a $\pspace$-completeness result, very similar to our Theorem \ref{thm:pspace}.

\myparagraph{Logics combining graph structure with Presburger arithmetic}
An early work on combining logic on interpreted structure with Presburger arithmetic
is \cite{bapa}. There the uninterpreted structure consists only of unary predicates, so decidability is much simpler. The formalisms we study are closely related to \cite{localpresburgerbartosztony,twovarpres}. We rely on one decidability results from
\cite{localpresburgerbartosztony}, but we will need to refine the analysis to get our complexity results.

Recently, logics that combine uninterpreted relations with Presburger arithmetic have been applied to the analysis of transformers -- transducers that process strings  \cite{chiangcholaktighter,anthonypablo}. Since this is outside of the context of general graphs, the details of the logics that are employed are a bit different than those we consider, and the focus is not on the decidability border.
\section{Preliminaries} \label{sec:prelims}

Let $\bbN$, $\bbN^+$, $\bbZ$, and $\bbQ$ be the set of natural numbers, positive natural numbers, integers, and rational numbers, respectively.
%For a vector $b \in \bbQ^m$ and $1 \le i \le m$, $b_i$ is the $i^{th}$ entry of $b$. %michael: use superscript for i^{th}
%For a matrix $A \in \bbQ^{m\times n}$, $1 \le i \le m$, and $1 \le j \le n$, $A_{i, j}$ denotes the $j^{th}$ entry of the $i^{th}$ row of $A$.
%michael: I don't think it is necessary to define this (above)
For $p, q \in \bbZ$ with $p \le q$, $\intinterval{p}{q}$ is the set of integers between $p$ and $q$, including $p$ and $q$. When $p$ is $1$, we omit it and simply write $\intsinterval{q}$.
For $r \in \bbQ$, $\ceil{r}$ is the smallest integer greater than or equal to $r$, and
$\floor{r}$ is the greatest integer less than or equal to $r$.

For a function $f$ mapping from $\bbQ$ to $\bbQ$ and a vector $\bfv \in \bbQ^m$, $f(\bfv)$ denotes that $f$ is applied to each entry of $\bfv$.

\begin{definition}
    An \emph{$n$-graph} is a tuple $\tuple{V, E, \set{U_c}_{c \in \intsinterval{n}}}$,
    where $n \in \bbN$ is the number of vertex colors;
    $V$ is a nonempty finite set of vertices;
    $E \subseteq V \times V$ is a set of edges;
    each $U_c \subseteq V$ is the set of $c$-colored vertices.
\end{definition}
%michael: self-loops are ok, but maybe we don't need to say that

Note that we allow self-loops in graphs, and a graph is by default a \emph{directed graph}.
For a graph $\cG$, we say that $\cG$ is a \emph{undirected graph} if for all $v, u \in V$, $(v, u) \in E$ if and only if $(u, v) \in E$.
For a vertex $v$ in $\cG$, we let
${\nbr{\out, \cG}(v) := \setc{u}{(v, u) \in E}}$ and refer to this as the set of \emph{out-neighbors of $v$}. 
The set of \emph{in-neighbors of $v$}, denoted 
$\nbr{\inc, \cG}(v)$ are defined analogously.

In the body of the paper, we will present results about GNNs that work over directed graphs. Some of our results will make heavy use of directionality, considering GNNs that only look at out-neighbors. But many results will deal with ``bidirectional GNNs'', those that can see both out- and in-neighbors. \emph{In the appendix
we show that the results which will deal with  bidirectional GNNs also apply to undirected graphs}.

\myparagraph{Graph neural networks}
We use a standard notion of ``aggregate-combine'' graph neural networks with rational coefficients. The only distinction from the usual presentation is that
we allow GNNs to work over directed graphs, with separate aggregations over incoming and outgoing edges, while traditional GNNs work on undirected graphs.

\begin{definition}\label{def:gnn}
    An \emph{$n$-graph neural network} (GNN) is a tuple
    \begin{equation*}
        \tuple{
            \set{\gnndim{\ell}}_{\ell \in \intinterval{0}{L}},
            \set{\act{\ell}}_{\ell \in \intsinterval{L}}, 
            \set{\coefC{\ell}}_{\ell \in \intsinterval{L}},
            \set{\coefA{\ell}{x}}_{\substack{
                \ell \in \intsinterval{L} \\
                x \in \set{\out, \inc}}},
            \set{\coefR{\ell}}_{\ell \in \intsinterval{L}},
            \set{\coefb{\ell}}_{\ell \in \intsinterval{L}}
        },
    \end{equation*}
    where $L \in \bbN^+$ is the number of layers;
    each ${\gnndim{\ell} \in \bbN^+}$, called the \emph{dimension} of the $\ell^{th}$ layer, requiring $\gnndim{0} := n$, the number of colors;
    each ${\act{\ell}: \bbQ \to \bbQ}$, the \emph{activation function} of the $\ell^{th}$ layer;
    each ${\coefC{\ell}, \coefA{\ell}{x}, \coefR{\ell} \in \bbQ^{\gnndim{\ell} \times \gnndim{\ell-1}}}$, the \emph{coefficient matrices} of the $\ell^{th}$ layer; and
    $\coefb{\ell} \in \bbQ^{\gnndim{\ell}}$, the \emph{bias vector} of the $\ell^{th}$ layer.
\end{definition}
%\ch{add FNN}

All the coefficients are rational. In order to have an effective representation of a GNN, we will also \emph{assume that the activation functions are computable}. 

\begin{definition}
    For an $n$-GNN $\cA$ and an $n$-graph $\cG$,
    the \emph{computation of $\cA$ on $\cG$} is a sequence of \emph{derived feature functions} $\set{\feat{\ell}_\cG: V \to \bbQ^{\gnndim{\ell}}}_{\ell \in \intinterval{0}{L}}$ defined inductively.
    For $v \in V$, $\feat{\ell}_\cG(v)$ is called the \emph{$\ell$-feature (vector) of $v$}, and 
    $\feat{\ell}_{\cG, i}(v)$ denotes the $i^{th}$ entry of $\feat{\ell}_\cG(v)$.

    For $\ell = 0$, for $1 \le i \le \gnndim{0}$,
    if $v \in U_i$, then $\feat{0}_{\cG, i}(v) = 1$; otherwise, $\feat{0}_{\cG, i}(v) = 0$.
    For $1 \le \ell \le L$,
    \begin{equation*}
        \feat{\ell}_\cG(v) := \actp{\ell}{
            \coefC{\ell} \feat{\ell-1}_\cG(v) + \!\!\!\!\!\!
            \sum_{x \in\set{\out, \inc}}\!\!
            \left(
                \coefA{\ell}{x} \!\!\!\!\!\!
                \sum_{u \in \nbr{x, \cG}(v)} \!\!\!\!\!
                \feat{\ell-1}_\cG(u)
            \right) +
            \coefR{\ell} \sum_{u \in V} \feat{\ell-1}_\cG(u) +
            \coefb{\ell}
        }.
    \end{equation*}
\end{definition}

That is, we compute the feature values of a node $v$ at layer $\ell+1$ by adding several components. One component aggregates over the $\ell$-layer feature vector from the outgoing neighbors of  $v$, and applies a affine transformation. Another component does the same for the incoming neighbors of $v$, a third does this for every node in the graph, while another applies a affine transformation to the $\ell$-layer feature vector of $v$ itself. The affine transformation can be different for each component, and in particular can be a zero matrix that just drops that component. The final component of the sum is the bias vector.

When the graph $\cG$ is clear from the context,
we omit it and simply write $\feat{\ell}(v)$ and $\feat{\ell}_{i}(v)$, and similarly
when the graph $\cG$ is clear from the context,
write $\nbr{\out}(v)$ and $\nbr{\inc}(v)$ for the in-neighbors and out-neighbors.

Note that in most presentations of GNNs, one deals with only undirected edges. The above definition degenerates in that setting to two aggregations per layer, with the aggregation over all nodes 
often referred to in the literature as the \emph{global readout}.

In some presentations of GNNs, a \emph{classification function}, which associates a final Boolean decision to a node, is included in the definition. In our case, we have separated out the classification function as an independent component in defining the expressiveness: see the last part of the preliminaries.

\myparagraph{Classes of activation functions}
Following prior work on analysis of GNNs, some of our results will deal with activation functions that are bounded in value:

\begin{definition}
    We say that the function $f: \bbQ \to \bbQ$
    is \emph{eventually constant},
    if there exists $t_\leftt, t_\rightt \in \bbQ$
    satisfying $t_\leftt < t_\rightt$,
    called the \emph{left and right thresholds of $f$},
    such that
    for every $x \le t_\leftt$, $f(x) = f(t_\leftt)$;
    for every $x \ge t_\rightt$, $f(x) = f(t_\rightt)$.
\end{definition}

A standard eventually constant function is the \emph{truncated $\relu$ function},
denoted by $\trrelu$,
which is $0$ for negatives, $1$ for $x$ greater than $1$, and $x$ otherwise~\cite{barceloetallogical}.
There are other eventually constant functions that are used in practice: for example, the linear approximation of standard bounded functions used in graph learning, like the Sigmoid activation function. %\michael{add citation}. %he \emph{clipped $\relu$} is a variation of this where $1$ is replaced by an arbitrary rational. 
We will be interested in functions that are defined on the reals, but which preserve the rationals. The definition of eventually constant extends to such a function in the obvious way.

For a GNN with eventually constant activation functions,
we use $\tl{\ell}$ and $\tr{\ell}$ 
to denote the left and right thresholds of the GNN's activation functions.

We also consider \emph{unbounded activation functions}, such as the \emph{standard $\relu$ function}, which is $x$ for non-negatives and $0$ for negatives.

\myparagraph{Flavors of GNN}
\begin{definition}
    For a GNN $\cA$,
    we say that
    $\cA$ is \emph{outgoing-only},
    if for every $1 \le \ell \le L$, $\coefA{\ell}{\inc}$ is a zero matrix.
    % $\cA$ is \emph{bidirectional}, denoted by $\cB$,
    % if there is no restriction on $\coefA{\ell}{\inc}$.
    We say that
    $\cA$ is \emph{local},
    if for every $1 \le \ell \le L$, $\coefR{\ell}$ is a zero matrix.
    % In the usual GNN terminology, this would mean that \emph{there is no global readout}.
    % $\cA$ is \emph{global}, denoted by $\cG$,
    % if global readout is allowed.
    % $\cA$ is \emph{eventually constant}, denoted by $\cC$,
    % if for $1 \le \ell \le L$, $\act{\ell}$ is an eventually constant function.
\end{definition}

% Our results outside of the eventually constant activation case will deal with either: \emph{piecewise linear} activations, denoted $\pw$,
% truncated $\relu$ activations, denoted $\trrelu$,
% or standard $\relu$ activations.
We use the following naming,
${(\cO)(\cL)(\cC|\pw|\trrelu|\relu)\text{-}\mathsf{GNN}}$,
for the set of GNNs satisfying constraints given by the prefix,
where $\cO$ for outgoing-only;
$\cL$ for local;
$\cC$ for eventually constant activation;
$\pw$ for $\emph{piecewise linear}$ activation.
For example, $\olcGNN$ is the set of outgoing-only and local GNNs with eventually constant activations;
$\bgpwGNN$ is the set of GNNs with piecewise linear activations.

\myparagraph{Classifiers and Boolean semantics}
Our GNNs define vector-valued classification functions on nodes. But for comparing with expressiveness and in defining verification problems, we will often use a derived function from nodes to Booleans.
We do this by thresholding at the end -- below we use $1/2$ for convenience, but other choices do not impact the results.
\begin{definition}
    For an $n$-GNN $\cA$,
    an $n$-graph $\cG$,
    and a vertex $v$ in $\cG$,
    we say that $\cA$ \emph{accepts} the tuple $\tuple{\cG, v}$,
    if $\feat{L}_{\cG, 1}(v) \ge 1/2$.
\end{definition}

Note that the global readout component 
can interact with the activation functions $\act{\ell}$, which can behave very differently on translated values due to non-linearity -- think of a typical $\act{\ell}$ as a piece-wise linear function.
Global readout can also interact with the classification threshold,  pushing some values above the threshold while leaving others below.

\myparagraph{Satisfaction of logical formulas}
We will deal with various classes of logical formulas on graphs,  $\phi(v_1 \ldots v_n)$. These will be evaluated on a graph with respect to a graph $G$ and a variable binding $\sigma$ that maps each free variable to a node of
$G$. We often write $G, \sigma \models \phi$, or $G, c_1 \ldots c_n$ for nodes $c_i$, where the ordering of nodes is understood. For our formulas, we will always have $n$ at most $2$.
We will often omit the input graph since it will be obvious from context.
%\michael{Currently we do this in Section 6}

We will also deal logical formulas
over integers $\Psi(x_1, \ldots, x_n)$, %\michael{stick with some convention for integer variables vs node variables} 
which are evaluated with respect to a variable binding $\sigma$ mapping free variables to integers. We write $\Psi(\sigma)$ or $\Phi(a_1, \ldots, a_n)$ to denote satisfaction of such a $\Psi$ on binding $\sigma$
mapping $x_i$ to $a_i$. We sometimes identify the value true with $1$ and false with $0$, writing $\eval{\Psi(\bfa)}=1$.

We will use  lowercase Greek letters to denote logical formulas over graphs,
and uppercase Greek letters for logical formulas over integers.

%\michael{Still todo:
%It seems like in Section 6 we say things like `` $\phi(\eta)$ holds'' where
%$\eta$ is a (feature) vector and $\phi$ is in Presburger Arithmetic. It is %worthwhile explaining that this means we choose a canonical ordering of the variables etc. }

\myparagraph{Two-variable modal logic with Presburger quantifiers}
We review logic with Presburger quantifiers. The basic idea is to combine a decidable logic on uninterpreted structures, like two-variable logic or guarded logic, with the ability to perform some arithmetic on the number of elements. There are several formalisms in the literature that combine Presburger arithmetic with a decidable uninterpreted logic, some originating many years ago~\cite{bapa}. We will rely on a recent logic from~\cite{localpresburgerbartosztony}, but we will need several variations of the underlying idea here.

\begin{definition}\label{def:presburger}
    A \emph{Presburger quantifier} is of the form:
    \begin{equation*}
        \cP(x)\ :=\ 
        \sum_{t \in \intsinterval{k}} \lambda_t \cdot \presby{\varphi_t(x,y)}\ \circledast\ \delta,
    \end{equation*}
    where $\delta \in \bbZ$;
    each $\lambda_t \in \bbZ$;
    each $\varphi_t(x,y)$ is a formula with free variables $x$ and $y$;
    $\circledast$ is one of $=$, $\neq$, $\leq$, $\geq$, $<$, or $>$.
    Note that $\cP(x)$ has one free variable $x$.

    We give the semantics of these quantifiers inductively,
    assuming a semantics for $\varphi_t(x, y)$.
    Given a graph $\cG$ and a vertex $v$ in $\cG$,
    we say that $\cP(x)$ holds in $\cG, x/v$,
    denoted by ${\cG \models \cP(v)}$,
    if the following (in)equality holds in $\bbZ$:
    \begin{equation*}
        \sum_{t \in \intsinterval{k}} \lambda_t \cdot
        \abs{\setc{u \in V}{\cG \models \varphi_t(v, u)}}
        \ \circledast\ \delta.
    \end{equation*}
\end{definition}

\begin{remark}
    Note that each Presburger quantifier
    can be rewritten as a Boolean combination of expressions which \emph{only use the inequality symbol $\ge$ as $\circledast$}. 
    For example,
    ${\left(\presby{\varphi(x,y)} = \delta\right)}$ and
    ${\left(\presby{\varphi(x,y)} \ge \delta\right)} \land
    \neg {\left(\presby{\varphi(x,y)} \ge \delta + 1\right)}$
    are semantically equivalent.
    Therefore it is sufficient to consider Presburger quantifiers which only use the inequality symbol $\ge$.
\end{remark}

\begin{remark}
    We will make use of Presburger quantifiers
    that allow for rational coefficients of the form:
    \begin{equation*}
        \widetilde{\cP}(x)\ :=\  
        \lambda_0 + \sum_{t \in \intsinterval{k}} \lambda_t \cdot \presby{\varphi_t(x,y)}
        \ \circledast\ 
        \lambda'_0 + \sum_{t \in \intsinterval{k'}} \lambda'_t \cdot \presby{\varphi'_t(x,y)},
    \end{equation*}
    where each $\lambda_t,\lambda'_t\in \bbQ$.
    This is a shorthand for the Presburger quantifier:
    \begin{equation*}
        \cP(x)\ :=\ 
        \sum_{t \in \intsinterval{k}} (D\lambda_t) \cdot \presby{\varphi_t(x,y)}
        + 
        \sum_{t \in \intsinterval{k'}} (-D\lambda'_t) \cdot \presby{\varphi'_t(x,y)}
        \ \circledast \
        D(\lambda_0-\lambda'_0),
    \end{equation*}
    where $D$ is the least common multiplier %of the absolute value
    of the denominators of the coefficients in $\widetilde{\cP}(x)$.

    When $\lambda_t = 1$, we omit it and simply write $\presby{\varphi_t(x,y)}$.
\end{remark}

%\ch{TODO: rename MP2}
\begin{definition}\label{def:mp2}
We give the syntax of \emph{two-variable modal logic with Presburger quantifiers} ($\MPtwo$) over vocabulary $\tau$. Formulas will have exactly one free variable, denoted $x$ below.
    \begin{itemize}
        \item $\top$ is an $\MPtwo$ formula.
        \item For a unary predicate $U \in \tau$, $U(x)$ is an $\MPtwo$ formula. %not formulae -- avoid using this word since 
        \item If $\varphi(x)$ is an $\MPtwo$ formula, then so is $\neg\varphi(x)$.
        \item If $\varphi_1(x)$ and $\varphi_2(x)$ are $\MPtwo$ formulas, then so is $\varphi_1(x) \land \varphi_2(x)$.
        \item If $\set{\varphi_t(x)}_{t \in \intsinterval{k}}$ is a set of $\MPtwo$ formulas and
        $\set{\epsilon_t(x, y)}_{t \in \intsinterval{k}}$ is a set of \emph{guard atoms},
        each of form $E(x, y)$, $E(y, x)$, or $\top$, then 
        \begin{equation*}
            \sum_{t \in \intsinterval{k}} \lambda_t \cdot \presby{\epsilon_t(x,y)\land\varphi_t(y)}
            \ \circledast\ \delta
        \end{equation*}
        is also an $\MPtwo$ formula.
        $\set{\epsilon_t(x, y)}_{t \in \intsinterval{k}}$ are the \emph{guards} of the formula.  Consistent with the restriction we announced on the logic, we consider the result as a formula with free variable $x$: if all $\epsilon_i$ are $\top$ it returns either every vertex or no vertex.
    \end{itemize}
    The semantics of the Boolean connectives is as usual, while the semantics of the Presburger quantifiers is given by Definition~\ref{def:presburger}.
\end{definition}
An $\MPtwo$ formula $\varphi(x)$ is an $n$-formula if its vocabulary consists of $n$ unary predicates.
We use abbreviations $\lor$ and $\to$ as usual.
%\michael{Do we need $n$-formulas or even ($n$) GNNs?}
%\ch{yes, we need it to define equivalence}
%michael: ok, although maybe we should say that we omit (n) when it is clear from context.
Note that the guarded universal quantifier
$\forall y\ E(x, y) \to \varphi(y)$ can be expressed as
$\left(\presby{E(x, y) \land \neg\varphi(y)} = 0\right)$,
and the guarded existential quantifier
$\exists y\ E(x, y) \land \varphi(y)$ can be expressed as
$\left(\presby{E(x, y) \land \varphi(y)} \ge 1\right)$.

\begin{definition}
    For an $\MPtwo$ formula $\varphi(x)$,
    for a graph $\cG$ and vertex $v$ in $\cG$,
    we define the valuation of $\varphi(x)$ on $\cG, v$, denoted by $\eval{\varphi(v)}_\cG$ as follows:
    \begin{equation*}
        \eval{\varphi(v)}_\cG\ :=\ 
        \begin{cases}
            1, &\text{if $\cG \models \varphi(v)$} \\
            0, &\text{otherwise}.
        \end{cases}
    \end{equation*}
\end{definition}
When the graph $\cG$ is clear from the context,
we omit it and simply write $\eval{\varphi(v)}$.

The logic ``\emph{local} $\MPtwo$'' is obtained by excluding $\top$ as a guard.
We say a $\MPtwo$ formula is \emph{outgoing-only} if there are no guarded atom of form $E(y, x)$.
Analogously to what we did with GNNs, we use $\cL$ and $\cO$ to indicate the ``local'' and ''outgoing-only'' fragment of $\MPtwo$, respectively.

The logic $\MPtwo$ combines Presburger arithmetic and quantification over the model.
Thus one might worry that it has an undecidable satisfiability problem. And indeed, we will show this: see Theorem~\ref{thm:global_mptwo_undecidable}.
An idea to gain decidability is to impose that the quantification is \emph{guarded} -- again, the underlying idea is from~\cite{localpresburgerbartosztony}.

The logic $\blMPtwo$ is contained in the following logic, defined in~\cite{localpresburgerbartosztony}:
\begin{definition}
    The syntax of the \emph{guarded fragment of two-variable logic with Presburger quantifiers} ($\GPtwo$) over colored graph vocabulary $\tau$ starts with arbitrary atoms over the vocabulary, with the usual connective closure and the following rules for quantifiers.
    \begin{itemize}
        \item If $\varphi(x)$ is a $\GPtwo$ formula,
        then so are $\forall x\ \epsilon(x) \to \varphi(x)$ and $\exists x\ \epsilon(x) \land \varphi(x)$, where $\epsilon$ is either $U(x)$ for some unary predicate $U \in \tau$ or $x=x$.
        \item If $\varphi(x, y)$ is a $\GPtwo$ formula,
        then so are $\forall x\ \epsilon(x, y) \to \varphi(x, y)$ and $\exists x\ \epsilon(x, y) \land \varphi(x, y)$, where $\epsilon(x, y)$ is either $E(x, y)$ or $E(y, x)$.
        \item If $\set{\varphi_t(x, y)}_{t \in \intsinterval{k}}$ is a set of $\GPtwo$ formulas and
        $\set{\epsilon_t(x, y)}_{t \in \intsinterval{k}}$ is a set of atoms,
        each of form $E(x, y)$ or $E(y, x)$,
        then 
        \begin{equation*}
            \sum_{t \in \intsinterval{k}} \lambda_t \cdot \presby{\epsilon_t(x,y)\land\varphi_t(x, y)}
            \ \circledast\ \delta
        \end{equation*}
        is also a $\GPtwo$ formula.
    \end{itemize}
\end{definition}
The main difference between the logic $\blMPtwo$ and the logic above is that the former is ``modal'', restricting to one-variable formulas, and allowing two variables only in the guards. While in the logic above we can build up more interesting two variable formulas, for example conjoining two guards. 

We will make use of the following prior decidability result:

\begin{theorem}[\cite{localpresburgerbartosztony}, Theorem 10] \label{thm:gp2_decidabel}
    The finite satisfiability problem of $\GPtwo$ is decidable.
\end{theorem}

From this we easily derive the decidability of $\blMPtwo$:
\begin{corollary}\label{corollary:mp2_decidable}
    The finite satisfiability problem of $\blMPtwo$ is decidable.
\end{corollary}

\begin{proof}
    Let $\varphi(x)$ be an $n$-$\blMPtwo$ formula and $U_{n+1}$ be a fresh unary predicate.
    We claim that $\varphi(x)$ is finitely satisfiable if and only if the $\GPtwo$ sentence
    $\psi := \exists x\ U_{n+1}(x) \land \varphi(x)$ is also finitely satisfiable.
    Then the corollary follows from the decidability of the finite satisfiability problem of $\GPtwo$ by Theorem~\ref{thm:gp2_decidabel}.
    
    If $\varphi(x)$ is finitely satisfiable by the $n$-graph $\cG$ and vertex $v \in G$,
    let $\cG'$ be the $(n+1)$-graph that extended $\cG$ with $U_{n+1} := \set{v}$.
    Then $\cG' \models U_{n+1}(v)$,
    which implies that $\cG' \models \psi$. 
    Hence $\psi$ is finitely satisfiable by $\cG'$.
    
    If $\psi$ is finitely satisfiable by the $(n+1)$-graph $\cG$,
    let $\cG'$ be the $n$-graph that restricted $\cG$ by removing $U_{n+1}$.
    By definition, there exists at least one vertex $v$ in $\cG$ such that $\cG \models U_{n+1}(v) \land \varphi(v)$,
    which implies that $\cG \models \varphi(v)$.
    Since there is no $U_{n+1}$ in $\varphi(x)$,
    $\cG' \models \varphi(v)$.
    Hence $\varphi(x)$ is finitely satisfiable by $\cG'$.
\end{proof}

\myparagraph{Notions of expressiveness for GNNs and $\MPtwo$ formulas}
Recalling that we have a node-to-Boolean semantics available for both logical formulas and GNNs (via thresholding),
we use the term \emph{$n$-specification} for either an $n$-GNN or an $n$-$\MPtwo$ formula. 

\begin{definition}
    If $S_1, S_2$ are $n$-GNNs, they are said to be \emph{equivalent} if they accept the same nodes within $n$-graphs.
    If $S_1$ is an $n$-GNN and $S_2$ a node formula in some logic,
    we say $S_1$ and $S_2$ are equivalent
    if for every $n$-graph $\cG$ and vertex $v$ in $\cG$,
    $S_1$ accepts $\tuple{\cG, v}$ if and only if $\cG, v$ satisfies $S_2$.
\end{definition}

The notions of two languages of specifications being equally expressive, or equally expressive over undirected graphs, is defined in the obvious way.

\myparagraph{Verification problems for GNNs}
We focus on two verification problems. 
The first is the most obvious analog of satisfiability for GNNs, whether it accepts some node of some graph:

\begin{definition}
    For an $n$-GNN $\cA$,
    we say that $\cA$ is \emph{satisfiable},
    if there exist an $n$-graph $\cG$ and a vertex $v$ in $\cG$, such that $\cA$ accepts $\tuple{\cG, v}$.
\end{definition}

We will also consider a variation of the problem which asks whether a GNN accepts every node of some graph:
\begin{definition}
    For an $n$-GNN $\cA$,
    we say that $\cA$ is \emph{universally satisfiable}, %michae: universal satisfiability is the noun, universally satisfiable is the adjective
    if there exist an $n$-graph $\cG$, such that
    for every vertex $v$ in $\cG$,
    $\cA$ accepts $\tuple{\cG, v}$.
\end{definition}

Two GNNs are equivalent if they accept the same tuples. Note that, like satisfiability and unlike universal satisfiability, this does not require a quantifier alternation. For brevity we will not state results for equivalence, but \emph{it can easily be seen that both our positive and negative results on satisfiability also apply to equivalence}. 

\subsection{Summary of results} \label{subsec:summary}
We summarize our analysis of the complexity and decidability of verification problems for GNNs in Table \ref{tab:sat} and \ref{tab:univsat}. 
In cases where we have only written ``decidable', we have not computed precise bounds: and in the ``eventually constant'' case we would require more hypotheses in order to compute such bounds.

\begin{table}[h!]
    \tbl{Complexity and decidability results for the satisfiability problem for fragments of GNNs.}
    {\begin{tblr}{
        colspec={l|c|c|c|c},
        hline{2, 3, 4, 5, 6},
        cell{1}{1} = {r=2}{m},
        cell{1}{2} = {c=2}{c},
        cell{1}{4} = {c=2}{c},
        cell{3}{2} = {c=2}{c},
        cell{3}{4} = {r=2}{m},
        cell{3}{5} = {r=4}{m},
        cell{4}{2} = {c=2}{c},
        cell{5}{2} = {r=2}{m},
        cell{5}{3} = {r=2}{m},
        cell{5}{4} = {r=2}{m},
        }
        & Local ($\cL$) & & Global \\
        & Outgoing-only ($\cO$) & Bidirectional & Outgoing-only ($\cO$) & Bidirectional \\
        $\trrelu$ &
        \makecell[c]{$\pspace$-complete \\(Theorem~\ref{thm:pspace})} & &
        Open &
        \makecell[c]{Undecidable \\(Theorem~\ref{thm:global_gnn_undecidable})} \\
        Eventually constant ($\cC$) &
        \makecell[c]{Decidable \\(Theorem~\ref{thm:local_gnn_decidable})} \\
        $\relu$ &
        \makecell[c]{$\nexp$-complete \\(Theorem~\ref{thm:outputonlydecidability})} &
        \makecell[c]{Undecidable \\(Theorem~\ref{thm:hilbert_to_blrelugnn})} &
        \makecell[c]{Undecidable \\(Theorem~\ref{thm:hilbert_to_ogrelugnn})} \\
        Piecewise linear ($\pw$)
    \end{tblr}} \label{tab:sat}
\end{table}

\begin{table}[h!]
    \tbl{Decidability results for the universal satisfiability problem for fragments of GNNs.}
    {\begin{tblr}{
        colspec={l|c|c|c|c},
        hline{2, 3, 4, 5, 6},
        cell{1}{1} = {r=2}{m},
        cell{1}{2} = {c=2}{c},
        cell{1}{4} = {c=2}{c},
        cell{3}{2} = {r=2, c=2}{c},
        cell{3}{4} = {r=2}{m},
        cell{3}{5} = {r=4}{m},
        cell{5}{2} = {r=2}{m},
        cell{5}{3} = {r=2}{m},
        cell{5}{4} = {r=2}{m},
        }
        & Local ($\cL$) & & Global \\
        & Outgoing-only ($\cO$) & Bidirectional & Outgoing-only ($\cO$) & Bidirectional \\
        $\trrelu$ &
        \makecell[c]{Decidable \\(Theorem~\ref{thm:local_gnn_universal_decidable})} & &
        Open &
        \makecell[c]{Undecidable \\(Theorem~\ref{thm:global_gnn_universal_undecidable})} \\
        Eventually constant ($\cC$) \\
        $\relu$ &
        Open &
        \makecell[c]{Undecidable \\(Theorem~\ref{thm:blgnn_unbounded_undec})} &
        \makecell[c]{Undecidable \\(Theorem~\ref{thm:oggnn_unbounded_undec})} \\
        Piecewise linear ($\pw$)
    \end{tblr}} \label{tab:univsat}
\end{table}

%\michael{TODO: add pointers to theorems}
%\ch{done}
%michael: thanks

\section{Characterization and decidability of GNNs with eventually constant activation functions} \label{sec:eventually_constant}

In this section, we only consider GNNs with eventually constant activations.
In Section~\ref{subsec:spectrum}, we establish a key tool to analyzing these GNNs: we show that the set of possible activation values is finite, and one can compute an overapproximation of this set.
We use this for two purposes. First we give a decidability result for GNNs with eventually constant activations and only local aggregation, and then we show that even with global aggregation we get an equivalence of the GNNs in expressiveness with a logic.

In Section~\ref{subsec:global_undecidable}, we show that the finite satisfiability of $\bgMPtwo$ is undecidable. Using the expressiveness characterization, this will imply that satisfiability problems for global GNNs are undecidable. These results were presented for GNNs and logics on directed graphs.
% In Section~\ref{subsec:undirected} we use the logical characterizations to show that they also apply to the standard setting for GNNs of undirected graphs.

\subsection{Decidability of satisfiability problems for GNNs with eventually constant functions, via logic}\label{subsec:spectrum}
%\ch{change title of subsection}

We now come to one of the crucial definitions in the paper, the spectrum of a GNN.

\begin{definition}
    For an $n$-GNN $\cA$ and $0 \le \ell \le L$,
    the \emph{$\ell$-spectrum} of $\cA$, denoted by $\spectrum{\ell}$,
    is the set $\setc{\feat{\ell}(v)}{\text{for every $n$-graph $\cG$ and vertex $v$ in $\cG$}}$.
\end{definition}

That is, the $\ell$-spectrum is the range of the features computed at layer $\ell$, as we range over all input graphs and vertices.
We show that, for a GNN with eventually constant activations, the spectrum is actually finite, and a finite superset is computable:

\begin{theorem} \label{thm:computespectrum}
    For every $\bgcGNN$ $\cA$ and $0 \le \ell \le L$, the $\ell$-spectrum of $\cA$ is finite. 
    We can compute a finite superset of the $\ell$-spectrum from the specification of $\cA$.
\end{theorem}

We give some intuition for the proof. Our effective overapproximation of the spectrum will simulate the computation of the GNN, and will be defined inductively on the layers.
Recall that a GNN is given by
dimensions $\set{\gnndim{\ell}}_{\ell \in \intinterval{0}{L}}$,
activation functions $\set{\act{\ell}}_{\ell \in \intsinterval{L}}$,
coefficient matrices for transforming the prior node value $\set{\coefC{\ell}}_{\ell \in \intsinterval{L}}$,
coefficient matrices for local aggregation $\set{\coefA{\ell}{x}}_{\substack{
            {\ell \in \intsinterval{L}} \\
            x \in \set{\out, \inc}}}$,
coefficient matrices for global readout $\set{\coefR{\ell}}_{\ell \in \intsinterval{L}}$,
and bias vectors $\set{\coefb{\ell}}_{\ell \in \intsinterval{L}}$.

We now define our overapproximations inductively as follows:
\begin{definition}
    For every $\bgcGNN$ $\cA$, for $0 \le \ell \le L$, the set $\spectrumover{\ell}$ is defined as follows:
    \begin{equation*}
        \begin{aligned}
            \spectrumover{0}\ :=\ &\set{0, 1}^{\gnndim{0}} \\
            \spectrumover{\ell}\ :=\ &
            \setc{
                \actp{\ell}{
                    \dfrac{\bfk}{\capa{\ell}}
                }
            }{
                \bfk \in \intinterval{\capa{\ell} \cdot \tl{\ell}}{\capa{\ell} \cdot \tr{\ell}}^{\gnndim{\ell}}
            },
        \end{aligned}
    \end{equation*}
    where
    $\capa{\ell}$ is the product of the least common denominator of elements in $\spectrumover{\ell-1}$ and the least common denominator of coefficients in $\coefC{\ell}$, $\coefA{\ell}{\out}$, $\coefA{\ell}\inc$, $\coefR{\ell}$, and $\coefb{\ell}$. 
    %\michael{maybe easier just to take the least common denominator of the union of these sets}.
    %\ch{it is the product of LCD of two sets, not LCD of union.}
\end{definition}

We show that the set $\spectrumover{\ell}$ overapproximates the $\ell$-spectrum:

\begin{lemma}
\label{lemma:syntacticoverapprox}
    For every $n$-$\bgcGNN$ $\cA$ and $0 \le \ell \le L$,
    for every $n$-graph $\cG$ and vertex $v$ in $\cG$,
    $\feat{\ell}(v) \in \spectrumover{\ell}$.
\end{lemma}

It is quite straightforward to see that every element of the spectrum is captured.
It is an overapproximation because different integers that we sum in an inductive step
may not be realized in the same graph.

\begin{proof}
    The proof is by induction on layers.
    The base case $\ell = 0$ is straightforward.
    For the inductive step $1 \le \ell \le L$,
    for every vertex $v$ in $\cG$, 
    by the induction hypothesis,
    there exists $\bfs(v) \in \spectrumover{\ell-1}$ such that
    ${\feat{\ell-1}(v) = \bfs(v)}$. 
    Let $\bfk(v)$ be defined as:
    \begin{equation*}
        \bfk(v)\ :=\ 
        \capa{\ell} \cdot \left(
            \coefC{\ell} \bfs(v) +  
            \sum_{x \in\set{\out, \inc}}
            \left(
                \coefA{\ell}{x} \sum_{u \in \nbr{x}(v)} \bfs(u)
            \right) +
            \coefR{\ell} \sum_{u \in V} \bfs(u) +
            \coefb{\ell}
        \right).
    \end{equation*}
    It is clear that $\feat{\ell}(v)$ can be rewritten as
    $
        \feat{\ell}(v) = 
        \actp{\ell}{
            \dfrac{\bfk(v)}{\capa{\ell}}
        }
    $.
    Recall that $\capa{\ell}$ is the product of the least common denominator of elements in
    $\spectrumover{\ell-1}$ and the least common denominator of coefficients in $\coefC{\ell}$, $\coefA{\ell}{\out}$, $\coefA{\ell}\inc$, $\coefR{\ell}$, and $\coefb{\ell}$. 
    Therefore, $\bfk(v)$ is an integer vector.
    
    % For $1 \le i \le \gnndim{\ell}$, let $\sigma$ be the $i^{th}$ entry of $\bfk(v)$.
    % If $\sigma \ge \capa{\ell} \cdot \tr{\ell}$, 
    % then $\feat{\ell}_i(v) = \actp{\ell}{\dfrac{\sigma}{\capa{\ell}}} = \actp{\ell}{\tr{\ell}}$
    % since $\act{\ell}$ is eventually constant with right threshold $\tr{\ell}$.
    % Analogously, if $\left(\bfk(v)\right)_i \le \capa{\ell} \cdot \tl{\ell}$, 
    % then $\feat{\ell}_i(v) = \actp{\ell}{\tl{\ell}}$.
    Since $\act{\ell}$ is eventually constant with right threshold $\tr{\ell}$.
    We can obtain another vector $\bfk'(v) \in \intinterval{\capa{\ell} \cdot \tl{\ell}}{\capa{\ell} \cdot \tr{\ell}}^{\gnndim{\ell}}$ satisfying that
    $
        \actp{\ell}{\dfrac{\bfk'(v)}{\capa{\ell}}} = \actp{\ell}{\dfrac{\bfk(v)}{\capa{\ell}}}
    $
    by substituting the elements which are greater than $\capa{\ell} \cdot \tr{\ell}$ by $\capa{\ell} \cdot \tr{\ell}$, and the elements that are less than $\capa{\ell} \cdot \tl{\ell}$ by $\capa{\ell} \cdot \tl{\ell}$.
    By the definition of $\spectrumover{\ell}$,
    $\actp{\ell}{\dfrac{\bfk'(v)}{\capa{\ell}}} \in \spectrumover{\ell}$,
    which implies that $\feat{\ell}(v) \in \spectrumover{\ell}$.
\end{proof}
 
We can show by induction on the number of the layers that the set is finite -- regardless of computability of the activation functions!

\begin{lemma}
    \label{lem:overapproxfinite}
    For every $n$-$\bgcGNN$ $\cA$ and $0 \le \ell \le L$,
    $\spectrumover{\ell}$ has finite size.
    $\spectrumover{\ell}$ can be computed if activations in $\cA$ are computable.   
\end{lemma}

In the inductive step, we have a finite set of rationals, thus some fixed precision. We take some integer linear combinations and we will obtain an infinite set of values, but only finitely many between the left and right thresholds of the eventually constant activations. Thus when we apply the activation functions to these values, we will get a finite set of rational values -- since the activation functions map rationals to rationals.

\begin{proof}
    We prove the lemma by induction on layers.
    The base case $\ell = 0$ is trivial.
    For the induction step $1 \le \ell \le L$,
    by the induction hypothesis, the size of the set $\spectrumover{\ell-1}$ is finite.
    Thus, the least common denominator of elements in $\spectrumover{\ell-1}$ is well-defined and finite,
    which implies that $\capa{\ell}$ is also finite.
    Because the size of the set $\spectrumover{\ell}$ is bounded by $\capa{\ell} \cdot \left(\tr{\ell}-\tl{\ell}\right) + 1$,
    the size of $\spectrumover{\ell}$ is finite.

    It is obvious that we can compute the set $\spectrumover{\ell}$ recursively, assuming the computability of activations in $\cA$.
\end{proof}

\begin{remark}
    The restriction to rational coefficients is crucial in the argument.
    Consider the following $1$-layer $1$-$\bltrreluGNN$.
    The dimensions are $\gnndim{0} = \gnndim{1} = 1$;
    the coefficient matrix $\coefC{1}$ is a zero matrix;
    $\left(\coefA{1}{\out}\right)_{1, 1} = \sqrt{2}$;
    $\left(\coefA{1}{\inc}\right)_{1, 1} = -1$;
    the bias vector $\coefb{1}$ is a zero vector.
    It is not difficult to see that its $1$-spectrum is 
    $\setc{\trrelup{\sqrt{2}t_1 - t_2}}{t_1, t_2 \in \bbN}$,
    whose size is infinite since $\sqrt{2}$ is irrational.
\end{remark}

\begin{remark}
    Even simple GNNs may have exponential size spectra. 
    For example, let $\cA_k$ be a $1$-layer $1$-$\oltrreluGNN$ defined as follows:
    the dimensions are $\gnndim{0} = \gnndim{1} = 1$;
    the coefficient matrix $\coefC{1}$ is a zero matrix;
    $\left(\coefA{1}{\out}\right)_{1,1} = 1/k$;
    the bias vector $\coefb{1}$ is a zero vector.
    By definition, its $1$-spectrum  is $\setc{t/k}{t \in \intinterval{0}{k}}$, whose size is $k+1$. But the description of $\cA_k$ is only linear in $\log k$.
\end{remark}

We now give several applications of the spectrum result.
First we can use the finiteness of the spectrum to get a characterization of the expressiveness of $\bgcGNN$ and logic:

\begin{theorem}\label{thm:gnn_to_logic}
    For every 
    $n$-$\bgcGNN$ $\cA$,
    there exists an
    $n$-$\MPtwo$ formula $\psi_\cA(x)$, effectively computable from the description of $\cA$,
    such that $\cA$ and $\psi_\cA(x)$ are equivalent.
    In the case we start with $n$-$\blcGNN$ and $n$-$\ogcGNN$, the formulas we obtain are in $n$-$\blMPtwo$ and $n$-$\ogMPtwo$, respectively.
\end{theorem}

This expressiveness equivalence will be useful in getting further decidability results, as well as separations in expressiveness,
for GNNs.

The idea of the proof of the theorem is that we have only finitely many elements in the overapproximation set to worry about, so we can fix each in turn and write a formula for each. 
Our main translation is captured in the following:

\begin{lemma}\label{lem:gnn_to_logic}
    For every $n$-$\bgcGNN$ $\cA$, $0 \le \ell \le L$,
    and $\bfs \in \spectrumover{\ell}$,
    there exists an $n$-$\bgMPtwo$ formula %michael not formulae 
    $\varphil{\ell}_\bfs(x)$,
    such that for every $n$-graph $\cG$ and vertex $v$ in $\cG$,
    $\cG \models \varphil{\ell}_\bfs(v)$ if and only if 
    $\feat{\ell}(v) = \bfs$.
    In the case we start with $n$-$\blcGNN$ and $n$-$\ogcGNN$,
    the formulas we obtain are in $n$-$\blMPtwo$ and $n$-$\ogMPtwo$, respectively.
\end{lemma}

\begin{proof}
    We define an $n$-$\bgMPtwo$ formula $\varphil{\ell}_\bfs(x)$ inductively on layers.
    For the base case $\ell = 0$,
    for $1 \le i \le \gnndim{0}$,
    let $\theta_{i, 1}(x) := U_i(x)$ and $\theta_{i, 0}(x) := \neg U_i(x)$.
    % \begin{equation*}
    %     \psi_{i, c}(x)\ :=\ 
    %     \begin{cases}
    %         U_i(x), &\text{if $c = 1$} \\
    %         \neg U_i(x), &\text{if $c = 0$}.
    %     \end{cases}
    % \end{equation*}
    For $\bfs \in \spectrumover{0}$,
    \begin{equation*}
        \varphil{0}_\bfs(x)\ :=\ \bigwedge_{i \in \intsinterval{\gnndim{0}}} \theta_{i, \bfs_i}(x).
    \end{equation*}

    For the inductive case $1 \le \ell \le L$,
    for ${\capa{\ell} \cdot \tl{\ell}} \le k \le {\capa{\ell} \cdot  \tr{\ell}}$, $\bfs' \in \spectrumover{\ell-1}$, and $1 \le i \le \gnndim{\ell}$,
    \begin{equation*}
        \begin{aligned}
            \phil{\ell}_{k, \bfs', i}(x)\ :=\ 
                \left(\coefC{\ell} \bfs' + \coefb{\ell}\right)_i
                &+
                \sum_{\substack{
                    x \in \set{\out, \inc} \\
                    \bfs'' \in \spectrumover{\ell-1} }}
                    \left(\coefA{\ell}{x} \bfs''\right)_i \cdot
                \presby{\epsilon_{x}(x, y) \land \varphil{\ell-1}_{s''}(y)} \\
                &+
                \sum_{\bfs'' \in \spectrumover{\ell-1} }
                    \left(\coefR{\ell} \bfs''\right)_i\cdot
                \presby{\varphil{\ell-1}_{\bfs''}(y)}
                \ \circledast_{k}\ \dfrac{k}{\capa{\ell}},
        \end{aligned}
    \end{equation*}
    where ${\epsilon_{\out}(x, y) := E(x, y)}$ and ${\epsilon_{\inc}(x, y) := E(y, x)}$.
    If $k = \capa{\ell} \cdot  \tl{\ell}$, then $\circledast_{k}$ is $\le$; 
    if $k = \capa{\ell} \cdot  \tr{\ell}$, then $\circledast_{k}$ is $\ge$;
    otherwise, $\circledast_{k}$ is $=$.
    Finally, for $\bfs \in \spectrumover{\ell}$,
    \begin{equation*}
        \varphil{\ell}_{\bfs}(x) \ :=\ 
        \bigvee_{\bfs' \in \spectrumover{\ell-1}}
            \left(
                \varphil{\ell-1}_{\bfs'}(x) \land
                \bigwedge_{i \in \intsinterval{\gnndim{\ell}}}
                \bigvee_{k \in \cK(\bfs_i)}
                \phil{\ell}_{k, \bfs', i}(x)
            \right),
    \end{equation*}
    where $\cK(s) := \setc{k \in \intinterval{\capa{\ell} \cdot \tl{\ell}}{\capa{\ell} \cdot \tr{\ell}}} {\actp{\ell}{\dfrac{k}{\capa{\ell}}} = s}$.
    Note that by Lemma~\ref{lem:overapproxfinite},
    the size of $\spectrumover{\ell}$ is finite.
    Thus the disjunction in the construction is over a finite set.
    
    We prove the correctness of the construction by induction on the layers.
    
    \textbf{\underline{For the base case $\ell = 0$}},
    for every $n$-graph $\cG$, vertex $v$ in $\cG$,
    and $\bfs \in \spectrumover{0}$,
    it is straightforward to check that
    $\cG \models \varphil{0}_\bfs(v)$ if and only if 
    $\feat{0}(v) = \bfs$.

    \textbf{\underline{For the induction step $1 \le \ell \le L$}},    
    for every $n$-graph $\cG$ and
    $\bfs' \in \spectrumover{\ell-1}$,
    we define $\bfwl{\ell}_{\bfs'}: V \to \bbZ$ as follows:
    \begin{equation*}
        \begin{aligned}
            \bfwl{\ell}_{\bfs'}(v)\ :=\ 
                \left(\coefC{\ell} \bfs' + \coefb{\ell}\right)
                & +
                \sum_{\substack{
                    x \in \set{\out, \inc} \\
                    \bfs'' \in \spectrumover{\ell-1} }}
                \coefA{\ell}{x} \bfs''\cdot
                \abs{\setc{u \in V}{\cG \models \epsilon_{x}(v, u)\land\varphil{\ell-1}_{\bfs''}(u)}} \\
                & +
                \sum_{\bfs'' \in \spectrumover{\ell-1} }
                    \coefR{\ell} \bfs''\cdot
                    \abs{\setc{u \in V}{\cG \models \varphil{\ell-1}_{\bfs''}(u)}}.
        \end{aligned}
    \end{equation*}
    By the semantics of Presburger quantifiers,
    for $1 \le i \le \gnndim{\ell}$,
    $\cG \models \phil{\ell}_{k, \bfs', i}(v)$ if and only if 
    $\left(\bfwl{\ell}_{\bfs'}(v)\right)_i\ \circledast_{k}\ \dfrac{k}{\capa{\ell}}$. 
    Let ${V_{s''} := \setc{u \in V}{\cG \models \varphil{\ell-1}_{\bfs''}(u)}}$.
    By the induction hypothesis,
    for every vertex $u$ in $\cG$,
    $\cG \models \varphil{\ell-1}_{\bfs''}(u)$ if and only if $\feat{\ell-1}(u) = \bfs''$.
    Hence $V_{\bfs''} = \setc{u \in V}{\feat{\ell-1}(u) = \bfs''}$, and
    $\set{V_{\bfs''}}_{\bfs'' \in \spectrumover{\ell-1}}$ form a partition of $V$.
    Thus, we can rewrite $\bfwl{\ell}_{\bfs'}(v)$ as follows. 
    \begin{equation*}
        \begin{aligned}
            \bfwl{\ell}_{\bfs'}(v) \ =\ &
                \left(\coefC{\ell} \bfs' + \coefb{\ell}\right) +
                \sum_{\substack{
                    x \in \set{\out, \inc} \\
                    \bfs'' \in \spectrumover{\ell-1} }} \!\!\!\!
                \coefA{\ell}{x} 
                \left(\sum_{u \in \nbr{x}(v) \cap V_{\bfs''}} \bfs'' \right) +
                \sum_{\bfs'' \in \spectrumover{\ell-1}} \!\!\!\!
                \coefR{\ell}
                \left(\sum_{u \in V \cap V_{\bfs''}} \bfs''\right) \\
            \ =\ &
                \left(\coefC{\ell} \bfs' + \coefb{\ell}\right) +
                \sum_{x \in \set{\out, \inc} }
                \left(
                    \coefA{\ell}{x} 
                    \sum_{u \in \nbr{x}(v)} \feat{\ell-1}(u)
                \right) +
                \coefR{\ell} \sum_{u \in V} \feat{\ell-1}(u)
        \end{aligned}
    \end{equation*}
    By the definition of features,
    we obtain that $\feat{\ell}(v) = \actp{\ell}{\bfwl{\ell}_{\feat{\ell-1}(v)}(v)}$.

    If $\cG \models \varphil{\ell}_\bfs(v)$,
    then there exists $\bfs' \in \spectrumover{\ell-1}$ such that
    $\cG \models \varphil{\ell-1}_{\bfs'}(v)$.
    By the induction hypothesis, 
    $\feat{\ell-1}(v) = \bfs'$.
    For $1 \le i \le \gnndim{\ell}$, because
    $\cG \models \bigvee_{\cK(\bfs_i)} \phil{\ell}_{n, \bfs', i}(x)$,
    there exists ${\capa{\ell} \cdot \tl{\ell}} \le k \le {\capa{\ell} \cdot  \tr{\ell}}$
    such that $\actp{\ell}{\dfrac{k}{\capa{\ell}}} = \bfs_i$
    and 
    $\cG \models \phil{\ell}_{k, \bfs', i}(v)$.
    By the semantics of Presburger quantifiers,
    $\cG \models \phil{\ell}_{k, \bfs', i}(v)$ implies that
    $\left(\bfwl{\ell}_{\bfs'}(v)\right)_i\ \circledast_{k}\ \dfrac{k}{\capa{\ell}}$.
    Recall that $\act{\ell}$ is eventually constant with the left threshold $\tl{\ell}$ and reight threshold $\tr{\ell}$ and
    $\feat{\ell}_i(v) = \actp{\ell}{\left(\bfwl{\ell}_{\bfs'}(v)\right)_i}$.
    \begin{itemize}
        \item 
        If $k = \capa{\ell} \cdot \tl{\ell}$, then $\left(\bfwl{\ell}_{\bfs'}(v)\right)_i \le \tl{\ell}$.
        Thus, $\feat{\ell}_i(v)
        = \actp{\ell}{\tl{\ell}}
        = \actp{\ell}{\dfrac{k}{\capa{\ell}}}
        = \bfs_i$.
        \item
        If $k = \capa{\ell} \cdot \tr{\ell}$, then $\left(\bfwl{\ell}_{\bfs'}(v)\right)_i \ge \tr{\ell}$.
        Thus, $\feat{\ell}_i(v)
        = \actp{\ell}{\tr{\ell}}
        = \actp{\ell}{\dfrac{k}{\capa{\ell}}}
        = \bfs_i$.
        \item
        If $\capa{\ell} \cdot \tl{\ell} \le k \le \capa{\ell} \cdot \tr{\ell}$,
        then $\left(\bfwl{\ell}_{\bfs'}(v)\right)_i = \dfrac{k}{\capa{\ell}}$,
        which implies that
        $\feat{\ell}_i(v)
        = \actp{\ell}{\dfrac{k}{\capa{\ell}}}
        = \bfs_i$.
    \end{itemize}
    Therefore $\feat{\ell}(v) = \bfs$.

    On the other hand, 
    suppose that $\feat{\ell}(v) = \bfs$.
    Let us denote $\feat{\ell-1}(v)$ by $\bfs'$.
    By definition, the features are captured by spectrum, $\bfs \in \spectrumover{\ell}$ and $\bfs' \in \spectrumover{\ell-1}$.
    By the induction hypothesis, $\cG \models \varphil{\ell-1}_{\bfs'}(v)$.
    For $1 \le i \le \gnndim{\ell}$, let
    \begin{equation*}
        k\ :=\ 
        \begin{cases}
            \capa{\ell} \cdot \tl{\ell},&
                \text{if $\left(\bfwl{\ell}_{\bfs'}(v)\right)_i \le \tl{\ell}$} \\
            \capa{\ell} \cdot  \tr{\ell},&
                \text{if $\left(\bfwl{\ell}_{\bfs'}(v)\right)_i \ge \tr{\ell}$} \\
            \capa{\ell} \cdot \left(\bfwl{\ell}_{\bfs'}(v)\right)_i,&
                \text{otherwise}.
        \end{cases}        
    \end{equation*}
    By the semantics of Presburger quantifiers,
    $\cG \models \phil{\ell}_{k, \bfs', i}(v)$.
    Furthermore,
    $\actp{\ell}{\dfrac{k}{\capa{\ell}}}
    = \actp{\ell}{\bfwl{\ell}_{\bfs'}(v)}$
    By the property of $\bfwl{\ell}_{\bfs'}(v)$,
    $\actp{\ell}{\bfwl{\ell}_{\bfs'}(v)}
    = \feat{\ell}(v)$.
    Therefore
    $\cG \models \varphil{\ell}_{\bfs}(v)$.

    If $\cA$ is local, then for $1 \le \ell \le L$, $\coefR{\ell}$ is a zero matrix.
    Thus, there are no $\top$ guard atoms in $\phil{\ell}_{k, \bfs', i}(x)$.
    Hence the formula we obtained is in $\blMPtwo$. 

    Similarly, if $\cA$ is outging-only, then for $1 \le \ell \le L$, $\coefA{\ell}{\inc}$ is a zero matrix.
    Thus, there are no $E(y, x)$ guard atoms in $\phil{\ell}_{k, \bfs', i}(x)$.
    Hence the formula we obtained is in $\ogMPtwo$. 
\end{proof}

We can now prove Theorem~\ref{thm:gnn_to_logic}.

\begin{proof}
    Let $\spectrumoverL := \setc{\bfs \in \spectrumover{L}}{\bfs_1 \ge 1/2}$ and
    $\psi_\cA(x) := \bigvee_{\bfs \in \spectrumoverL} \varphil{L}_\bfs(x)$,
    where $\varphil{L}_\bfs(x)$ is the formulas defined in Lemma~\ref{lem:gnn_to_logic}.
    
    For every $n$-graph $\cG$ and vertex $v$ in $\cG$, 
    if $\cG \models \psi_\cA(v)$,
    then there exists ${\bfs \in \spectrumoverL}$,
    such that $\cG \models \varphil{L}_\bfs(v)$.
    By Lemma~\ref{lem:gnn_to_logic}, 
    $\feat{L}(v) = s$.
    Hence ${\bfs_1 = \feat{L}_1(v) \ge 1/2}$, which implies that $\cA$ accepts $\tuple{\cG, v}$.

    On the other hand, if $\cA$ accepts $\tuple{\cG, v}$,
    by Lemma~\ref{lem:gnn_to_logic},
    $\cG \models \varphil{L}_{\feat{L}(v)}(v)$.
    By definition of acceptness of GNNs, 
    $\feat{L}_1(v) \ge 1/2$,
    which implies that $\feat{L}(v) \in \spectrumoverL$.
    Therefore $\cG \models \psi_\cA(v)$.
    
    If $\cA$ is an $n$-$\blcGNN$ ($n$-$\ogcGNN$), the formulas defined in Lemma~\ref{lem:gnn_to_logic} are in $n$-$\blMPtwo$ ($n$-$\ogMPtwo$).
    Hence $\psi_\cA(x)$ is also in $\blMPtwo$ ($\ogMPtwo$). 
\end{proof}

Recall that finite satisfiability of $\blMPtwo$ is decidable by Corollary~\ref{corollary:mp2_decidable}. Combining this with Theorem~\ref{thm:gnn_to_logic} we get decidability of satisfiability for $\blcGNN$:

\begin{theorem}\label{thm:local_gnn_decidable}
    % For $\blcGNN$s, we can compute the $\ell$-spectrum for each $\ell$. In particular,
    The satisfiability problem of $\blcGNN$s is decidable.
\end{theorem}

% By Theorem~\ref{thm:gnn_to_logic} it suffices to decide finite satisfiability of the corresponding $\blMPtwo$ formula. But we can do this by Corollary~\ref{corollary:mp2_decidable}.

Recall from Theorem~\ref{thm:gp2_decidabel} that finite satisfiability for the richer logic $\GPtwo$, allowing unguarded unary quantification and containing $\blMPtwo$, is also decidable.
Using this with Theorem~\ref{thm:gnn_to_logic} gives decidability of universal satisfiability for $\blcGNN$.

\begin{theorem}\label{thm:local_gnn_universal_decidable}
    The universal satisfiability problem of $\blcGNN$s is decidable.
\end{theorem}

\begin{proof}
    For every $\blcGNN$ $\cA$, by Theorem~\ref{thm:gnn_to_logic},
    there exists a $\blMPtwo$ formula $\psi_\cA(x)$ such that $\cA$ and $\psi_\cA(x)$ are equivalent.
    We claim that $\cA$ is universally satisfiable if and only if the $\GPtwo$ sentence $\varphi := \forall x\ (x=x) \to \psi_\cA(x)$ is finitely satisfiable.
    % Besides, by Lemma 

    Suppose that $\psi$ is finitely satisfiable by the graph $\cG$.
    Then for every $v$ in $\cG$, $\cG \models \psi_\cA(v)$.
    Since $\cA$ and $\psi_\cA(x)$ are equivalent, $\cA$ accepts $\tuple{\cG, v}$.
    Hence $\cA$ is universally satisfiable, with witness graph $\cG$.
    On the other hand, Suppose that $\cA$ is universally satisfiable, with witness the finite graph $\cG$.
    By definition, for every $v$ in $\cG$, $\cA$ accepts $\tuple{\cG, v}$.
    Since $\cA$ and $\psi_\cA(x)$ are equivalent, $\cG \models \psi_\cA(v)$.
    Hence $\varphi$ is satisfiable by $\cG$.
    Note that the size of $\cG$ is finite, which implies that $\varphi$ is finitely satisfiable.
\end{proof}

The following converse to Theorem~\ref{thm:gnn_to_logic} shows that the logic is equally expressive as the GNN model:

\begin{theorem}\label{thm:logic_to_gnn}
    For every $n$-$\MPtwo$ formula $\psi(x)$,
    there exists an $n$-$\bgtrreluGNN$ $\cA_\psi$,
    such that $\psi(x)$ and $\cA_\psi$ are equivalent.
    If we start with $n$-$\blMPtwo$ and $n$-$\ogMPtwo$ formulas, we obtain $n$-$\bltrreluGNN$ and $n$-$\ogtrreluGNN$, respectively.
\end{theorem}

The idea of the proof is induction on the formula structure. For each subformula there will be an entry of a  feature for the GNN which represents the subformula, in the sense that -- for the second to last layer -- 
its value is $1$ if the subformula holds, or $0$ otherwise. We will have an entry for each subformula at every iteration, but as we progress to later layers of the GNN, more of these entries will be correct with respect to the corresponding subformula. In an inductive case for a Presburger quantifier that uses some coefficients $\lambda_t$, the corresponding matrix coefficient will be multiplying certain quantifies by $\lambda_t$.
Note that this translation is polynomial time, thus the size of the corresponding GNN is polynomial in the formula.

For every $n$-$\bgMPtwo$ formula $\psi(x)$, let $L$ be the number of subformulas of $\psi(x)$ and $\set{\varphi_\ell(x)}_{\ell \in \intsinterval{L}}$ be an enumeration of subformulas of $\psi(x)$ satisfying that
$\varphi_L(x)$ is $\psi(x)$
and for each $\varphi_i(x)$ and $\varphi_j(x)$, if $\varphi_i(x)$ is a strict subformula of $\varphi_j(x)$, then $i < j$.

\begin{definition}\label{def:correct}
    For a $\bgMPtwo$ formula $\Psi(x)$, for a graph $\cG$ and a vertex $v$ in $\cG$,
    we say that a vector $\bfv$ is \emph{$\ell$-correct} for $\Psi(x)$ w.r.t. $\tuple{\cG, v}$,
    if, for $1 \le t \le \ell$, $\bfv_t = \eval{\varphi_t(v)}_\cG$.
\end{definition}

We define the $(L+1)$-layer $n$-$\bgtrreluGNN$ $\cA_\Psi$ as follows.
The input dimension $\gnndim{0}$ is $n$.
For $1 \le \ell \le L$, the dimension $\gnndim{\ell}$ is $L + n$, and $\gnndim{L+1} = 1$.
The coefficient matrices and bias vectors are chosen so that the features of $\cA_\Psi$ satisfies the following conditions.
% Note that we will omit some $\trrelu$,
% since $\trrelup{\trrelup{x}} = \trrelup{x}$,
% which implies that $\trrelup{\feat{i}_j(v)} = \feat{i}_j(v)$.

We first describe the conditions for the first $L$ entries of features.
For $1 \le \ell \le L$,
\begin{equation*}
    \feat{L}_{\ell}(v)\ =\ \feat{L-1}_{\ell}(v)\ =\ \cdots\ =\ \feat{\ell}_{\ell}(v).
\end{equation*}
The condition of $\feat{\ell}_{\ell}(v)$ depend on the formula $\varphi_\ell(x)$.
\begin{itemize}
    \item If $\varphi_\ell(x)$ is $\top$,
    then $\feat{\ell}_\ell(v) = 1$.
    \item If $\varphi_\ell(x)$ is $U_j(x)$ for some unary predicate $U_j$,
    then $\feat{\ell}_\ell(v) = \feat{\ell-1}_{L+j}(v)$.
    \item If $\varphi_\ell(x)$ is $\neg\varphi_j(x)$,
    then $\feat{\ell}_\ell(v) = \trrelup{1-\feat{\ell-1}_j(v)}$.
    \item If $\varphi_\ell(x)$ is $\varphi_{j_1}(x)\land\varphi_{j_2}(x)$,
    then $\feat{\ell}_\ell(v) = \trrelup{\feat{\ell-1}_{j_1}(v)+\feat{\ell-1}_{j_2}(v)-1}$.
    \item If $\varphi_\ell(x)$ is ${\left(\sum_{t \in \intsinterval{k}}\ \lambda_t \cdot \presby{\epsilon_t(x, y)\land\varphi_{j_t}(y)} \ge \delta\right)}$,
    then 
    \begin{equation*}
        \feat{\ell}_\ell(v)\ =\ 
        \trrelup{
            \sum_{t \in \intsinterval{k}} \lambda_t \cdot \sum_{u \in X_t(v)}\feat{\ell-1}_{j_t}(u) - \delta + 1
        },
    \end{equation*}
    where
    \begin{equation*}
        X_t(v)\ :=\ 
        \begin{cases}
            \nbr{\out}(v), &\text{if $\epsilon_t(x, y)$ is $E(x, y)$} \\
            \nbr{\inc}(v), &\text{if $\epsilon_t(x, y)$ is $E(y, x)$} \\
            V,             &\text{if $\epsilon_t(x, y)$ is $\top$}.
        \end{cases}
    \end{equation*}
\end{itemize}
Next, the conditions for the last $n$ entries of features are as follows:
for $1 \le t \le n$,
\begin{equation*}
    \feat{L}_{L+t}(v)\ =\ \feat{L-1}_{L+t}(v)\ =\ \cdots\ =\ \feat{1}_{L+t}(v)\ =\ \feat{0}_{t}(v).
\end{equation*}
Finally, for the last layer, $\feat{L+1}_1(v) = \feat{L}_L(v)$.

It is easy to see that there are GNNs that satisfy these conditions.

The theorem will follow once we have shown the following property of $\cA_\Psi$:
% \begin{lemma} \label{lem:logicmimicsgnn}
%     For every $n$-$\bgMPtwo$ formula $\Psi(x)$,
%     let $\cA_\Psi$ be the $\bgcGNN$ defined above.
%     For every $n$-graph $\cG$ and vertex $v \in V$,
%     for $1 \le i \le L$,
%     if $\cG \models \varphi_i(v)$,
%     then for $i \le \ell \le L$, $\feat{\ell}_i(v) = 1$.
%     Otherwise, if $\cG \not\models \varphi_i(v)$,
%     then for $i \le \ell \le L$, $\feat{\ell}_i(v) = 0$.
% \end{lemma}
\begin{lemma} \label{lem:logicmimicsgnn}
    For every $n$-graph $\cG$ and vertex $v$ in $\cG$,
    for $1 \le \ell \le L$, the feature
    $\feat{\ell}(v)$ within $\cA_\Psi$ is $\ell$-correct for $\Psi(x)$ w.r.t. $\tuple{\cG, v}$.
\end{lemma}

\begin{proof}
    For $1 \le \ell \le L$, 
    for $1 \le t \le \ell - 1$,
    note that $\feat{\ell}_t(v) = \feat{\ell-1}_t(v)$.
    By the induction hypothesis, 
    $\feat{\ell-1}(v)$ is $(\ell-1)$-correct for $\Psi(x)$ w.r.t. $\tuple{\cG, v}$,
    which implies that $\feat{\ell-1}_t(v) = \eval{\varphi_{t}(v)}$.
    Thus, it holds that $\feat{\ell}_t(v) = \eval{\varphi_{t}(v)}$.

    % Note that for $1 \le i \le L$, for $i+1 \le \ell \le L$,
    % $\feat{\ell}_i(v) = \feat{i}_i(v)$.
    % Thus it is sufficient to show that 
    % $\feat{i}_i(v) = \eval{\varphi_i(v)}$.
    % if $\cG \models \varphi_i(v)$,
    % then $\feat{i}_i(v) = 1$.
    % Otherwise, if $\cG \not\models \varphi_i(v)$,
    % then $\feat{i}_i(v) = 0$.
    It remains to show that $\feat{\ell}_\ell(v) = \eval{\varphi_{\ell}(v)}$.
    We prove the property by induction on subformulas.
    \begin{itemize}
        \item If $\varphi_\ell(x)$ is $\top$,
        then $\feat{\ell}_\ell(v) = 1$ by the definition of $\cA_\Psi$.

        \item If $\varphi_\ell(x)$ is $U_j(x)$ for some unary predicate $U_j$,
        then $\feat{\ell}_\ell(v) = \feat{\ell-1}_{L+j}(v)$.
        Note that $\feat{\ell-1}_{L+j}(v) = \feat{0}_{j}(v)$.
        If $\cG \models U_j(v)$,
        then $\feat{0}_j(v) = 1$,
        which implies $\feat{\ell}_\ell(v) = 1$. 
        Otherwise, if $\cG \not\models U_j(v)$,
        then $\feat{0}_j(v) = 0$
        which implies $\feat{\ell}_\ell(v) = 0$. 
        
        \item Suppose that $\varphi_\ell(x)$ is $\neg\varphi_j(x)$.
        % then $\feat{i}_i(v) = \trrelup{1 - \feat{i-1}_j(v)}$.
        Because $\varphi_j(x)$ is a strict subformula of $\varphi_\ell(x)$, $\ell-1 \ge {j}$.
        By the induction hypothesis, $\feat{\ell-1}_j(v) = \eval{\varphi_j(v)}$.
        Thus, we have
        \begin{equation*}
            \feat{\ell}_\ell(v)
            % \ =\ \trrelup{1 - \feat{i-1}_j(v)}
            \ =\ \trrelup{1 - \eval{\varphi_j(v)}}
            \ =\ \eval{\varphi_\ell(v)}.
        \end{equation*}
        % If $\cG \models {\neg\varphi_j(v)}$,
        % then $\cG \not\models {\varphi_j(v)}$.
        % By the induction hypothesis,
        % $\feat{i-1}_j(v) = 0$
        % which implies $\feat{i}_i(v) = 1$. 
        % Otherwise, if $\cG \not\models \neg\varphi_j(v)$,
        % then $\cG \models {\varphi_j(v)}$.
        % By the induction hypothesis,
        % $\feat{i-1}_j(v) = 1$
        % which implies $\feat{i}_i(v) = 0$. 
        
        \item Suppose that $\varphi_\ell(x)$ is $\varphi_{j_1}(x)\land\varphi_{j_2}(x)$.
        % then $\feat{i}_i(v) = \trrelup{\feat{i-1}_{j_1}(v) + \feat{i-1}_{j_2}(v) - 1}$.
        Because $\varphi_{j_1}(x)$ and $\varphi_{j_2}(x)$ are strict subformulas of $\varphi_\ell(x)$,
        $\ell-1 \ge j_1$ and $\ell-1 \ge j_2$.
        By the induction hypothesis,
        $\feat{\ell-1}_{j_1}(v) = \eval{\varphi_{j_1}(v)}$ and
        $\feat{\ell-1}_{j_2}(v) = \eval{\varphi_{j_2}(v)}$.
        Thus, we have
        \begin{equation*}
            \feat{\ell}_\ell(v)
            % \ =\ \trrelup{\feat{i-1}_{j_1}(v) + \feat{i-1}_{j_2}(v) - 1}
            \ =\ \trrelup{\eval{\varphi_{j_1}(v)} + \eval{\varphi_{j_2}(v)} - 1}
            \ =\ \eval{\varphi_\ell(v)}.
        \end{equation*}
        % If $\cG \models \varphi_{j_1}(v)\land\varphi_{j_2}(v)$,
        % then $\cG \models \varphi_{j_1}(v)$ and $\cG \models \varphi_{j_2}(v)$.
        % By the induction hypothesis,
        % $\feat{i-1}_{j_1}(v) = \feat{i-1}_{j_2}(v) = 1$
        % which implies $\feat{i}_i(v) = 1$.
        % Otherwise, if $\cG \not\models \varphi_{j_1}(v)\land\varphi_{j_2}(v)$,
        % then $\cG \not\models \varphi_{j_1}(v)$ or $\cG \not\models \varphi_{j_2}(v)$.
        % By the induction hypothesis,
        % $\feat{i-1}_{j_1}(v) + \feat{i-1}_{j_2}(v) \le 1$
        % which implies $\feat{i}_i(v) = 0$. 
        
        \item Suppose that $\varphi_\ell(x)$ is ${\left(\sum_{t \in \intsinterval{k}}\ \lambda_{t}\cdot \presby{\epsilon_t(x, y)\land\varphi_{j_t}(y)} \ge \delta\right)}$.
        For $1 \le t \le k$,
        because $\varphi_{j_t}(x)$ is a strict subformula of $\varphi_\ell(x)$,
        $\ell - 1 \ge {j_t}$.
        By the induction hypothesis,
        for every $u$ in $\cG$, $\feat{\ell-1}_{j_t}(u) = \eval{\varphi_{j_t}(u)}$.
        % if $\cG \models \varphi_{j_t}(u)$, then $\feat{i-1}_{j_t}(u) = 1$.
        % Otherwise, if $\cG \not\models \varphi_{j_t}(u)$, then $\feat{i-1}_{j_t}(u) = 0$.
        We note that
        % If $\epsilon_t(x, y) = E(x, y)$, then 
        \begin{equation*}
            \begin{aligned}
                \abs{\setc{u \in V}{\cG \models \epsilon_t(v, u)\land\varphi_{j_t}(u)}}
                % \ =\ & \abs{\setc{u \in X_t(v)}{\cG\models \varphi_{j_t}(u)}} \\
                \ =\ & \sum_{u \in X_t(v)} \eval{\varphi_{j_t}(u)} 
                \ =\ & \sum_{u \in X_t(v)} \feat{\ell-1}_{j_t}(u).
            \end{aligned}
        \end{equation*}
        % Note that $\nbr{\out}(v)$ is $X_t(v)$ in this case.
        % We can treat the other two cases analogously and show that
        % \begin{equation*}
        %     \abs{\setc{u \in V}{\cG \models \epsilon_t(v, u)\land\varphi_{j_t}(u)}}
        %     \ =\ 
        %     \sum_{u \in X_t(v)} \feat{i-1}_{j_t}(u). 
        % \end{equation*}
        Let $w(v)$ be a value defined as follows:
        \begin{equation*}
            w(v) \ :=\ 
            \sum_{t \in \intsinterval{k}} \lambda_t \cdot
            \abs{\setc{u \in V}{\cG \models \epsilon_t(v, u)\land\varphi_{j_t}(u)}}
            \ =\ 
            \sum_{t \in \intsinterval{k}} \lambda_t \cdot \sum_{u \in X_t(v)}\feat{\ell-1}_{j_t}(u).
        \end{equation*}
        Observe that
        \begin{equation*}
            \feat{\ell}_\ell(v)\ =\ 
            \trrelup{
                \sum_{t \in \intsinterval{k}} \lambda_t \cdot \sum_{u \in X_t(v)}\feat{\ell-1}_{j_t}(u) - \delta + 1
            }\ =\ \trrelup{w(v)-\delta+1}.
        \end{equation*}
        By the semantic of Presburger quantifiers,
        if $\cG \models \varphi_\ell(v)$, then $w(v) \ge \delta$,
        which implies $\feat{\ell}_\ell(v) = 1$.
        On the other hand,
        if $\cG \not\models \varphi_\ell(v)$, then $w(v) < \delta$,
        which implies $\feat{\ell}_\ell(v) = 0$.
    \end{itemize}
\end{proof}

We can now prove Theorem~\ref{thm:logic_to_gnn}.

\begin{proof}
    For every $n$-$\bgMPtwo$ formula $\Psi(x)$,
    let $\cA_\Psi$ be the $n$-$\bgtrreluGNN$ defined in Lemma~\ref{lem:logicmimicsgnn}.
    Note that $\feat{L+1}_1(v) = \feat{L}_L(v)$.
    
    For every $n$-graph $\cG$ and vertex $v$ in $\cG$,
    if $\cG \models \Psi(v)$,
    since $\varphi_L(x)$ is $\Psi(x)$,
    we have $\cG \models \varphi_L(v)$.
    By Lemma~\ref{lem:logicmimicsgnn},
    $\feat{L}_L(v) = 1$,
    which implies $\feat{L+1}_1(v) = 1$ and $\cA_\Psi$ accepts $\tuple{\cG, v}$.
    On the other hand,
    if $\cG \not\models \Psi(v)$,
    we have $\cG \not\models \varphi_L(v)$.
    By Lemma~\ref{lem:logicmimicsgnn} again,
    $\feat{L}_L(v) = 0$,
    which implies $\feat{L+1}_1(v) = 0$ and $\cA_\Psi$ does not accept $\tuple{\cG, v}$.
    
    If $\Psi(x)$ is an $n$-$\blMPtwo$ formula,
    by construction,
    for $1 \le \ell \le L$,
    $\coefR{\ell}$ is a zero matrix.
    Therefore $\cA_\Psi$ is an $n$-$\bltrreluGNN$.
    Similarly, if $\Psi(x)$ is an $n$-$\ogMPtwo$ formula,
    by construction,
    for $1 \le \ell \le L$,
    $\coefA{\ell}{\inc}$ is a zero matrix.
    Therefore $\cA_\Psi$ is an $n$-$\ogtrreluGNN$.
\end{proof}

Putting together the two translation results, we have:
\begin{corollary} \label{cor:logicequalgnn}
    The logic $\MPtwo$ and $\bgcGNN$s are expressively equivalent,
    as are $\blMPtwo$ and $\blcGNN$s; $\ogMPtwo$ and $\ogcGNN$s.
\end{corollary}

The translations also tell us that \emph{the expressiveness of GNNs with $\trrelu$
is the same as that of GNNs with arbitrary eventually constant activations} -- provided we use the Boolean semantics based on thresholds.

\subsection{Undecidability of $\MPtwo$, and of GNNs with $\trrelu$ and global readout} \label{subsec:global_undecidable}

Note that we claimed in Theorem~\ref{thm:computespectrum}  that the spectrum is finite for GNNs with eventually constant activations, even when they have  global readout. And  we could compute a finite overapproximation of the spectrum. But in our decidability argument for $\blcGNN$, we required further the ability to decide membership in the spectrum for any fixed rational, and for this we utilized decidability of the logic.  So what happens to decidability of the GNNs -- or the corresponding logic -- when global readout is allowed?

We show undecidability of finite satisfiability for the logic $\MPtwo$, and of the corresponding GNN satisfiability problem. First for the logic:

\begin{theorem}\label{thm:global_mptwo_undecidable}
     The finite satisfiability problem of $\MPtwo$ is undecidable.
\end{theorem}

For the proof we apply an approach based on ideas  
in~\cite{localpresburgerbartosztony},  
using a reduction from Hilbert's tenth problem.

\begin{definition}\label{def:simple_eq_sys}
    A \emph{simple equation system $\varepsilon$} (with $n$ variables and $m$ equations) is a set of $m$
    equations of one of the forms  $\upsilon_{t_1} = 1$, $\upsilon_{t_1} = \upsilon_{t_2} + \upsilon_{t_3}$, or $\upsilon_{t_1} = \upsilon_{t_2} \cdot \upsilon_{t_3}$,
    where $1 \le t_1, t_2, t_3 \le n$ are pairwise distinct.
    We say the system $\varepsilon$ is solvable if it has a solution in $\bbN$.
\end{definition}

\begin{lemma}\label{lemma:hilbert_to_mptwo}
    For every simple equation system $\varepsilon$ with $n$ variables and $m$ equations,
    there exists an $(n+m)$-$\MPtwo$ formula $\Psi_\varepsilon(x)$
    such that the following are equivalent.
    \begin{enumerate}
        \item The system $\varepsilon$ has a solution in $\bbN$.
        \item There exists an $(n+m)$-graph $\cG$ such that for every vertex $v$ in $\cG$, $\cG \models \Psi_\varepsilon(v)$.
        \item The formula $\Psi_\varepsilon(x)$ is finitely satisfiable.
    \end{enumerate}
\end{lemma}

\begin{proof}
    We construct the formula $\Psi_\varepsilon(x)$ as follows.
    The vocabulary of $\Psi_\varepsilon(x)$ consists of unary predicates $\set{P_t}_{t \in \intsinterval{m}}$ and $\set{U_i}_{i \in \intsinterval{n}}$.
    For $1 \le t \le m$, we define $\varphi_t(x)$ depending on the $t^{th}$ equation in $\varepsilon$. %depending on, not depended on
    \begin{itemize}
        \item If the $t^{th}$ equation is $\upsilon_{t_1} = 1$,
        then $\varphi_t(x) := (\presby{P_t(y) \land U_{t_1}(y)} = 1)$.

        \item If the $t^{th}$ equation is $\upsilon_{t_1} = \upsilon_{t_2} + \upsilon_{t_3}$,
        then $\varphi_t(x) := (\presby{\psi_t(y)} = \presby{\top})$ where
        \begin{equation*}
            \begin{aligned}
                \psi_t(y)\ :=\ 
                &\left(P_t(y) \land (U_{t_2}(y) \lor U_{t_3}(y)) \to \left(\presbx{E(y, x)\land P_t(x)\land U_{t_1}(x)} = 1\right)\right)\ \land\ \\
                &\left(P_t(y) \land U_{t_1}(y) \to
                    \left(\presbx{E(x, y) \land P_t(x) \land (U_{t_2}(x) \lor U_{t_3}(x))} = 1\right)\right).
            \end{aligned}
        \end{equation*}
        %intuition: we have a bijection between U_{j_2} cup U_{j_3} and U_1
        
        \item If the $t^{th}$ equation is $\upsilon_{t_1} = \upsilon_{t_2} \cdot \upsilon_{t_3}$,
        then $\varphi_t(x) := (\presby{\psi_t(y)} = \presby{\top})$ where
        \begin{equation*}
            \begin{aligned}
                \psi_t(y)\ :=\ 
                &\left(P_t(y) \land U_{t_2}(y) \to \left(\presbx{E(y, x)\land P_t(x) \land U_{t_1}(x)} = \presbx{P_t(x) \land U_{t_3}(x)}\right)\right)\ \land\ \\
                &\left(P_t(y) \land U_{t_1}(y) \to
                    \left(\presbx{E(x, y)\land P_t(x) \land U_{t_2}(x)} = 1\right)\right).
            \end{aligned}
        \end{equation*}
    \end{itemize}
    We now define $\Psi_\varepsilon(x)$:
    \begin{equation*}
        \Psi_\varepsilon(x)\ :=\ 
            \psi^{\itdisj}(x) \land \psi^{\iteq}(x) \land
            \bigwedge_{t \in \intsinterval{m}} \varphi_t(x),
    \end{equation*}
    where
    \begin{equation*}
        \begin{aligned}
            \psi^{\itdisj}(x)\ :=\ &
                \bigwedge_{\substack{
                    t, t' \in \intsinterval{m} \\
                    t \neq t'
                }}
                (\presby{P_{t}(y) \land P_{t'}(y)} = 0)
                \ \land\ 
                \bigwedge_{\substack{
                    i, i' \in \intsinterval{n} \\
                    i \neq i'
                }}
                (\presby{U_{i}(y) \land U_{i'}(y)} = 0) \\
            \psi^{\iteq}(x)\ :=\ &
                \bigwedge_{\substack{
                    t, t' \in \intsinterval{m} \\
                    i \in \intsinterval{n}
                }}
                (\presby{P_{t}(y) \land U_i(y)} =
                \presby{P_{t'}(y) \land U_i(y)}) \\
        \end{aligned}   
    \end{equation*}

    \textbf{\underline{(1) $\Rightarrow$ (2)}}, suppose $\varepsilon$ has a solution
    $\set{\upsilon_i \gets a_i}_{i \in \intsinterval{n}}$.
    Let $\cG$ be the $(n+m)$-graph defined as follows.
    For $1 \le t \le m$, $1 \le i \le n$, and $1 \le j \le a_i$, let $v_{t, i, j}$ be a fresh vertex
    and $V$ be the set of all such vertices.
    For $1 \le t \le m$,
    let $P_t := \setc{v_{t, i, j}}{1 \le i \le n; 1 \le j \le a_i}$.
    For $1 \le i \le n$, 
    let $U_i := \setc{v_{t, i, j}}{1 \le t \le m; 1 \le j \le a_i}$.
    For $1 \le t \le m$, we define
    $E_t$ depending on the $t^{th}$ equation in $\varepsilon$.
    \begin{itemize}
        \item If the $t^{th}$ equation is $\upsilon_{t_1} = 1$,
        then $E_t := \emptyset$.
        
        \item If the $t^{th}$ equation is $\upsilon_{t_1} = \upsilon_{t_2} + \upsilon_{t_3}$, then
        \begin{equation*}
            E_t := 
            \setc{(v_{t, t_2, j}, v_{t, t_1, j})}{1 \le j \le a_{t_2}} \cup
            \setc{(v_{t, t_3, j}, v_{t, t_1, a_{t_2} + j})}{1 \le j \le a_{t_3}}.
        \end{equation*}

        \item If the $t^{th}$ equation is $\upsilon_{t_1} = \upsilon_{t_2} \cdot \upsilon_{t_3}$, then
        \begin{equation*}
            E_t :=
            \setc{(v_{t, t_2, j}, v_{t, t_1, (j-1) \cdot a_{t_3} + j'})}{1 \le j \le a_{t_2}; 1 \le j' \le a_{t_3}}.
        \end{equation*}
    \end{itemize}
    Finally, let $E := \ \bigcup_{t \in \intsinterval{m}}E_t$.
    % \begin{equation*}
    %     E\ :=\ \bigcup_{t \in \intsinterval{m}}
    %     \setc{(v, u)}{(v, u) \in E_t} \cap \setc{(u, v)}{(v, u) \in E_t}
    % \end{equation*}
    
    Note that $\set{P_t}_{t \in \intsinterval{m}}$
    and $\set{U_i}_{i \in \intsinterval{n}}$ form partitions of $V$, respectively.
    Furthermore, for $1 \le t \le m$ and $1 \le i \le n$,
    $\abs{P_t \cap U_i} = a_i$.
    Thus, it is straightforward to verify that for every $v \in V$,
    it holds that $\cG \models \psi^{\iteq}(v) \land \psi^{\itdisj}(v)$.
    For $1 \le t \le m$, we check the satisfiability of $\varphi_t(x)$ depending on the $t^{th}$ equation in $\varepsilon$.
    \begin{itemize}
        \item If the $t^{th}$ equation is $\upsilon_{t_1} = 1$,
        then $a_{t_1} = 1$.
        Thus $\abs{P_t \cap U_{t_1}} = 1$,
        which implis that for every $v \in V$,
        then $\cG \models \varphi_t(v)$.

        \item If the $t^{th}$ equation is $\upsilon_{t_1} = \upsilon_{t_2} + \upsilon_{t_3}$,
        then $a_{t_1} = a_{t_2} + a_{t_3}$.
        We first claim that for every $u \in V$, $\cG \models \psi_t(u)$.
        \begin{itemize}
            \item If $u = v_{t, t_2, j}$ for some $1 \le j \le a_{t_2}$,
            then $\abs{\setc{u' \in V}{E(u, u') \land P_t(u') \land U_{t_1}(u')}} = \abs{\set{v_{t, t_1, j}}} = 1$.
            \item If $u = v_{t, t_3, j}$ for some $1 \le j \le a_{t_3}$,
            then $\abs{\setc{u' \in V}{E(u, u') \land P_t(u') \land U_{t_1}(u')}} = \abs{\set{v_{t, t_1, a_{t_2} + j}}} = 1$.
            \item If $u = v_{t, t_1, j}$ for some $1 \le j \le a_{t_1}$,
            if $j \le a_{t_2}$, then there is only one edge from $v_{t, t_2, j}$ to $u$. 
            Otherwise, there is only one edge from $v_{t, t_3, j - a_{t_2}}$ to $u$. 
            Thus $\abs{\setc{u' \in V}{E(u', u) \land P_t(u') \land \left(U_{t_2}(u') \lor U_{t_3}(u') \right)}} = 1$.
        \end{itemize} 
        Therefore, for every $v \in V$,
        since $\abs{\setc{u \in V}{\cG \models \psi_t(u)}} = \abs{V}$,
        it holds that $\cG \models \varphi_t(v)$.
        
        \item If the $t^{th}$ equation is $\upsilon_{t_1} = \upsilon_{t_2} \cdot \upsilon_{t_3}$,
        then $a_{t_1} = a_{t_2} \cdot a_{t_3}$.
        We first claim that for every $u \in V$, $\cG \models \psi_t(u)$.
        \begin{itemize}
            \item If $u = v_{t, t_2, j}$ for some $1 \le j \le a_{t_2}$,
            then
            \begin{equation*}
                \abs{\setc{u' \in V}{E(u, u') \land P_t(u') \land U_{t_1}(u')}}
                \ =\ 
                \abs{\setc{v_{i, t_1, (j-1)\cdot a_{t_3} + j'}}{1 \le j' \le a_{t_3}}}
                \ =\ a_{t_3}.
            \end{equation*}
            Note that
            $\abs{\setc{u' \in V}{P_t(u') \land U_{t_3}(u')}} =
            \abs{P_t \cap U_{t_3}} = a_{t_3}$.
            \item If $u = v_{t, t_1, k}$ for some $1 \le j \le a_{t_1}$,
            then there is only one edge from $v_{t, t_2, j'}$ to $u$,
            where $j' = \floor{j / a_{t_2}} + 1$.
            Thus $\abs{\setc{u' \in V}{E(u', u) \land P_t(u') \land U_{t_2}(u')}} = 1$.
        \end{itemize} 
        Therefore, for every $v \in V$,
        since $\abs{\setc{u \in V}{\cG \models \psi_t(u)}} = \abs{V}$,
        $\cG \models \varphi_t(v)$.
    \end{itemize}
    Hence for every vertex $v \in V$, $\cG \models \Psi_\varepsilon(v)$.

    \textbf{\underline{(2) $\Rightarrow$ (3)}},
    it is clear that (3) is a simple corollary (2).

    \textbf{\underline{(3) $\Rightarrow$ (1)}},
    suppose $\cG \models \Psi_\varepsilon(v)$.
    First of all,
    for $1 \le t \le m$ and $1 \le i \le n$,
    let
    \begin{equation*}
        V_{t, i}\ :=\ \setc{u \in V}{\cG \models P_t(u) \land U_i(u)}.
    \end{equation*}
    Because ${\cG \models \psi^{\itdisj}(v)}$,
     $V_{t, i}$ forms a partition of $V$.
    Next, because $\cG \models \psi^{\iteq}(v)$,
    for every $1 \le t < t' \le m$, and $1 \le i \le n$,
    $\abs{V_{t_1, i}} = \abs{V_{t_2, i}}$.
    Finally, we claim that 
    $\set{\upsilon_i \gets \abs{V_{1, i}}}_{i \in \intsinterval{n}}$
    is a solution of $\varepsilon$.
    We show that $\cG \models \varphi_t(v)$ implies that
    $\set{\upsilon_i \gets \abs{V_{1, i}}}_{i \in \intsinterval{n}}$ satisfies the $t^{th}$ equation in $\varepsilon$.
    \begin{itemize}
        \item If the $t^{th}$ equation is $\upsilon_{t_1} = 1$,
        then $\cG \models \left(\presby{P_t(y) \land U_{t_1}(y)} = 1\right)$,
        which implies that 
        $\abs{V_{t, t_1}} = 1$.
        Hence $\abs{V_{1, t_1}} = \abs{V_{t, t_1}} = 1$.

        \item If the $t^{th}$ equation is $\upsilon_{t_1} = \upsilon_{t_2} + \upsilon_{t_3}$,
        then $\cG \models \left(\presby{\psi_t(y)} = \presby{\top}\right)$, which implies that 
        $\abs{\setc{u \in V}{\cG \models \psi_t(u)}} = \abs{V}$.
        Hence for each $u \in V$, $\cG \models \psi_t(u)$.

        We will argue that  the  edges between $V_{t, t_2} \cup V_{t, t_3}$ and $V_{t, t_1}$ give us a bijection, which will establish the satisfaction of the equation.
        For each $u \in V_{t, t_2} \cup V_{t, t_3}$,
        since $\cG \models \left(\presbx{E(u, x)\land P_t(x) \land U_{t_1}(x)} = 1\right)$,
        there exists only one edge from $u$ to the set $V_{t, t_1}$.
        On the other hand,
        for each $u \in V_{t, t_1}$,
        $\cG \models \left(\presbx{E(x, u) \land P(x)\land (U_{t_2}(x) \lor U_{t_3}(x))} = 1\right)$,
        which implies that there exists only one edge from the set $V_{t, t_2} \cup V_{t, t_3}$ to $u$.
        Therefore,
        by considering the number of edges from the set $V_{t, t_2} \cup V_{t, t_3}$ to $V_{t, t_1}$,
        we have 
        $1 \cdot \abs{V_{t, t_1}} = 1 \cdot \abs{V_{t, t_2} \cup V_{t, t_3}}$.
        Thus $\abs{V_{t, t_1}} = \abs{V_{t, t_2}} + \abs{V_{t, t_3}}$,
        since $V_{t, t_2}$ and $V_{t, t_3}$ are disjoint.

        \item If the $t^{th}$ equation is $\upsilon_{t_1} = \upsilon_{t_2} \cdot \upsilon_{t_3}$,
        by an argument similar to the one above,
        for each $u \in V$, $\cG \models \psi_t(u)$.
                
        We will establish the equality by using the edges from the set $V_{t, t_2}$ to the set $V_{t, t_1}$
        to show that there is a $\abs{V_{t, t_3}}$-to-one relationship from $V_{t, t_2}$ to $V_{t, t_1}$.
        For each $u \in V_{t, t_2}$, since
        \begin{equation*}
            \cG \models \left(\presbx{E(u, x)\land P_t(x) \land U_{t_1}(x)} = \presbx{P_t(x) \land U_{t_3}(x)}\right), 
        \end{equation*}
        there exists $\abs{V_{t, t_3}}$ edges from $u$ to the set $V_{t, t_1}$.
        On the other hand,
        for each $u \in V_{t, t_1}$,
        since $\cG \models \left(\presbx{E(x, u)\land P_t(x) \land U_{t_2}(x)} = 1\right)$,
        there exists only one edge from the set $V_{t, t_2}$ to $u$.
        Therefore we have $1 \cdot \abs{V_{t, t_1}} = \abs{V_{t, t_3}} \cdot \abs{V_{t, t_2}}$. 
    \end{itemize}
\end{proof}

Since the solvability (over $\bbN$) of simple equation systems is undecidable~\cite{hilbert_theth}, Theorem~\ref{thm:global_mptwo_undecidable}
follows. From the theorem and Corollary~\ref{cor:logicequalgnn} we obtain undecidability of static analysis for GNNs with global readout:

\begin{theorem}\label{thm:global_gnn_undecidable}
    The satisfiability problem of $\bgtrreluGNN$s is undecidable.
\end{theorem}

% \ch{it is trivial}
% \begin{proof}
%     For every $\bgMPtwo$ formula $\varphi(x)$,
%     by Theorem~\ref{thm:logic_to_gnn}
%     there exists an equivalent $\bgtrreluGNN$ $\cA_\varphi$.
%     It is easy to see that $\varphi(x)$ is finitely satisfiable if and only if $\cA_\varphi$ is satisfiable.
%     Assuming the claim, the undecidibility of the satisfibility problem for $\bgtrreluGNN$s follows from Theorem~\ref{thm:global_mptwo_undecidable}.
% \end{proof}

Using the same reduction, we also obtain undecidability for universal satisfiability:

\begin{theorem}\label{thm:global_gnn_universal_undecidable}
    The universal satisfiability problem of $\bgtrreluGNN$s is undecidable.
\end{theorem}

% \ch{it is trivial}
% \begin{proof}
%     For every simple equation system $\varepsilon$,
%     let $\Psi_\varepsilon(x)$ be the $\bgMPtwo$ formula from the claim.
%     By Theorem~\ref{thm:logic_to_gnn},
%     there exists a $\bgtrreluGNN$ $\cA_{\Psi_\varepsilon}$,
%     such that $\Psi_\varepsilon(x)$ and $\cA_{\Psi_\varepsilon}$ are equivalent.
%     We claim that $\varepsilon$ has a solution in $\bbN$ if and only if $\cA_{\Psi_\varepsilon}$ is universally satisfiable.
%     Since the solvability of simple equation systems is undecidable, so is the universal satisfibility problem of $\bgtrreluGNN$s.

%     If $\varepsilon$ has a solution in $\bbN$, 
%     by Lemma~\ref{lemma:hilbert_to_mptwo},
%     there exists a graph $\cG$, such that for every vertex $v \in V$, $\cG \models \Psi_\varepsilon(v)$.
%     Since $\Psi_\varepsilon(x)$ and $\cA_{\Psi_\varepsilon}$ are equivalent, $\tuple{\cG, v}$ also satisfies $\cA_{\Psi_\varepsilon}$.
%     Therefore $\cA_{\Psi_\varepsilon}$ is universally satisfiable with $\cG$ as the witness.
    
%     On the other hand, if $\cA_{\Psi_\varepsilon}$ is universally satisfiable,
%     then $\cA_{\Psi_\varepsilon}$ is satisfiable.
%     Thus $\Psi_\varepsilon(x)$ is finitely satisfiable.
%     By Lemma~\ref{lemma:hilbert_to_mptwo}, $\varepsilon$ has a solution in $\bbN$.
% \end{proof}

\section{$\pspace$-completeness of satisfiability for GNNs with local aggregation and $\trrelu$ activations} \label{sec:pspace}

A more realistic analysis of complexity requires stronger assumptions on the activation functions.
For now we consider only one special case, where the activation function is truncated $\relu$:
\begin{theorem}\label{thm:pspace}
    The satisfiability problem of $\bltrreluGNN$s is $\pspace$-complete, and so is $\oltrreluGNN$. 
    They are $\np$-complete when the number of layers is fixed.
\end{theorem}
We start with the $\pspace$ upper bound argument. We can show that for an arbitrary input graph, there are only exponentially many activation values, each representable with a polynomial number of bits. We also show, via an ``unravelling construction'', a common technique used in analysis of modal and guarded logics~\cite{modallogicpspace,modallogichandbook}, that a satisfying model
can be taken to be a tree of polynomial depth and branching. These two facts immediately give an elementary bound, since we could guess the tree and the activation values. We can improve to $\pspace$ by exploring a satisfying tree-like model on-the-fly: again, this is in line with the $\pspace$ algorithm for modal logic~\cite{modallogicpspace}.

The $\pspace$ lower bound is established by embedding the description logic $\alc$ into $\blMPtwo$.
$\pspace$-hardness will follow from this, since concept satisfiability problem of $\alc$ with one role is  $\pspace$-hard~\cite{alc-pspace}. The $\np$ upper bound will use the same on-the-fly algorithm as in the $\pspace$ case, just observing that
for fixed number of layers it can be implemented in $\np$. A direct encoding of $\threeSAT$ gives the lower bound.

\subsection{Exponential history-space property}

The first step is to establish a bound on the size of the numbers that can be computed by GNNs. Remember that we begin with a graph where the feature values are only binary, and in each layer we do one aggregation and truncate the result. Thus it is intuitive that we cannot build up large values in any intermediate result at any node, regardless of the size of the graphs. This property actually holds even with global aggregation, but we prove it here only for the local case, since this is the only one relevant to this proof.
We first extend the idea of the spectrum, which is defined in Section~\ref{sec:eventually_constant}.

\begin{definition} \label{def:hist}
    % [History of a node]
    For an $n$-GNN $\cA$,
    for an $n$-graph $\cG$, vertex $v$ in $\cG$, and $0 \leq \ell \leq L$,
    the \emph{$\ell$-history of $v$ in $\cG$} (w.r.t. $\cA$), 
    denoted by 
    $\hist{\ell}_{\cG}(v)$, 
    is the tuple that collects the first $(\ell+1)$ feature vectors of $v$.
    Formally, for $0 \le t \le \ell$, $\left(\hist{\ell}_{\cG}(v)\right)[t] = \feat{t}_{\cG}(v)$.
\end{definition}
When the graph $\cG$ is clear from the context, we omit it
and simply write $\hist{\ell}(v)$.

\begin{definition}%[History space of a GNN]
    For an $n$-GNN $\cA$ and $0 \le \ell \le L$,
    the \emph{$\ell$-history-space} of $\cA$,
    denoted by $\hsp{\ell}$,
    is the set $\setc{\hist{\ell}_\cG(v)}{\text{for every $n$-graph $\cG$ and vertex $v$ in $\cG$}}$.
\end{definition}

For a vertex $v$,
the $\ell$-feacture of $v$ is computed from $(\ell-1)$-features of vertices in the graph.
Similarly $\ell$-history of $v$ can be computed efficiently from $(\ell-1)$-features.
\begin{lemma} \label{lemma:comp_hist}
    For an $n$-GNN $\cA$, 
    for an $n$-graph $\cG$, vertex $v$ in $\cG$, and $0 \leq \ell \leq L$,
    the $\hist{\ell}(v)$ can be computed efficiently from $\feat{0}(v)$ and $\hist{\ell-1}(u)$ for vertices $u$ in $\cG$.
\end{lemma}
\begin{proof}
    For $x \in \set{\out, \inc}$,
    let $\bfh_x = \sum_{u \in \nbr{x}(v)} \hist{\ell-1}(u)$.
    Let $\bfh_g = \sum_{u \in V} \hist{\ell-1}(u)$.
    We can compute entries of $\hist{\ell}(v)$ bottom-up.
    For the base case, it is clear that $\hist{\ell}(v)[0] = \feat{0}(v)$.
    For $1 \le i \le \ell$,
    note that for $x \in \set{\out, \inc}$
    $\sum_{u \in \nbr{x}(v)} \feat{\ell-1}(u) = \bfh_x[\ell-1]$ and
    $\sum_{u \in V} \feat{\ell-1}(u) = \bfh_g[\ell-1]$.
    Thus, since $\hist{\ell}(v)[i] = \feat{i}(v)$, it follows that
    \begin{equation*}
        \feat{i}(v)\ =\ \actp{i}{
            \coefC{i} \feat{i-1}(v) +
            \sum_{x \in\set{\out, \inc}}
            \left(
                \coefA{i}{x}
                \bfh_x[i-1]
            \right) +
            \coefR{i} \bfh_g[i-1] +
            \coefb{i}
        }.
    \end{equation*}
\end{proof}
%\michael{Add the lemma about how to compute}
%\ch{added}

We now show that for GNNs with $\trrelu$ activation,
the size of the history-space (and thus, in particular, the size of the spectrum) is bounded by an exponential in the description of the GNN.
Recall that in Section~\ref{sec:eventually_constant}, Lemma~\ref{lemma:syntacticoverapprox},
we proved that the $\ell$-spectrum is overapproximated by the set $\spectrumover{\ell}$. We now calculate the size of this set in the case of $\trrelu$.
For a GNN with $\trrelu$ activations, 
the set $\spectrumover{\ell}$ is $\setc{\dfrac{\bfk}{\capa{\ell}}}{\bfk \in \intinterval{0}{\capa{\ell}}^\gnndim{\ell}}$, 
where $\capa{\ell}$ is the product of the least common denominator of $\spectrumover{\ell-1}$ and the least common denominator of the coefficients in $\coefC{\ell}$, $\coefA{\ell}{\out}$, $\coefA{\ell}{\inc}$, and $\coefb{\ell}$,
since the least common denominator of $\spectrumover{\ell-1}$ is just $\capa{\ell-1}$ in this case.
%\tony{Something is missing in the previous statement ``Since ...''}
%Lemma~\ref{lemma:overapprox-trrelu} follows immediately from %Lemma~\ref{lemma:syntacticoverapprox},
%since truncated $\relu$ is an eventually constant function.
Thus we get a bound on the number of elements in the spectrum:

\begin{lemma}
\label{lemma:overapprox-trrelu}
    For every $\bgtrreluGNN$ $\cA$ and $0 \le \ell \le L$,
    the number of elements in  $\spectrumover{\ell}$ is bounded by $\left(\capa{\ell}\right)^{d^{(\ell)}}$.

%\tony{I added the missing exponent $d^{(\ell)}$ here.}
%michael: thanks
\end{lemma}

For $0 \le \ell \le L$, let $\hspover{\ell}$ be the Cartesian product of $\spectrumover{t}$ for $0 \le t \le \ell$.
Let $D_\cA = \Pi_{\ell \in \intsinterval{L}} \left(\capa{\ell}\right)^{d^{(\ell)}}$. Note that $D_\cA$ is exponential in the description of the GNN.
It is obvious that the $\ell$-history-space is overapproximated by $\hspover{\ell}$.
Thus, looking at the set $\setc{\dfrac{\bfk}{\capa{\ell}}}{\bfk \in \intinterval{0}{\capa{\ell}}^\gnndim{\ell}}$, we can get a bound on both the number of elements in $\ell$-history-space, as well as the number of bits to represent each element:

%\tony{$\cM$ should be $M_{\cA}$?}
%\ch{fix definition of $D_\cA$ and $M_\cA$ }

%\tony{Lemma~\ref{lemma:smallbits} and Lemma 5.4 above needs proof, I think.
%I don't think they are difficult, but they are not trivial.}

\begin{lemma} \label{lemma:smallbits}
    For every $\bgtrreluGNN$ $\cA$ and $0 \le \ell \le L$,
    the number of elements in  $\hspover{\ell}$ is bounded by $D_\cA$.
    Furthermore, every value in $\hspover{\ell}$ is at most exponential in the size of the description of $\cA$, hence, can be encoded with polynomially many bits.
%    Furthermore, every element in $\hsp{\ell}$ can be encoded with a number of bits which is only polynomial in the description of $\cA$.
\end{lemma}

\subsection{Exponential tree model property}

The previous subsection bounded the size of individual values, independent of the input graph.
We now show that whenever a GNN is satisfiable, there is \emph{some} satisfying graph that is both reasonably small and nicely-structured. Recall that we are dealing here with directed graphs with some number of node colors.
\begin{definition}[Trees, children, and depth]
An \emph{$n$-tree} is an $n$-graph such that when we remove the direction we have a tree.
For a tree $\cT$ and vertices $v, u$ in $\cG$,
we say that $u$ is an out-child of $v$,
if $u$ is a child of $v$ and there exists an edge from $v$ to $u$.
We let $\nbrc{\out, \cG}(v)$ be the set of all out-children of $v$.

The depth of a vertex $v$ in a tree is the distance to the root.
The depth of a tree is the maximal depth of any vertex.
\end{definition}
Let $\cA$ be a GNN with $L$ layers.
Let $\cT$ be a tree with depth $L$
and $v$ be a vertex in $\cG$ with depth $\ell$.
We define $\hspc{\out, \cG}(v)$ as the set of $(L-\ell-1)$-histories of out-children of $v$.
Formally, $\hspc{\out, \cG}(v) := \setc{\hist{L-\ell-1}_\cG(u)}{u \in \nbrc{\out, \cG}(v)}$.
We define $\nbrc{\inc, \cG}(v)$, $\hspc{\inc, \cG}(v)$ for in-children analogously.
When the graph $\cG$ is clear from the context, we omit it and simply write
$\nbrc{\out}(v)$, $\hspc{\out}(v)$, $\nbrc{\inc}(v)$, and $\hspc{\inc}(v)$.

We can now state the \emph{exponential tree model property}.
\begin{theorem}\label{thm:exp_tree_model}
    There are constants $c_1, c_2 \in \bbN$, so that
    for every $n$-$\bltrreluGNN$ $\cA$ with $L$ layers,
    if $\cA$ is satisfiable,
    then there is an $n$-tree $\cT$% \michael{Should use a different symbol here and elsewhere for trees} 
    with depth at most $L$ and root $v_r$ satisfying the following conditions.
    %\tony{I add a name for each property to make it easier to refer to later.} michael: thanks
    \begin{enumerate}
        \item (tree model property) $\cA$ accepts $\tuple{\cT, v_r}$.
        \item (small degree property) For every vertex $v$ in $\cT$ of depth $\ell$,
        for $x \in \set{\out, \inc}$, the size of $\hspc{x}(v)$ is at most $\alpha_{\cA}$,
        and for each $\bfh \in \hspc{x}(v)$, the number of $x$-children of $v$ with $(L-\ell-1)$-history  $\bfh$ is at most $\beta_{\cA}$, where
        \begin{equation*}
            \begin{aligned}
                \alpha_{\cA}\ :=\ &c_1d_\cA \log\left(c_2d_\cA M_\cA\right), \\
                \beta_{\cA}\ :=\ &c_1 d_\cA \left(d_\cA M_\cA \right)^{c_2d_\cA},
            \end{aligned}
        \end{equation*}
        $d_\cA := \sum_{\ell \in \intsinterval{L}} \gnndim{\ell}$,
        and $M_\cA$ is the product of $2D_\cA$ and the maximal numerator of coefficients in $\cA$.
    \end{enumerate}
\end{theorem}

We prove the theorem in two steps.
First, we show that if a local GNN is satisfiable,
then it has a tree model.
We construct such tree model by an \emph{unravelling procedure}, a common tool for modal and description logics.
%also employe~\cite{langeverifygnn}.
%michael: I don't understand this citation -- that paper does not mention unravellings

\begin{definition}
    For a graph $\cG$, 
    we say a sequence $\frP = v_0R_1v_1R_2v_2 \cdots v_\ell$ is a \emph{path} in $\cG$
    if $\frP$ satisfies the following conditions.
    \begin{itemize}
        \item For every $0 \le t \le \ell$, 
        $v_t \in V$.

        \item For every $1 \le t \le \ell$, 
        $R_t \in \set{\rightarrow, \leftarrow}$.

        \item For every $1 \le t \le \ell$, 
        if $R_t$ is $\rightarrow$,
        then there is an edge from $v_{t-1}$ to $v_{t}$.
        Otherwise, if $R_t$ is $\leftarrow$,
        then there is an edge from $v_{t}$ to $v_{t-1}$.
    \end{itemize}
    We say that $\frP$ is \emph{proper} if it has no contiguous subsequence of the form
  $v \leftarrow u \rightarrow v$ or $v \rightarrow u \leftarrow v$.
\end{definition}

\begin{definition}
    For a path $\frP = v_0R_1v_1R_2v_2 \cdots v_\ell$,
    $\ell$ is called the \emph{length} of $\frP$;
    $v_0$ is called the \emph{source} of $\frP$;
    $v_\ell$ is called the \emph{destination} of $\frP$.
\end{definition}

\begin{figure}[t]
    \centering
    \begin{tikzpicture}[thick, main/.style = {draw, circle, minimum size=25}, scale=0.8]
        \node[main] (av1) at (-9,-1)  {$v_1$};
        \node[main] (av2) at (-9,-3) {$v_2$};
        \node[main] (av3) at (-11,-1) {$v_3$};
        \node[main] (av4) at ( -7,-1) {$v_4$};
        
        \path [->] (av1) edge (av2);
        \path [->] (av1) edge (av3);
        \path [->] (av3) edge (av2);
        \path [->] (av1) edge [bend left=20] (av4);
        \path [->] (av4) edge [bend left=20] (av1);

        \node at (-9,-4.5) {$\cG_0$};

        %% the unravelling tree
        \node[main] (v1)  at (0,0)     {$v_1$};
        \node[main] (v2)  at (-3,-1.5) {$v_2$};
        \node[main] (v3)  at (-1,-1.5) {$v_3$};
        \node[main] (v41) at ( 1,-1.5) {$v_4$};
        \node[main] (v42) at ( 3,-1.5) {$v_4$};
        \node[main] (v21) at (-1,-3)   {$v_2$};
        \node[main] (v31) at (-3,-3)   {$v_3$};
        \node[main] (v12) at (-1,-4.5) {$v_1$};
        \node[main] (v13) at (-3,-4.5) {$v_1$};
        
        \path [->] (v1)  edge (v2);
        \path [->] (v1)  edge (v3);
        \path [->] (v1)  edge (v41);
        \path [->] (v42) edge (v1);
        \path [->] (v31) edge (v2);
        \path [->] (v3)  edge (v21);
        \path [->] (v12) edge (v21);
        \path [->] (v13) edge (v31);

        \node at (0,-6) {$\urv{3}{\cG_0, v_1}$};
    \end{tikzpicture}    

    \caption{Example of a $3$-unravelling tree of $\cG_0$ on $v_1$.}
    \label{fig:unravelling}
%tony: it is possible (though I am not sure) that the minipage/subcaption messes up with the labelling, so I deleted them. The figure number is ok now.
\end{figure}

For example, considering the graph $\cG_0$ in Fig.~\ref{fig:unravelling},
%\tony{the figure number here is different from the one in the figure.}
$v_3 \leftarrow v_1 \rightarrow v_4 \rightarrow v_1$ is a proper path in $\cG_0$ with length 3, source $v_3$, and destination $v_1$.
However $v_3 \leftarrow v_1 \rightarrow v_4 \leftarrow v_1$ is a path in $\cG_0$ but not a proper path.

\begin{definition}
    For an $n$-graph $\cG$ and a vertex $v$ in $\cG$,
    the \emph{$\ell$-unravelling tree} of $\cG$ on $v$,
    denoted by $\urv{\ell}{\cG, v}$, 
    is the $n$-graph defined as follows.
    \begin{itemize}
        \item The set of vertices of $\urv{\ell}{\cG, v}$ is the collection of proper paths in $\cG$
        with source $v$ and length at most $\ell$,
        and there is exactly one vertex in $\urv{\ell}{\cG, v}$
        with length $0$, called the root.

        \item For every vertex $\frP$ in $\urv{\ell}{\cG, v}$
        with non-zero length, suppose 
        $\frP = \frP'Ru$,
        where $\frP'$ is also a vertex in $\urv{\ell}{\cG, v}$.
        If $R$ is $\rightarrow$,
        then there is an edge from $\frP'$ to $\frP$ in $\urv{\ell}{\cG, v}$.
        Otherwise, if $R$ is $\leftarrow$,
        then there is an edge from $\frP$ to $\frP'$ in $\urv{\ell}{\cG, v}$.

        \item For every vertex $\frP$ in $\urv{\ell}{\cG, v}$,
        the label of $\frP$ is exactly the label of its destination.

        %\item The root of $\urv{\ell}{\cG, v}$ is the only vertex whose length is 0.
        %\tony{Should it be: ``The root is the path  of length $0$ with source $v$''?}
    \end{itemize}
\end{definition}
Note that we will use symbols like $\frP$ to refer to paths, even when they
are considered as vertices in the unravelling.
It is obvious that $\urv{\ell}{\cG, v}$ is a tree, and
the depth is at most $L$.
% \tony{to do: ``depth'' and ``height'' to be made consistent throughout the paper.}
% \ch{fixed. chaneged height to depth}
For every vertex $\frP$ in $\urv{\ell}{\cG, v}$
the depth of $\frP$ in $\urv{\ell}{\cG, v}$ is the length of $\frP$.

\begin{lemma} \label{lemma:unravelling}
    For every $n$-$\lGNN$ $\cA$,
    for every $n$-graph $\cG$ and vertex $v_r$ in $\cG$,
    $\hist{L}_{\urv{L}{\cG, v_r}}\left(\frP_r\right) = 
     \hist{L}_{\cG}(v_r)$,
    where $\frP_r$ is the root of $\urv{L}{\cG, v_r}$.
\end{lemma}

\begin{proof}
    We will prove the following stronger property:
    for $0 \le \ell \le L$,
    for every vertex $\frP$ in $\urv{L}{\cG, v_r}$ with destination $v$,
    if the depth of $\frP$ is less than or equal to $L-\ell$,
    then $\feat{\ell}_{\urv{L}{\cG, v_r}}(\frP) = \feat{\ell}_\cG(v)$.
    Note that the depth of the root vertex $\frP_r$ is $0$.
    So the lemma follows from the property directly.

    We prove the property by induction on $\ell$.
    For the base case $\ell = 0$,
    $\feat{0}_{\urv{L}{\cG, v_r}}(\frP)$ is just the label of $\frP$.
    By the definition of the unravelling graph,
    the label of $\frP$ is the label of its destination.
    
    For the inductive step $1 \le \ell \le L$,
    % for every vertex $\frP$ in $\urv{L}{\cG, v_r}$ with destination $v$,
    since $L - \ell < L - (\ell - 1)$,
    by the induction hypothesis,
    $\feat{\ell-1}_{\urv{L}{\cG, v_r}}(\frP) = \feat{\ell-1}_{\cG}(v)$.
    Note that the GNN we considered is local, that is,
    the value of $\feat{\ell}_{\cG}(v)$ only depends on
    $\feat{\ell-1}_{\cG}(v)$ and
    summations of $\feat{\ell-1}_{\cG}(u)$ for its outgoing and incoming neighbors $u$. %michael: avoid "respectively"
    Thus, it is sufficient to show that
    \begin{equation*}
        \begin{aligned}
            \sum_{\frP' \in \nbr{\out, {\urv{L}{\cG, v_r}}}(\frP)}
            \feat{\ell-1}_{\urv{L}{\cG, v_r}}(\frP')
            \ =\ &
            \sum_{u \in \nbr{\out, \cG}(u)}
            \feat{\ell-1}_{\cG}(u) \\
            \sum_{\frP' \in \nbr{\inc, {\urv{L}{\cG, v_r}}}(\frP)}
            \feat{\ell-1}_{\urv{L}{\cG, v_r}}(\frP')
            \ =\ &
            \sum_{u \in \nbr{\inc, \cG}(u)}
            \feat{\ell-1}_{\cG}(u).
        \end{aligned}
    \end{equation*}

    Suppose that $\frP = \frP_r$.
    For every $u \in \nbr{\out, \cG}(v_r)$,
    by the definition of the unravelling tree,
    $\frP' = v_r\rightarrow u$ is a vertex in $\urv{L}{\cG, v_r}$.
    Furthermore, there is an edge from $\frP_r$ to $\frP'$.
    Note that the depth of $\frP'$ is $1 \le L-(\ell-1)$.
    Thus, by the induction hypothesis,
    for every $u \in \nbr{\out, \cG}(v_r)$,
%    \michael{for every $u \in \nbr{\out, \cG}(v_r)$?}
  %  \ch{rewritten}
    $\feat{\ell-1}_{\urv{L}{\cG, v_r}}(\frP')
    = \feat{\ell-1}_{\cG}(u)$,
    which implies that the summation of $(\ell-1)$-features of outgoing neighbors of $\frP_r$ is the same as $v_r$.
    We can treat the incoming neighbors analogously.
    
    If $\frP \neq \frP_r$, 
    suppose that $\frP$ decomposes into a concatenation of
    a path $\frP_p$ with destination $v_p$
    and a vertex $v$.
    Without loss of generality, we assume that $\frP = \frP_p \leftarrow v$.
    For every $u \in \nbr{\out, \cG}(v)$,
    if $u \neq v_p$,
    then $\frP' = \frP \rightarrow u$ is a vertex in $\urv{L}{\cG, v_r}$.
    By an argument similar to the previous case, 
    we can show that
    $\feat{\ell-1}_{\urv{L}{\cG, v_r}}(\frP')
    = \feat{\ell-1}_{\cG}(u)$.
    Note that $\frP \rightarrow v_p$ is not a vertex in $\cG^L_{v_r}$,
    because it is not a proper path.
    However,$\frP_p$ is a vertex in $\urv{L}{\cG, v_r}$ and
    there is an edge from $\frP$ to $\frP_p$.
    Thus $\frP_p$ is also a outgoing neighbor of $\frP$.
    Since the depth of $\frP_p$ is less than $\frP$,
    the depth of $\frP_p$ is less than $L-(\ell-1)$.
    Hence
    $\feat{\ell-1}_{\urv{L}{\cG, v_r}}(\frP_p)
    = \feat{\ell-1}_{\cG}(v_p)$.
    Thus, it holds that
    the summation of $(\ell-1)$-features of outgoing neighbors of $\frP$ is the same as $v$.
    We can treat the incoming neighbors analogously without considering $\frP_p$,
    since $\frP_p$ is not an incoming neighbors of $\frP$.
\end{proof}

Next, we show that for every tree model,
we can apply some surgery and obtain another tree model whose size is bounded.

\begin{definition} \label{def:characteristicsystem}
    For every $\bltrreluGNN$ $\cA$,
    for every $n$-tree $\cG$ and $v$ in $\cG$ with depth $\ell$,
    the \emph{characteristic system of $v$} (w.r.t. $\cA$),
    denoted by $\cQ_v$,
    is a linear constraint system with variables
    $\set{z_{x, \bfh}}_{\substack{x \in \set{\out, \inc}\\ \bfh \in \hspc{x}(v)}}$ defined as follows.
    For $1 \le t \le L-\ell$ and $1 \le i \le \gnndim{t}$,
    \begin{equation*}
        D_\cA \left(
            \coefC{t} \feat{t-1}(v) + \featp{t}(v) + 
            \sum_{x \in \set{\out, \inc}} \coefA{t}{x}
            \sum_{\bfh \in \hspc{x}(v)} z_{x, \bfh} \cdot \bfh[t-1] +
            \coefb{t}
        \right)_i\ \circledast_{t, i}\ D_\cA \feat{t}_i(v)
    \end{equation*}
    is a constraint in $\cQ_v$,
    where
    \begin{equation*}
        \featp{t}(v)\ =\ 
        \begin{cases}
            0,&                               \text{if $v = v_r$} \\
            \coefA{t}{\out} \feat{t-1}(v_p),& \text{if there is an edge from $v$ to $v$'s parent $v_p$} \\
            \coefA{t}{\inc} \feat{t-1}(v_p),& \text{if there is an edge from $v$'s parent $v_p$ to $v$}.
        \end{cases}
    \end{equation*}
    If $\feat{t}_i(v) = 1$, then $\circledast_{t, i}$ is $\ge$;
    if $\feat{t}_i(v) = 0$, then $\circledast_{t, i}$ is $\le$;
    otherwise, $\circledast_{t, i}$ is $=$.

A solution to the system $\cQ_v$ where each variable $z_{x, \bfh}$ is assigned  the value $a_{x, \bfh}\in\bbN$ is denoted by 
    $\set{z_{x, \bfh} \gets a_{x, \bfh}}_{\substack{x \in \set{\out, \inc}\\\bfh \in \hspc{x}(v)}}$.
\end{definition}
%\tony{Is the multiplication by $M_{\cA}$ to ensure that the system is a ``proper'' ILP instances?}

We remark that the multiplication by the factor $D_{\cA}$ is to ensure the coefficients in the constraint are all integers.
Note that there are at most $d_\cA = \sum_{\ell \in \intsinterval{L}} \gnndim{\ell}$ constraints.
Because $\hspc{x}(v) \subseteq \hspover{L-\ell}$,
by Lemma~\ref{lemma:smallbits},
$\abs{\hspc{x}(v)} \le D_\cA$.
Thus, there are at most $2D_\cA$ variables in the system $\cQ_v$.
Since $\feat{t}_i(v)$ is the output of $\trrelu$, $0 \le \feat{t}_i(v) \le 1$.
Therefore, the absolute value of the coefficients in the system $\cQ_v$ is bounded by
the product of $2D_\cA$ and the maximal numerator of coefficients in $\cA$:
that is, the values are bounded by $M_\cA$.
%\michael{I don't follow "Thus, by inspection...". It either follows from the previous, or one needs to inspect something (what?)}
%\ch{rewritten}

Next, we will show that the number of children 
of each node in the tree can be bounded.
We will use the following lemma, which stems from~\cite{caratheodory-integer,papa-ilp}, see also \cite[Corollary 2.2]{twovarpres}.

%To reduce the number of children in the tree, we will use the following lemma, which stems from~\cite{caratheodory-integer,papa-ilp}, see also \cite[Corollary 2.2]{twovarpres}.

\begin{lemma}\label{lemma:ilp}
    There are constants $c_1, c_2 \in \bbN$ such that for every linear constraint system $\cQ$,
    if $\cQ$ admits a solution in $\bbN$, then it admits a solution in $\bbN$ in which
    the number of variables assigned  non-zero values is at most $c_1t\log\left(c_2tM\right)$ and
    every variable is assigned  a value at most $c_1t(tM)^{c_2t}$,
    where $t$ is the number of constraints in $\cQ$ and $M$ is the maximal absolute value of the coefficient in $\cQ$.
\end{lemma}

We will use Lemma~\ref{lemma:ilp} above to show that
every satisfiable $\bltrreluGNN$ is satisfiable by a small tree model in which the values of the feature vector of each node and the number of children of each node is at most exponential in the length of the description of the GNN, as stated in the following lemma.

\begin{lemma}\label{lemma:charsys_solvable}
  Let $\cA$ be a $\bltrreluGNN$, $\cG$ an $n$-tree and $v$ in $\cG$ with depth $\ell$, Let 
    $\cQ_v$ be the characteristic system of equations of $v$ with respect to $\cA$.
    Then $\cQ_v$  has a solution in $\bbN$ such that
    the number of variables assigned to non-zero values is at most $\alpha_\cA$,
    and every variable is assigned to a value
    bounded by $\beta_\cA$,
    where $\alpha_\cA$ and $\beta_\cA$ are as defined in Theorem~\ref{thm:pspace}.
\end{lemma}
    
\begin{proof}
    We first prove that $\set{z_{x, \bfh} \gets a_{x, \bfh}}_{\substack{x \in \set{\out, \inc}\\\bfh \in \hspc{x}(v)}}$ is a solution of the characteristic system $\cQ_v$,
    where $a_{x, \bfh}$ is the number of $v$'s $x$-children whose $(L-\ell-1)$-history is $\bfh$.
    
    Without loss of generality,
    suppose $v$ has a parent $v_p$ and there exists an edge from $v$ to $v_p$.
The other two cases, i.e., when the edge is oriented from $v_p$ to $v$ or when the $v$ is the root node, can be treated analogously.

Note that for $1 \le t \le L-\ell$,       
    \begin{equation*}
        \begin{aligned}
            \sum_{u \in \nbr{\out}(v)} \feat{t-1}(u)
            \ =\ &\feat{t-1}(v_p) + 
            \sum_{u \in \nbrc{\out}(v)} \feat{t-1}(u) \\
            \ =\ &\feat{t-1}(v_p) + 
            \sum_{u \in \nbrc{\out}(v)} \left(\hist{L-\ell-1}(u)\right)[t-1] \\
            \ =\ &\feat{t-1}(v_p) + 
            \sum_{\bfh \in \hspc{\out}(v)} a_{\out, \bfh} \cdot \bfh[t-1] \\
            \sum_{u \in \nbr{\inc}(v)} \feat{t-1}(u)
            \ =\ &\sum_{\bfh \in \hspc{\inc}(v)} a_{\inc, \bfh} \cdot \bfh[t-1].
        \end{aligned}
    \end{equation*}
    Combining with the definition of $\feat{t}(v)$,
    \begin{equation*}
        \feat{t}(v) =
        \trrelup{
            \coefC{t} \feat{t-1}(v) + \featp{t}(v) + 
            \sum_{x \in \set{\out, \inc}} \coefA{t}{x}
            \sum_{\bfh \in \hspc{x}(v)} a_{x, \bfh} \cdot \bfh[t-1] +
        \coefb{t}}.
    \end{equation*}
    By the definition of $\trrelu$, we can verify that 
    $\set{z_{x, h} \gets a_{x, \bfh}}_{\substack{x \in \set{\out, \inc}\\\bfh \in \hspc{x}(v)}}$ is a solution of the characteristic system $\cQ_v$.

    Recall that there are $d_\cA = \sum_{\ell \in \intsinterval{L}} \gnndim{\ell}$ constraints in $\cQ_v$ and
    the absolute value of each coefficient is bounded by $M_\cA$.
    By Lemma~\ref{lemma:ilp}, $\cQ_v$ has a \emph{small} solution 
    $\set{z_{x, \bfh} \gets a'_{x, \bfh}}_{\substack{x \in \set{\out, \inc}\\\bfh \in \hspc{x}(v)}}$: that is, a solution
    where the number of variables assigned to non-zero values is at most $\alpha_\cA$,
    and every variable is assigned to a value
    bounded by $\beta_\cA$.
    Thus $a'_{x, \bfh}$ is the desired solution of the lemma.
\end{proof}

Finally, we are ready to prove Theorem~\ref{thm:exp_tree_model}. 
\begin{proof}
%\tony{I rewrote the proof here to address Michael's questions. The previous proof is clearly marked below.}
The constants $c_1,c_2$ are the same constants as in Lemma~\ref{lemma:ilp}.
Let $\cA$ be a satisfiable $n$-$\bltrreluGNN$ with $L$ layers
which satisfied by $\tuple{\cG_0,v}$.
Let us denote the $L$-unravelling tree $\urv{L}{\cG_0, v}$ by $\cT$ and the root of $\cT$ by $v_r$.
By Lemma~\ref{lemma:unravelling}, $\cA$ also accept $\tuple{\cT,v_r}$.
%\tony{It is now mentioned explicitly that we take the unravelling tree, which by definition does not contain self-loop.}
This proves the tree model property as required by the theorem.
We will show how to convert it into another tree $\cT'$ that satisfies the small degree property in the theorem. That is, we show that
for every vertex $v$ in $\cT'$,
for $x \in \set{\out, \inc}$, the size of $\hspc{x}(v)$ is at most $\alpha_{\cA}$,
and for each $\bfh \in \hspc{x}(v)$, the number of $x$-children of $v$ with $(L-\ell-1)$-history $\bfh$ is at most $\beta_{\cA}$, where
\begin{equation*}
    \begin{aligned}
        \alpha_{\cA}\ :=\ &c_1d_\cA \log\left(c_2d_\cA M_\cA\right), \\
        \beta_{\cA}\ :=\ &c_1 d_\cA \left(d_\cA M_\cA \right)^{c_2d_\cA},
    \end{aligned}
\end{equation*}
$d_\cA := \sum_{\ell \in \intsinterval{L}} \gnndim{\ell}$,
and $M_\cA$ is the product of $2D_\cA$ and the maximal numerator of coefficients in $\cA$.
%\michael{A bit annoying that we have to scroll several pages back to find ``the second property of the theorem'' -- maybe at least clarify that it has to do with a bound on the outdegree?}
%\tony{I added a name to the second property of Theorem 5.6 and made a reference there.
%See the previous statement ``.. another tree $(\cT',v_r)$ that satisfies the small degree property...''}

%For every satisfiable $n$-$\bltrreluGNN$ $\cA$,
%by Lemma~\ref{lemma:unravelling},
%$\cA$ has an $n$-tree model $\cG$ with root $v_r$ and height at most $L$
%such that $\cA$ accepts $\tuple{\cG, v_r}$.

To this end,
we say that a vertex $u$ in $\cT$ is a \emph{bad vertex}
if $u$ does not satisfy the small degree property of the theorem, i.e., 
there is $x \in \set{\out, \inc}$ such that 
at least one of the following conditions holds.
\begin{itemize}
    \item
    The size of $\hspc{x}(u)$ is strictly bigger than $\alpha_{\cA}$.
    \item 
    There is $\bfh \in \hspc{x}(u)$ where the number of $x$-children of $u$ with $(L-\ell-1)$-history  $\bfh$ is strictly bigger than $\beta_{\cA}$.
\end{itemize}
If there is no bad vertex, then $\tuple{\cT, v_r}$ is the desired $n$-tree model of the theorem.

Suppose there is a bad vertex in $\cT$.
Let $v$ be a bad vertex with the maximal depth $\ell$ in $\cT$.
%\michael{maybe you mean ''let $v$ be a bad vertex of minimal height $\ell$''? (here and below). I do not see why there would be a unique vertex}
%\tony{This has been rephrased. See the previous statement.}
Let $a_{x, \bfh}$ be the number of $v$'s $x$-children in $\cT$ whose $(L-\ell-1)$-history is $\bfh$.

Consider the characteristic system $\cQ_v$.
Obviously, $\set{z_{x, \bfh} \gets a_{x, \bfh}}_{\substack{x \in \set{\out, \inc}\\\bfh \in \hspc{x}(v)}}$ is a solution to the system $\cQ_v$.
By Lemma~\ref{lemma:charsys_solvable},
$\cQ_v$ has a \emph{small} solution $\set{z_{x, \bfh} \gets a'_{x, \bfh}}_{\substack{x \in \set{\out, \inc}\\\bfh \in \hspc{x}(v)}}$
where the number of variables assigned to non-zero values is at most $\alpha_\cA$,
and every value $a'_{x, \bfh}$ is 
at most $\beta_\cA$.
Without loss of generality, we may assume that  $a'_{x,\bfh}=0$ whenever $a_{x,\bfh}=0$,
for every $x \in \set{\out, \inc}$
and $\bfh \in \hspc{x}(v)$.

Let $\cT'$ be the $n$-tree obtained by applying the following modification to $\cT$.
For every $x \in \set{\out, \inc}$ and $\bfh \in \hspc{x}(v)$,
let $u_\bfh$ be a $x$-child of $v$ whose $(L-\ell-1)$-history is $\bfh$.
We first remove all children of $v$.
Then for every $x \in \set{\out, \inc}$ and $\bfh \in \hspc{x}(v)$,
we  duplicate $a'_{x, \bfh}$ copies of the subtree induced by $u_\bfh$
and attach each subtree copy as an $x$-child of $v$.
%\michael{Here, we are attaching each subtree copy as an $x$-child of $v$?}
%\tony{I added the phrase ``.. attach each subtree copy..''}
%    \michael{Not clear to me what happens with self-loops, which we allow: the text says that we ``remove all children of $v$'', so if there is a self-loop we remove $v$ itself?}
%\tony{I made it explicit now that the tree is the unravelling tree which by definition does not contain self-loop. See comment above.}
If $x$ is $\out$, we add an edge from $v$ to the root of the duplicated subtrees.
Otherwise, if $x$ is $\inc$, we add an edge from the root of the duplicated subtrees to $v$.

Obviously the number of $v$'s $x$-children in $\cT'$ with $(L-\ell-1)$-history $\bfh$ is $a'_{x, \bfh}$,
which is bounded by $\beta_{\cA}$.
Moreover, the size $\hspc{x}(v)$ is bounded
by the number of variables in $\cQ_v$ that are assigned non-zero values, which is bounded by $\alpha_{\cA}$.
Thus, the vertex $v$ is no longer a bad vertex in $\cT'$.
Furthermore, since $\set{z_{x, \bfh} \gets a'_{x, \bfh}}_{\substack{x \in \set{\out, \inc}\\\bfh \in \hspc{x}(v)}}$ is the solution to the system $\cQ_v$,
the $(L-\ell)$-history of $v$ in $\cT'$ is the same 
as in $\cT$.

Finally, note that $v$ is a bad vertex with the maximal depth in $\cT$. 
Hence every descendant of $v$ in $\cT$ is not a bad vertex in $\cT$.
Since the only difference between $\cT$ and $\cT'$ is that
we make duplicates of the subtrees rooted at the children of $v$,
every descendant of $v$ in $\cT'$ is also not a bad vertex in $\cT'$.
%\michael{why also not in $\cT'$?}
%\tony{I added the explanation ``Since the only difference between ...''}
Therefore the number of bad vertices in $\cT'$ is one less than $\cT$.
We can repeat the procedure until there are no more bad vertices and the desired $n$-tree model is obtained.

\end{proof}

\subsection{$\pspace$-completeness for 
the satisfiability problem of $\bltrreluGNN$s}

Now we are ready to prove Theorem~\ref{thm:pspace}.
Intuitively, in Theorem~\ref{thm:exp_tree_model},
we showed that every satisfiable $\bltrreluGNN$ has an exponential size tree model.
The theorem did not say anything about the size of the numbers in features. But
Lemma~\ref{lemma:smallbits} tells us that for every graph the size of computed features is not very large.

Since the size of the model may be exponential, the na\"ive algorithm, which guesses the whole model and checks it, takes nondeterministic exponential time.
However, we can reduce to $\pspace$ using the same approach as in the $\pspace$ bound for modal logic~\cite{modallogicpspace}: in order to check the validity of the tree model, it is sufficient to check the validity of each vertex locally. 
That is, we only need to compute the history from its parent and children and check this value.
Thus the following Algorithm~\ref{alg:pspace} decides the satisfiability problem for $\bltrreluGNN$s by guessing children of a vertex and checking one by procedure \Call{Check}{$\cA$, $\ell$, $x$, $\bfh_p$, $\bfh$} at a time.
% The procedure \Call{Check}{$\cA$, $\ell$, $x$, $\bfh_p$, $\bfh$} 
% check whether there exists a vertex $\ell$ with history $\bfh$ $x$-parebe
For the base case in line 7, if the depth of the vertex $v$ is $L$, meaning that $v$ is a leaf vertex of the tree,
then the history $\bfh$ is always satisfiable by assigning it to the label of $v$.
Otherwise, the procedure guesses the label $e$ of $v$ in line 9,
the lower-level history of its children in line 10, 
and the number of children of each history in line 11.
While the number of children may be exponential, most of them share the same history.
More precisely, it suffices to guess a polynomial number of distinct histories and a polynomial number of exponential counts, the sizes of which remain polynomial.
After that, the procedure checks the validity of each child's history by recursively calling itself in lines 12 to 16.
Finally, it verifies the correctness of the guess by computing the history and comparing it to $\bfh$ in line 17, as described in Lemma~\ref{lemma:comp_hist}.
Since the depth of the recursive calls is at most $L$,
Algorithm~\ref{alg:pspace} is a nondeterministic polynomial-space algorithm.
The correctness follows from Theorem~\ref{thm:exp_tree_model} and Lemma~\ref{lemma:smallbits}.
%\michael{More explanation of the algorithm, and how the comments here in text correspond to line numbers in the algorithm, is needed}
%\ch{added}

\begin{algorithm}
\caption{Algorithm for the satisfiability problem of $\bltrreluGNN$s} \label{alg:pspace}
\begin{algorithmic}[1]

\Procedure{SAT}{$\cA$}\Comment{$\bltrreluGNN$ $\cA$}
    \State Guess a over-approximated $L$-history $\bfh_r$ with $(\bfh_r[L])_1 \ge 1/2$.
    \State \Return \Call{Check}{$\cA$, 0, root, $\epsilon$, $\bfh_r$}
\EndProcedure

\Procedure{Check}{$\cA$, $\ell$, $x$, $\bfh_p$, $\bfh$}
            \Comment{$\bltrreluGNN$ $\cA$, depth $\ell \in \intinterval{0}{L}$,}
    \If{$\ell = L$}
            \Comment{direction $x \in \set{\out, \inc, \text{root}}$,}
        \State \Return True
            \Comment{parent's $(L-\ell+1)$-history $\bfh_p$,}
    \Else
            \Comment{desired $(L-\ell)$-history $\bfh$.}
            %\ch{add comment}
        \State Guess initial feature $e$ from $\set{0, 1}^{\gnndim{0}}$.
        \State Guess $\alpha_\cA$ over-approximated $(L-\ell-1)$-histories $\bfh_{\out, i}$ and $\bfh_{\inc, i}$ from $\hspover{L-\ell-1}$.
        \State Guess $\alpha_\cA$ numbers $a_{\out, i}$ and $a_{\inc, i}$ from $\intinterval{0}{\beta_\cA}$.
        \For{$x' \in \set{\out, \inc}$ and $1 \le i \le \alpha_\cA$}
            \If{not \Call{Check}{$\cA$, $\ell+1$, $x'$, $\bfh$, $\bfh_{x, i}$}}
                \State \Return False
            \EndIf
        \EndFor
        \State Use Lemma~\ref{lemma:comp_hist} to compute $\bfh'$ from $e$, $x$, $\bfh_p$, $\bfh_{\out, i}$, $\bfh_{\inc, i}$, $a_{\out, i}$, $a_{\inc, i}$.
       % \michael{some comment about how?}
       % \ch{added}
        \If{$\bfh = \bfh'$}
            \State \Return True
        \Else
            \State \Return False
        \EndIf
    \EndIf
\EndProcedure
\end{algorithmic}
\end{algorithm}

The lower bound is established by embedding the description logic $\alc$ into $\olMPtwo$.
Since the concept satisfiability problem of $\alc$ with one role is  $\pspace$-hard~\cite{alc-pspace},
it will follow from the embedding that the finite satisfiability problem of $\olMPtwo$ is also $\pspace$-hard.
Since the reduction from $\olMPtwo$ to $\oltrreluGNN$ mentioned in Theorem~\ref{thm:gnn_to_logic} is polynomial,
the satisfiability problem of $\oltrreluGNN$ is also $\pspace$-hard, and so it $\bltrreluGNN$.

\begin{lemma}\label{lemma:alc}
    There exists a polynomial time translation $\pi_x$ from $\alc$ concepts with one role $R$ to $\olMPtwo$ formulas such that the $\alc$ concept $C$ is satisfiable if and only if the $\olMPtwo$ formula $\pi_x(C)$ is finitely satisfiable.
\end{lemma}

\begin{proof}
    We will define $\pi_x$ and $\pi_y$, which are the standard translation from $\alc$ concepts to first-order logic formulas, with some slight modification to quantifiers to fit our logic.
    It is routine to check that the $\alc$ concept $C$ is satisfiable if and only if $\pi_x(C)$ is finitely satisfiable.
    \begin{equation*}
        \begin{aligned}
            \pi_x(A)\ =\ &A(x)&
            \pi_y(A)\ =\ &A(y)
            \\
            \pi_x(\neg C)\ =\ &\neg \pi_x(C)&
            \pi_y(\neg C)\ =\ &\neg \pi_y(C)
            \\
            \pi_x(C \sqcap D)\ =\ &\pi_x(C) \land \pi_x(D)&
            \pi_y(C \sqcap D)\ =\ &\pi_y(C) \land \pi_y(D)
            \\
            \pi_x(C \sqcup D)\ =\ &\pi_x(C) \lor \pi_x(D)&
            \pi_y(C \sqcup D)\ =\ &\pi_y(C) \lor \pi_y(D)
            \\
            \pi_x(\exists R.C)\ =\ &\presby{E(x, y) \land \pi_y(C)} \ge 1&
            \pi_y(\exists R.C)\ =\ &\presbx{E(y, x) \land \pi_x(C)} \ge 1
            \\
            \pi_x(\forall R.C)\ =\ &\presby{E(x, y) \land \neg \pi_y(C)} = 0&
            \pi_y(\forall R.C)\ =\ &\presbx{E(y, x) \land \neg \pi_x(C)} = 0
            \\
        \end{aligned}
    \end{equation*}
\end{proof}

% A similar approach holds for the satisfiability problem over undirected graphs.

\subsection{$\np$-completeness for the satisfiability problem of fixed layer $\bltrreluGNN$s}

Let us consider the run time of Algorithm~\ref{alg:pspace}.
The procedure {\sc Check} calls itself $2\alpha_\cA$ times.
The depth of recursion is $L$.
It takes only polynomial time to compute the history from the parent and children.
Hence the runtime of the algorithm is proportional to
\begin{equation*}
    1 + 2\alpha_\cA + \left(2\alpha_\cA\right)^2 + \cdots + \left(2\alpha_\cA\right)^L
    \ =\ 
    \frac{\left(2\alpha_\cA\right)^{L+1}-1}{2\alpha_\cA-1}.
\end{equation*}
Note that $\alpha_\cA$ is polynomial in the description of $\cA$. Thus for fixed layer $\bltrreluGNN$ $\cA$, Algorithm~\ref{alg:pspace} only takes nondeterministic polynomial time.

We show that when the number of layers is fixed, there exists a polynominal time reduction from $\threeSAT$
to the satisfiability problem of fixed layer $\oltrreluGNN$.
Since $\threeSAT$ is $\np$-hard, so is the satisfiability problem of fixed layer $\oltrreluGNN$ and $\bltrreluGNN$.

%\ch{We can simulate FNN by GNN. The satisfiability problem of FNN is already NP-hard}
\begin{lemma}
    There exists a polynomial time reduction from 3-CNF formulas $\varphi$ to 2-layer $\oltrreluGNN$ $\cA_\varphi$ such that $\varphi$ is satisfiable if and only if $\cA_\varphi$ is satisfiable.
\end{lemma}

\begin{proof}
    Let $\varphi$ be a 3-CNF formula with $n$ propositional variables and $m$ clauses.
    \begin{equation*}
        \varphi\ :=\ 
        (\ell_{11} \lor \ell_{12} \lor \ell_{13}) \land
        (\ell_{21} \lor \ell_{22} \lor \ell_{23}) \land
        \cdots \land 
        (\ell_{m1} \lor \ell_{m2} \lor \ell_{m3})
    \end{equation*}

    We define the $2$-layer $\oltrreluGNN$ $\cA_\varphi$ as follows.
    The dimensions $\gnndim{0} = n$, $\gnndim{1} = m$, and $\gnndim{2} = 1$.
    The coefficient matrices and bias vectors are chosen so that the features of $\cA_\varphi$ satisfies the following conditions.
    For the first layer,
    for $1 \le t \le m$ and $1 \le j \le 3$, let
    \begin{equation*}
        \begin{aligned}
            g(t, j)\ :=\ 
            \begin{cases}
                \feat{0}_s(v), &\text{if $\ell_{tj} = x_s$}\\
                1-\feat{0}_s(v), &\text{if $\ell_{tj} = \neg x_s$},
            \end{cases} \\
        \end{aligned}
    \end{equation*}
    and $\feat{1}_t(v) = \trrelup{g(t, 1) + g(t, 2) + g(t, 3)}$.
    For the second layer,
    $\feat{2}_1(v) = \trrelup{\sum_{t \in \intsinterval{m}} \feat{1}_t(v) - m + 1}$.

    Suppose $\varphi$ is satisfied by the assignment $\set{x_i \gets a_i}_{i \in \intsinterval{n}}$,
    we define the $n$-graph $\cG$ as follows.
    The graph $\cG$ has only one vertex $v$.
    For $1 \le i \le n$, $\feat{0}_i(v) = a_i$.
    It is not difficult to check that for $1 \le t \le m$,
    for $1 \le j \le 3$,
    if $\ell_{tj}$ is satisfied by the assignment,
    then $g(t, j) = 1$.
    Otherwise, $g(t, j) = 0$.
    For the first layer, for $1 \le t \le m$,
    since the $t^{th}$ clause is satisfied,
    there exists some $1 \le j \le 3$,
    such that $\ell_{tj}$ is satisfied,
    which implies that $\feat{1}_t(v) = 1$.
    For the second layer,
    note that all clauses are satisfied.
    Therefore $\sum_{t \in \intsinterval{m}} \feat{1}_t(v) = m$,
    which implies that $\feat{2}_1(v) = 1$.
    Thus $\cA_\varphi$ is satisfied by $\tuple{\cG, v}$.
    
    Suppose $\cA_\varphi$ is satisfied by $\tuple{\cG, v}$.
    We claim that $\set{x_i \gets \feat{0}_{i}(v)}_{i \in \intsinterval{n}}$ is a satisfying assignment for $\varphi$.
    Note that for $1 \le t \le m$, $\feat{1}_t(v)$ is either 0 or 1.
    Thus if $\feat{2}_1(v) \ge 1/2$, then $\feat{1}_t(v) = 1$ for all $1 \le t \le m$.
    Since $\feat{1}_t(v) = 1$,
    there exists $1 \le j \le 3$, such that $g(t, j) = 1$,
    which implies that $\ell_{tj}$ is satisfied by the assignment. So the $t^{th}$ clause is satisfied by the assignment.
    Since all clauses are satisfied by the assignment,
    $\varphi$ is satisfied by $\set{x_i \gets \feat{0}_{i}(v)}_{i \in \intsinterval{n}}$.
\end{proof}

\section{GNNs with unbounded activation functions} \label{sec:unbounded}

In this section, we consider  GNNs with unbounded activations, such as the standard $\relu$. Since we have already shown that global aggregation leads
to undecidability even in the bounded case, in this section \emph{we will only deal with GNNs that lack global aggregation}.
In Section~\ref{subsec:sep_unbounded},
we introduce a logic that also helps with understanding expressiveness of this class of GNNs.
In Section~\ref{subsec:unbounded_undec_sat},
we show that the satisfiability problem for $\ogreluGNN$ and $\blreluGNN$ are undecidable, a contrast to the case with eventually constant activation functions.
In Section~\ref{subsec:unbounded_undec_univ_sat},
we prove undecidability for universal satisfiability problems.

In Section~\ref{subsec:dec_unbounded},
we turn to the satisfiability problem for $\olreluGNN$, and give a positive result about decidability.
% Here we will not use the logic directly, but rather use components from decidability proofs for Presburger logics~\cite{localpresburgerbartosztony}.
We will use the idea of representing the possible values of activation, which was also utilized in the case of decidability for eventually constant activations. But in this case, we will be representing an infinite set of values, using Presburger formulas.

\subsection{Logics for unbounded activation function GNNs and separation of expressiveness} 
\label{subsec:sep_unbounded}

%\michael{Rewrote since there was a mismatch between the subsection title and intro sentences and the content}
We will prove that one can express more with unbounded activations then with bounded ones. We will use logical characterizations of these GNNs as a tool for this separation, and we will also use them in the next subsection to obtain undecidability results. We will not obtain an expressiveness characterization of GNNs with unbounded activations, but merely a logic that embeds in $\blreluGNN$s: local two-variable modal logic with two-hop Presburger quantifiers ($\blMtwoPtwo$),
which is the extension of $\blMPtwo$ where the guards are conjunctions of at most two binary predicates. 

%\ch{should include REGP and EGP here?}
%michael: I think we agreed not to do it
\begin{definition}
    The syntax of \emph{local two-variable modal logic with two-hop Presburger quantifiers} ($\blMtwoPtwo$) over vocabulary $\tau$ is defined inductively:
    \begin{itemize}
        \item $\top$ is a $\blMtwoPtwo$ formula.
        \item for a unary predicate $U \in \tau$, $U(x)$ is a $\blMtwoPtwo$ formula.
        \item if $\varphi(x)$ is a $\blMtwoPtwo$ formula, then so is $\neg\varphi(x)$.
        \item if $\varphi_1(x)$ and $\varphi_2(x)$ are $\blMtwoPtwo$ formulas, 
        then so is $\varphi_1(x) \land \varphi_2(x)$.
        \item if $\set{\varphi_t(x)}_{t \in \intsinterval{k}}$ and
        $\set{\varphi_t'(x)}_{t \in \intsinterval{k'}}$ 
        are sets of $\blMtwoPtwo$ formulas,
        $\set{\epsilon_{t}(x, z, y)}_{t \in \intsinterval{k}}$ is a set of guard formulas, each of form
                ${E(x, z)\land E(z, y)}$,
                ${E(x, z)\land E(y, z)}$,
                ${E(z, x)\land E(z, y)}$, or
                ${E(z, x)\land E(y, z)}$,
        and 
        $\set{\epsilon_{t}'(x, y)}_{t \in \intsinterval{k'}}$ is another set of guard formulas, each of form $E(x, y)$ or $E(y, x)$,
        then 
        \begin{equation*}
            % \left(
            \sum_{t \in \intsinterval{k}} \lambda_t \cdot \presbzy{\epsilon_{t}(x,z,y)\land\varphi_t(y)} + 
            \sum_{t \in \intsinterval {k'}} \lambda_t' \cdot \presby{\epsilon_{t}'(x,y)\land\varphi_t'(y)}
            \ \circledast\ \delta
            % \right)
        \end{equation*}
        is also a $\blMtwoPtwo$ formula.
       % \michael{Shouldn't the big parens exclude the inequality or equality with $\delta$?}
      %  \ch{We use parentheses to emphasize the scope of the Presburger quantifiers.}
      %  \michael{Seems strange to me in this display above, but not important}
      %  \ch{rewritten}
        The numbers $\delta, \lambda_t$, $\lambda'_t$, and the comparison $\circledast$ are as in the standard Presburger quantifier definition.
    \end{itemize}
\end{definition}

The idea is that we can still count a linear combination of cardinalities of the number of nodes satisfying a given lower-level formula that are one-hop away from the current node -- as in $\blMPtwo$. Optionally, we can add on a linear combination of the number of two-hop paths that lead to a 
node satisfying other lower-level formulas.

The semantics of the formulas is given inductively, with the only step that is different from the usual cases being for the quantification, which is the obvious one. We call these \emph{two-hop Presburger quantifiers}.
We can apply a similar proof technique as in Theorem~\ref{thm:logic_to_gnn}
to show that every $\blMtwoPtwo$ formula is expressible using a$\blreluGNN$.
% We give the semantics of the logic inductively, with the only interesting step being for the new quantifiers, which we refer to as \emph{two-hop Presburger quantifiers}.
% We assume the obvious semantics for the conjunctions of guard atoms $\varphi_i(x, z, y)$ and $\varphi_i'(x,y)$.
% Given a graph $\cG$ and a vertex $v \in V$,
% we say that $\cP(x)$ holds in $\cG, x/v$,
% denoted by ${\cG \models \cP(v)}$,
% if the following (in)equality holds in $\bbZ$.
% \begin{equation*}
% \sum_{i=1}^k\lambda_i \cdot
%         \abs{\setc{(u_z, u_y) \in V^2}{\cG \models \varphi_i(v, u_z, u_y)}}
% \ + \
% \sum_{i=1}^{k'}\lambda_i' \cdot
%         \abs{\setc{u_y \in V}{\cG \models \varphi_i'(v,u_y)}}
%         \ \circledast\ \delta
%     \end{equation*}

\begin{theorem} \label{thm:logic_to_unbounded_gnn}
    For every $n$-$\blMtwoPtwo$ formula $\Psi(x)$,
    there exists an $n$-$\blreluGNN$ $\cA_\Psi$,
    such that $\Psi(x)$ and $\cA_\Psi$ are equivalent.
\end{theorem}

For every $n$-$\blMtwoPtwo$ formula $\Psi(x)$, let $L$ be the number of subformulas of $\Psi(x)$ and $\set{\varphi_\ell(x)}_{\ell \in \intsinterval{L}}$ be an enumeration of subformulas of $\Psi(x)$ satisfying that
$\varphi_L(x)$ is $\Psi(x)$
and for each $\varphi_i(x)$ and $\varphi_j(x)$, if $\varphi_i(x)$ is a strict subformula of $\varphi_j(x)$, then $i < j$.

% Since $\Psi$ will be clear from context, we will omit it from the notation.

We define the $(3L+1)$-layer $n$-$\blreluGNN$ $\cA_\Psi$ as follows.
The input dimension $\gnndim{0}$ is $n$.
For $1 \le \ell \le 3L$, the dimension $\gnndim{\ell}$ is $L + K + n$, and $\gnndim{3L+1} = 1$,
where $K$ is the maximum $k$ in any two-hop Presburger quantifier within $\Psi(x)$.
The coefficient matrices and bias vectors will be chosen so that the features of $\cA_\Psi$ satisfy the following conditions.
%Note that we will omit some $\relu$,
%since $\relup{\relup{x}} = \relup{x}$,
%which implies that $\relup{\feat{i}_j(v)} = \feat{i}_j(v)$.

% The conditions for the first layer,
% for $1 \le i \le n$, $\feat{1}_{3L + i}(v) = \feat{0}_{i}(v)$.

Firstly, for $1 \le \ell \le L$,
\begin{equation*}
    \feat{3L}_{\ell}(v)\ =\ \feat{3L-1}_{\ell}(v)\ =\ \cdots\ =\ \feat{3\ell}_{\ell}(v).
\end{equation*}
% for $1 \le i \le L$ and $i \neq \ell$,
% $\feat{3\ell+1}_i(v) = \feat{3\ell}_i(v) = \feat{3\ell-1}_i(v) = \feat{3\ell-2}_i(v)$.
The condition of $\feat{3\ell}_{\ell}(v)$ depends on the formula $\varphi_\ell(x)$.
\begin{itemize}
    \item If $\varphi_\ell(x)$ is $\top$,
    then $\feat{3\ell}_\ell(v) = 1$.
    \item If $\varphi_\ell(x)$ is $U_j(x)$ for some unary predicate $U_j$,
    then $\feat{3\ell}_\ell(v) = \feat{3\ell-1}_{L+K+j}(v)$.
    \item If $\varphi_\ell(x)$ is $\neg\varphi_j(x)$,
    then $\feat{3\ell}_\ell(v) = \relup{1-\feat{3\ell-1}_j(v)}$.
    \item If $\varphi_\ell(x)$ is $\varphi_{j_1}(x)\land\varphi_{j_2}(x)$,
    then $\feat{3\ell}_\ell(v) = \relup{\feat{3\ell-1}_{j_1}(v)+\feat{3\ell-1}_{j_2}(v)-1}$.
    \item If $\varphi_\ell(x)$ is
    \begin{equation*}
        % \left(
        \sum_{t \in \intsinterval{k}} \lambda_t\cdot \presbzy{\epsilon_{t}(x,z,y)\land\varphi_{j_t}(y)} + 
        \sum_{t \in \intsinterval {k'}} \lambda_t'\cdot \presby{\epsilon_{t}'(x,y)\land\varphi_{j_t'}(y)}
        \ \circledast\ \delta,
        % \right),
    \end{equation*}
    then
    \begin{equation*}
        \begin{aligned}
            \feat{3\ell}_\ell(v)\ =\ &\relup{1 - \feat{3\ell-1}_\ell(v)} \\
            \feat{3\ell-1}_\ell(v)\ =\ &
            \relup{
                -\sum_{t \in \intsinterval{k}} \lambda_t \cdot \sum_{u \in X^1_t(v)}\feat{3\ell-2}_{L + t}(u)
                -\sum_{t \in \intsinterval{k'}} \lambda_t' \cdot \sum_{u \in X'_t(v)}\feat{3\ell-2}_{j'_t}(u)
                + \delta - 1
            },
        \end{aligned}
    \end{equation*}
    and, for $1 \le t \le k$,
    \begin{equation*}
        \feat{3\ell-2}_{L+t}(v)\ =\ 
        \relup{\sum_{u \in X^2_t(v)} \feat{3\ell-3}_{j_t}(u)},
    \end{equation*}
    where
    \begin{equation*}
        \begin{aligned}
            X^1_t(v)\ :=\ &
            \begin{cases}
                \nbr{\out}(v), &\text{if $\epsilon_t(x, z, y)$ is $E(x, z) \land E(z, y)$ or $E(x, z) \land E(y, z)$} \\
                \nbr{\inc}(v), &\text{if $\epsilon_t(x, z, y)$ is $E(z, x) \land E(z, y)$ or $E(z, x) \land E(y, z)$}
            \end{cases} \\
            X^2_t(v)\ :=\ &
            \begin{cases}
                \nbr{\out}(v), &\text{if $\epsilon_t(x, z, y)$ is $E(x, z) \land E(z, y)$ or $E(z, x) \land E(z, y)$} \\
                \nbr{\inc}(v), &\text{if $\epsilon_t(x, z, y)$ is $E(x, z) \land E(y, z)$ or $E(z, x) \land E(y, z)$}
            \end{cases} \\
            X'_t(v)\ :=\ &
            \begin{cases}
                \nbr{\out}(v), &\text{if $\epsilon'_t(x, y)$ is $E(x, y)$} \\
                \nbr{\inc}(v), &\text{if $\epsilon'_t(x, y)$ is $E(y, x)$}.
            \end{cases}
        \end{aligned}
    \end{equation*}
    % Note that by the definition of subformulas, $k, k' \le 2L$.
\end{itemize}
Next, the conditions for the last $n$ entries of features are as follows:
for $1 \le t \le n$,
\begin{equation*}
    \feat{3L}_{L+K+t}(v)\ =\ \feat{3L-1}_{L+K+t}(v)\ =\ \cdots\ =\ \feat{1}_{L+K+t}(v)\ =\ \feat{0}_{t}(v).
\end{equation*}
Finally, for the last layer, $\feat{3L+1}_1(v) = \feat{3L}_L(v)$.

It is easy to see that there are GNNs that satisfy these conditions.
%\michael{Leave this for now}

Recall that a vector $\bfv$ is $\ell$-correct for $\Psi(x)$ w.r.t. $\tuple{\cG, v}$,
if, for $1 \le i \le \ell$, $\bfv_i = \eval{\varphi_i(v)}_\cG$.
The theorem will follow once we have shown the following property of $\cA_\Psi$:
\begin{lemma} \label{lem:inductivetranslation}
    % For every $n$-$\bgMPtwo$ formula $\Psi(x)$,
    % let $\cA_\Psi$ be the $n$-$\blreluGNN$ defined above.
    For every $n$-graph $\cG$ and vertex $v$ in $\cG$,
    for $1 \le \ell \le L$, the feature
    $\feat{3\ell}(v)$ of $\cA_\Psi$ is $\ell$-correct w.r.t. $\tuple{\cG, v}$.
\end{lemma}

\begin{proof} 
    The proof is similar to Lemma~\ref{lem:logicmimicsgnn}.
    Here we only present the case where the subformula is a two-hop Presburger quantifier.

    For $1 \le \ell \le L$, 
    for $1 \le i \le \ell - 1$,
    note that $\feat{3\ell}_i(v) = \feat{3(\ell-1)}_i(v)$.
    By the induction hypothesis, 
    $\feat{3(\ell-1)}(v)$ is $(\ell-1)$-correct for $\Psi(x)$ w.r.t. $\tuple{\cG, v}$,
    which implies that $\feat{3(\ell-1)}_i(v) = \eval{\varphi_{i}(v)}$.
    Thus, it holds that $\feat{3\ell}_i(v) = \eval{\varphi_{i}(v)}$.

    It remains to show the case where $i = \ell$,
    % $\feat{3\ell}_\ell(v) = \eval{\varphi_{\ell}(v)}$.
    %   \michael{Isn't this exactly what is stated above, with "Thus, it holds.."?}
    %  \ch{added some comments}
    Suppose that $\varphi_\ell(x)$ is a two-hop Presburger quantifier:
    \begin{equation*}
        % \left(
        \sum_{t \in \intsinterval{k}} \lambda_t\cdot \presbzy{\epsilon_{t}(x,z,y)\land\varphi_{j_t}(y)} + 
        \sum_{t \in \intsinterval {k'}} \lambda_t'\cdot \presby{\epsilon_{t}'(x,y)\land\varphi_{j_t'}(y)}
        \ \circledast\ \delta.
        % \right).
    \end{equation*}
   % \michael{Added, so that the reader will know what $j_t$ and $j'_t$ refer to}
   % \ch{thanks}
   % \begin{equation*}
   %     \left(\sum_{t \in \intsinterval{k}} \lambda_t\cdot \presbzy{\epsilon_{t}(x,z,y)\land\varphi_{j_t}(y)} + 
   %     \sum_{t \in \intsinterval {k'}} \lambda_t'\cdot \presby{\epsilon_{t}'(x,y)\land\varphi_{j_t'}(y)}
   %     \ \circledast\ \delta\right),
   % \end{equation*}
   % We first claim that
   % \begin{equation*}
   %     \abs{\setc{(u, u') \in V \times V}{\cG \models \epsilon_t(v, u, u')\land\varphi_{j_t}(u')}}
   %     \ =\ 
   %     \sum_{u \in X^1_t(v)} \feat{3\ell-2}_{L + t}(u).
   % \end{equation*}
   For $1 \le t \le k$,
   because $\varphi_{j_t}(x)$ is a strict subformula of $\varphi_\ell(x)$,
   $\ell - 1 \ge j_t$.
   By the induction hypothesis, 
   for every vertex $u \in V$,
   $\feat{3(\ell - 1)}(u)$ is $(\ell-1)$-correct for $\Psi(x)$ w.r.t. $\tuple{\cG, u}$,
   which implies that $\feat{3(\ell-1)}_{j_t}(u) = \eval{\varphi_{j_t}(u)}$.
   % Thus, if $\cG \models \varphi_{j_t}(u)$,
   % then $\feat{3\ell-2}_{j_t}(u) = 1$.
   % Otherwise, if $\cG \not\models \varphi_{j_t}(u)$,
   % then $\feat{3\ell-2}_{j_t}(u) = 0$.
   % If $\epsilon_t(x, z, y) = E(x, z) \land E(z, y)$, then
   Thus, we have
   \begin{equation*}
       \begin{aligned}
           \abs{\setc{(u, u') \in V \times V}{\cG \models \epsilon_t(v, u, u')\land\varphi_{j_t}(u')}}
           % \ =\ &
           % \abs{\setc{(u, u') \in V \times V}{\cG \models E(v, u)\land E(u, u')\land\varphi_{j_t}(u')}} \\
           \ =\ &
           \sum_{u \in X^1_t(v)} \sum_{u' \in X^2_t(u)} \eval{\varphi_{j_t}(u')} \\
           \ =\ &
           \sum_{u \in X^1_t(v)} \sum_{u' \in X^2_t(u)} \feat{3(\ell-1)}_{j_t}(u').
           % \ =\ &
           % \sum_{u \in \nbr{\out}(v)} \feat{3\ell-2}_{L + t}(u).
       \end{aligned}
   \end{equation*}
   % Note that $\nbr{\out}(v) = X_t(v)$ in this case.
   % We can treat the other three cases analogously. 
   For $1 \le t \le k'$,
   because $\varphi_{j'_t}(x)$ are strict subformulas of $\varphi_\ell(x)$,
   $\ell-1 \ge j'_t$.
   %  \michael{What is $j'_t$? Also above, the first time $j'_t$ appears}
   %  \ch{thanks, fixed typos in the definition}
   % \michael{I still don't see a clear definition of $j'_t$}
   % \ch{$j'_t$ is the second index of two-hop Presburger quantifiers.}
   By the induction hypothesis,
   for every $u \in V$,
   $\feat{3(\ell - 1)}(u)$ is $(\ell-1)$-correct for $\Psi(x)$ w.r.t. $\tuple{\cG, u}$,
   which implies that $\feat{3\ell-2}_{j'_t}(u) = \feat{3(\ell-1)}_{j'_t}(u) = \eval{\varphi_{j'_t}(u)}$.
   % if $\cG \models \varphi_{j'_t}(u)$, then $\feat{3i-2}_{j'_t}(u) = 1$.
   % Otherwise, if $\cG \not\models \varphi_{j_t}(u)$, then $\feat{3i-2}_{j'_t}(u) = 0$.
   % Furthermore, we can show that
   Thus, it holds that
   \begin{equation*}
       \begin{aligned}
           \abs{\setc{u \in V}{\cG \models \epsilon'_t(v, u)\land\varphi'_{j'_t}(u)}}
           \ =\ &
           \sum_{u \in X'_t(v)} \eval{\varphi_{j'_t}(u)}
           \ =\ &
           \sum_{u \in X'_t(v)} \feat{3\ell-2}_{j'_t}(u).
       \end{aligned}
   \end{equation*}

   Let $w_2(v)$ and $w_1(v)$ be values defined as follows:
   \begin{equation*}
       \begin{aligned}
           w_2(v)\ :=\ &
           \sum_{t \in \intsinterval{k}} \lambda_t \cdot
           \abs{\setc{(u, u') \in V\times V}{\cG \models \epsilon_t(v, u, u')\land\varphi_{j_t}(u')}} \\
           \ =\ &
           \sum_{t \in \intsinterval{k}} \lambda_t \cdot
           \left(\sum_{u \in X^1_t(v)} \sum_{u' \in X^2_t(u)} \feat{3\ell-3}_{L+t}(u)\right)
           \\
           w_1(v)\ :=\ &
           \sum_{t \in \intsinterval{k'}} \lambda'_t \cdot
           \abs{\setc{u \in V}{\cG \models \epsilon'_t(v, u)\land\varphi_{j'_t}(u)}} \\
           \ =\ &
           \sum_{t \in \intsinterval{k'}} \lambda'_t \cdot \sum_{u \in X'_t(v)}\feat{3\ell-2}_{j'_t}(u).
        \end{aligned}
    \end{equation*}
    Observe that
    \begin{equation*}
        \feat{3\ell}_\ell(v)\ =\ 
        \relup{1 - \relup{-w_2(v)-w_1(v)+\delta-1}}.
    \end{equation*}
    By the semantics of two-hop Presburger quantifiers,
    if $\cG \models \varphi_\ell(v)$, then $w_1(v) + w_2(v) \ge \delta$,
    which implies $\feat{3\ell}_\ell(v) = 1$.
    On the other hand,
    if $\cG \not\models \varphi_\ell(v)$, then $w_1(v) + w_2(v) < \delta$,
    which implies $\feat{3\ell}_\ell(v) = 0$.
\end{proof}

Theorem~\ref{thm:logic_to_unbounded_gnn} follows easily from Lemma~\ref{lem:inductivetranslation}.

We can use the logic to get an expressiveness separation for GNNs.
By Theorem~\ref{thm:logic_to_unbounded_gnn}
to show that $\blreluGNN$s can do more than $\blcGNN$s,
it is sufficient to show that there is a $\blMtwoPtwo$ formula that is not given by any $\blcGNN$:

\begin{lemma} \label{lemma:mtwoptwo_gap1}
    $\blMtwoPtwo$ is strictly more expressive than $\blcGNN$s.
\end{lemma}

We claim that the property ``the number of two-hop paths from the vertex $v$ to the green vertices is the same as the number of two-hop paths from the vertex $v$ to the blue vertices'' gives the separation. It is easy to express
in the two-hop logic. To show that no $\blcGNN$ can express it, we construct a sequence of pairs of graphs, each with a special node, such that the property holds
in the special node of the first graph and fails in the special node of the second, while for every $\blcGNN$ $\cA$,
for any sufficiently large pairs of graphs in this sequence,
the special nodes are indistinguishable by $\cA$.

\begin{definition}
    For $n_1, n_2 \in \bbN$,
    the \emph{$(n_1, n_2)$-bipolar graph} $\tuple{V, E, U_1, U_2}$ is a $2$-graph defined as follows,
    see Fig.~\ref{fig:bipolar}.
    \begin{equation*}
        \begin{aligned}
            U_1\ :=\ &\setc{v_{1, i}}{1 \le i \le n_1} \\
            U_2\ :=\ &\setc{v_{2, i}}{1 \le i \le n_2} \\
            V\ :=\ &U_1 \cup U_2 \cup \set{v_0, v_1, v_2} \\
            E\ :=\ &\set{(v_0, v_1), (v_0, v_2)}
                \cup \setc{(v_1, v_{1, i})}{1 \le i \le n_1}
                \cup \setc{(v_2, v_{2, i})}{1 \le i \le n_2} \\
            % E\ :=\ &\widetilde{E} \cup \setc{(u, v)}{(v, u) \in \widetilde{E}}
        \end{aligned}
    \end{equation*}
    \begin{figure}[ht!]
        \centering
        \begin{tikzpicture}[thick, main/.style = {draw, circle, minimum size=30}, scale=0.8]
            \node[main] (g0) at (0,0) {$v_0$};
            \node[main] (gr) at (2,0) {$v_2$};
            \node[main] (gl) at (-2,0) {$v_1$};
            
            \node[main] (gr1) at (.5+2.377, 1.798)  {$v_{2, 1}$};
            \node[main] (gr2) at (.5+3.402, 0.618)  {$v_{2, 2}$};
            \node[    ] ()    at (.5+3.402, -0.618) {$\vdots$};
            \node[main] (grn) at (.5+2.377, -1.798) {$v_{2, n_2}$};
            
            \node[main] (gl1) at (-.5-2.377, 1.798)  {$v_{1, 1}$};
            \node[main] (gl2) at (-.5-3.402, 0.618)  {$v_{1, 2}$};
            \node[    ] ()    at (-.5-3.402, -0.618) {$\vdots$};
            \node[main] (gln) at (-.5-2.377, -1.798) {$v_{1, n_1}$};
            
            \path [->] (g0) edge (gr);
            \path [->] (gr) edge (gr1);
            \path [->] (gr) edge (gr2);
            \path [->] (gr) edge (grn);
            \path [->] (g0) edge (gl);
            \path [->] (gl) edge (gl1);
            \path [->] (gl) edge (gl2);
            \path [->] (gl) edge (gln);
        \end{tikzpicture}
        \caption{$(n_1, n_2)$-bipolar graph.}
        \label{fig:bipolar}
    \end{figure}
\end{definition}

We first show that $\blcGNN$s cannot distinguish sufficiently large pairs of bipolar graphs.

\begin{lemma} \label{lemma:bipolar}
    For each $2$-$\blcGNN$ $\cA$,
    there exist a threshold $n_\cA \in \bbN$,
    such that for every ${n_1, n_2 \ge n_\cA}$,
    the following properties hold.
    Let $\cG$ be the $(n_\cA, n_\cA)$-bipolar graph and $\cG'$ be the $(n_1, n_2)$-bipolar graph.
    For every $0 \le \ell \le L$,
    \begin{enumerate}
        \item $\feat{\ell}_{\cG}(v_0) = \feat{\ell}_{\cG'}(v'_0)$,
        $\feat{\ell}_{\cG}(v_1) = \feat{\ell}_{\cG'}(v'_1)$, and
        $\feat{\ell}_{\cG}(v_2) = \feat{\ell}_{\cG'}(v'_2)$.

        \item for $1 \le i \le n_\cA$ and $1 \le i' \le n_1$,
        $\feat{\ell}_{\cG}(v_{1, i}) = \feat{\ell}_{\cG'}(v'_{1, i'})$.

        \item for $1 \le i \le n_\cA$ and $1 \le i' \le n_2$,
        $\feat{\ell}_{\cG}(v_{2, i}) = \feat{\ell}_{\cG'}(v'_{2, i'})$.
    \end{enumerate}
\end{lemma}

\begin{proof}
    Recall from Theorem~\ref{thm:computespectrum} that the spectrum of a $\blcGNN$ at any layer $\ell$, denoted $\spectrum{\ell}$, is finite.
    We will show that we  can compute the required threshold using the thresholds for the eventually constant activations along with the maximum rational number in the spectrum.

    For $1 \le \ell \le L$ and $1 \le i \le \gnndim{\ell}$,
    for $\bfs_1, \bfs_2, \bfs_3 \in \spectrum{\ell-1}$, let
    \begin{equation*}
        \begin{aligned}
            \fq{\ell}{i}{\bfs_1}\ :=\ &\left(\coefA{\ell}{\out} \bfs_1\right)_i \\
            \fp{\ell}{i}{\bfs_2, \bfs_3}\ :=\ &\left(\coefC{\ell} \bfs_2 + \coefA{\ell}{\inc} \bfs_3 + \coefb{\ell}\right)_i.
        \end{aligned}
    \end{equation*}
    We define $\fn{\ell}_i(\bfs_1, \bfs_2, \bfs_3) \in \bbN$ as follows:
    \begin{equation*}
        \fn{\ell}_i(\bfs_1, \bfs_2, \bfs_3)\ :=\ 
        \begin{dcases}
            \ceil{\frac{\tr{\ell} - \fp{\ell}{i}{\bfs_2, \bfs_3}}{\fq{\ell}{i}{\bfs_1}}},&
            \text{if $\fq{\ell}{i}{\bfs_1} > 0$ and $\fp{\ell}{i}{\bfs_2, \bfs_3} < \tr{\ell}$} \\
            \ceil{\frac{\fp{\ell}{i}{\bfs_2, \bfs_3} - \tl{\ell}}{-\fq{\ell}{i}{\bfs_1}}}, &
            \text{if $\fq{\ell}{i}{\bfs_1} < 0$ and $\fp{\ell}{i}{\bfs_2, \bfs_3} > \tl{\ell}$} \\
            0,& \text{otherwise}.
        \end{dcases}
    \end{equation*}
    Let $n_\cA$ be the maximum of $\fn{\ell}_i(\bfs_1, \bfs_2, \bfs_3)$.
    Since $\spectrum{\ell}$ has finite size, a maximum value exists.

    We prove the lemma by induction on the layers of $\cA$.
    For the base case $\ell = 0$, the properties hold by the definition of bipolar graphs.
    For the induction step $1 \le \ell \le L$,
    we first compute the features for $v_0$ and $v_0'$.
    Note that by the induction hypothesis,
    $\feat{\ell-1}_\cG(v_0) = \feat{\ell-1}_{\cG'}(v'_0)$,
    $\feat{\ell-1}_\cG(v_1) = \feat{\ell-1}_{\cG'}(v'_1)$, and
    $\feat{\ell-1}_\cG(v_2) = \feat{\ell-1}_{\cG'}(v'_2)$.
    Therefore, we have
    \begin{equation*}
        \begin{aligned}
            \feat{\ell}_{\cG}(v_0)
            \ =\ &
            \actp{\ell}{\coefC{\ell} \feat{\ell-1}_{\cG}(v_0) + \coefA{\ell}{\out} \left(\feat{\ell-1}_{\cG}(v_1) + \feat{\ell-1}_{\cG}(v_2)\right) + \coefb{\ell}} \\
            \ =\ &
            \actp{\ell}{\coefC{\ell} \feat{\ell-1}_{\cG'}(v'_0) + \coefA{\ell}{\out} \left(\feat{\ell-1}_{\cG'}(v'_1) + \feat{\ell-1}_{\cG'}(v'_2)\right) + \coefb{\ell}}
            \ =\ 
            \feat{\ell}_{\cG'}(v'_0).
        \end{aligned}
    \end{equation*}
    We can check the features of $v_{1, i}$ and $v_{2, i}$ similarly.
    
    For the features of $v_1$,
    by the induction hypothesis, it holds that
    $\feat{\ell-1}_{\cG}(v_0) = \feat{\ell-1}_{\cG'}(v'_0)$ and
    $\feat{\ell-1}_{\cG}(v_1) = \feat{\ell-1}_{\cG'}(v'_1)$.
    For $1 \le i \le n_\cA$ and $1 \le i' \le n_1$,
    $\feat{\ell-1}_{\cG}(v_{1,i}) = \feat{\ell-1}_{\cG'}(v'_{1,i'})$.
    For $1 \le i \le \gnndim{\ell}$,
    we can rewrite $\feat{\ell}_{\cG}(v_1)$ and $\feat{\ell}_{\cG'}(v'_1)$ as follows:
    \begin{equation*}
        \begin{aligned}
            \feat{\ell}_{\cG, i}(v_1)
            \ =\ &
            \left(
            \actp{\ell}{
            \coefC{\ell} \feat{\ell-1}_\cG(v_1) +
                \coefA{\ell}{\inc} \feat{\ell-1}_\cG(v_0) + 
                \coefA{\ell}{\out} \sum_{j \in \intsinterval{n_\cA}} \feat{\ell-1}_\cG(v_{1, j}) +
                \coefb{\ell}
            }
            \right)_i \\
            \ =\ &
            \actp{\ell}{\fp{\ell}{i}{\feat{\ell-1}_\cG(v_1), \feat{\ell-1}_\cG(v_0)} +
                n_\cA \cdot \fq{\ell}{i}{\feat{\ell-1}_\cG(v_{1, 1})}} \\
            \ =\ &
            \actp{\ell}{p + n_\cA q} \\
            \feat\ell_{\cG', i}(v'_1)
            \ =\ &
            \left(
            \actp{\ell}{
            \coefC{\ell} \feat{\ell-1}_{\cG'}(v'_1) +
                \coefA{\ell}{\inc} \feat{\ell-1}_{\cG'}(v'_0) + 
                \coefA{\ell}{\out} \sum_{j \in \intsinterval{n_1}} \feat{\ell-1}_{\cG'}(v'_{1, j}) +
                \coefb{\ell}
            }
            \right)_i \\
            \ =\ &
            \actp{\ell}{\fp{\ell}{i}{\feat{\ell-1}_{\cG'}(v'_1), \feat{\ell-1}_{\cG'}(v'_0)} +
                n_1 \cdot \fq{\ell}{i}{\feat{\ell-1}_{\cG'}(v'_{1, 1})}} \\
            \ =\ &
            \actp{\ell}{p + n_1 q}
        \end{aligned}
    \end{equation*}
    where $q := \fq{\ell}{i}{\feat{\ell-1}_\cG(v_{1, 1})}$ and $p := \fp{\ell}{i}{\feat{\ell-1}_\cG(v_1), \feat{\ell-1}_\cG(v_0)}$.
    % Note that $\feat{\ell-1}_\cG(v_1), \feat{\ell-1}_\cG(v_0), \feat{\ell-1}_\cG(v_{1, 1}) \in \spectrum{\ell-1}$.
    % By the definition of $n_\cA$,
    % \begin{equation*}
    %     n_\cA \ge \fn{\ell}_i\left(\feat{\ell-1}_\cG(v_1), \feat{\ell-1}_\cG(v_0), \feat{\ell-1}_\cG(v_{1, 1})\right).
    % \end{equation*}
    We consider the following cases of $q$ and $p$.
    Note that $\act{\ell}$ is eventually constant with the left threshold $\tl{\ell}$ and the right threshold $\tr{\ell}$.
    \begin{itemize}
        \item If $q = 0$,
        then $\feat\ell_{\cG, i}(v_1) = \actp{\ell}{p} = \feat\ell_{\cG', i}(v'_1)$.
        
        \item Suppose that $q > 0$.
        \begin{itemize}
            \item
            If $p \ge \tr{\ell}$, then
            $p + n_\cA q \ge p \ge \tr{\ell}$ and 
            $p + n_1 q \ge p \ge \tr{\ell}$.
        
            \item
            If $p < \tr{\ell}$, 
            by the definition of $n_\cA$,
            $n_1 \ge n_\cA \ge \ceil{\left(\tr{\ell}-p\right) / q}$,
            which implies that $p + n_1 q \ge p + n_\cA q \ge \tr{\ell}$.
        \end{itemize}
        Hence, for both cases,
        $\feat{\ell}_{\cG, i}(v_1) = \actp{\ell}{\tr{\ell}} = \feat{\ell}_{\cG', i}(v'_1)$.

        \item Suppose that $q < 0$.
        \begin{itemize}
            \item
            If $p \le \tl{\ell}$, then
            $p + n_\cA q \le p \le \tl{\ell}$ and 
            $p + n_1 q \le p \le \tl{\ell}$.

            \item
            If $p > \tl{\ell}$, 
            by the definition of $n_\cA$,
            $n_1 \ge n_\cA \ge \ceil{\left(p - \tl{\ell}\right) / -q}$,
            which implies that $p + n_1 q \le p + n_\cA q \le \tl{\ell}$.
        \end{itemize}
        Hence, for both cases,
        $\feat{\ell}_{\cG, i}(v_1) = \actp{\ell}{\tl{\ell}} = \feat{\ell}_{\cG', i}(v'_1)$.
    \end{itemize}
    Thus, it holds that $\feat{\ell, i}_\cG(v_1) = \feat{\ell}_{\cG', i}(v'_1)$.
    
    The feature of $v_2$ can be treated analogously.
    This completes the proof of the lemma.
\end{proof}

We are now ready to show that there is a $\blMtwoPtwo$ formula that is not captured
by a $\blcGNN$, which is the main claim of Lemma~\ref{lemma:mtwoptwo_gap1}:

\begin{proof}
    Consider the following $2$-$\blMtwoPtwo$ formula:
    \begin{equation*}
        \Psi(x)\ :=\ 
        \left(
        \presbzy{E(x, z) \land E(z, y) \land U_1(y)} = 
        \presbzy{E(x, z) \land E(z, y) \land U_2(y)}
        \right)
    \end{equation*}
    Let $\cG$ be the $(n_1, n_2)$-bipolar graph.
    It is routine to check that $\cG \models \Psi(v_0)$
    if and only if $n_1 = n_2$.

    For every $2$-$\blcGNN$ $\cA$,
    let $n_\cA$ be the constant defined in Lemma~\ref{lemma:bipolar}.
    Let $\cG$ be the $(n_\cA, n_\cA)$-bipolar graph and
    $\cG'$ be the $(n_\cA, n_\cA + 1)$-bipolar graph.
    By Lemma~\ref{lemma:bipolar}, $\xi^L_{\cG, 1}(v_0) = \xi^L_{\cG', 1}(v'_0)$,
    which implies that $\tuple{\cG, v_0}$ and $\tuple{\cG', v'_0}$ are indistinguishable by $\cA$.
    On the other hand, $\cG \models \Psi(v_0)$ but $\cG' \not\models \Psi(v'_0)$,
    they are distinguishable by $\Psi(x)$.
    Therefore $\cA$ and $\Psi(x)$ are not equivalent.
\end{proof}

We can now prove the inclusion in Lemma 
\ref{lemma:mtwoptwo_gap1}. We know that $\blcGNN$s are equivalent in expressiveness to $\blMPtwo$.
So it suffices to show that $\blMPtwo$ is subsumed in expressiveness by $\blMtwoPtwo$. 
But note that each Presburger quantifier is also a two-hop Presburger quantifier with no two-hop terms.
Then it is obvious that $\blMPtwo$ is subsumed in expressiveness by $\blMtwoPtwo$.

The following results are direct consequences of the lemma above and the logical characterization in the prior section:

\begin{corollary}
    $\blMtwoPtwo$ is strictly more expressive than $\blMPtwo$.
\end{corollary}

\begin{corollary}
    $\blreluGNN$s  are strictly more expressive than $\blcGNN$s.
\end{corollary}

\subsection{Undecidability of the satisfiability problem for GNNs with unbounded activations} \label{subsec:unbounded_undec_sat}

In Section~\ref{subsec:global_undecidable},
we show that the satisfiability problem for $\bgtrreluGNN$ is undecidable.
Since $\bgreluGNN$ subsumes $\bgtrreluGNN$, the satisfiability problem of $\bgreluGNN$ if also undecidable.
%\michael{The way the sentence is written it seems to be saying that the second thing was also shown in
%the earlier section} \ch{rewritten}
However, recall that for GNNs with $\trrelu$ activation,
the satisfiability problem is decidable when we consider \emph{local} subclasses of GNNs with $\trrelu$.
In contrast, we will prove that the satisfiability problems remain \emph{undecidable} even for the \emph{local} subclass of GNNs with $\relu$ activation. 
For $\trrelu$ activation we do not know what happens when we do not restrict to local aggregation, but restrict to outgoing-only aggregation on directed graphs. But for the $\relu$ case,  we are able to show undecidability of satisfiability in the $\relu$ case even with outgoing-only aggregation.  That is, the two main results of this subsection are:
%\michael{Right now I do not see any undecidability theorem for satisfiability of unbounded GNNs...}

\begin{theorem}\label{thm:hilbert_to_ogrelugnn}
    The satisfiability problem for $\ogreluGNN$s is undecidable.
\end{theorem}

\begin{theorem}\label{thm:hilbert_to_blrelugnn}
    The satisfiability problem for $\blreluGNN$s is undecidable.
\end{theorem}

For both results, we will reduce from Hilbert's tenth problem to the satisfiability problem for both local and outgoing-only GNNs with $\relu$ activations.
Recall that the solvability (over $\bbN$) of simple equation systems,
as defined in Definition~\ref{def:simple_eq_sys}, is undecidable~\cite{hilbert_theth}. %michael{Maybe repeat citation for undecidability}\ch{fix}
Equations take one of the forms:
$\upsilon_{t_1} = 1$,
$\upsilon_{t_1} = \upsilon_{t_2} + \upsilon_{t_3}$,
or $\upsilon_{t_1} = \upsilon_{t_2} \cdot \upsilon_{t_3}$.
Encoding the solvability of the first two types of equations is straightforward.

% We can apply an approach similar to Lemma~\ref{lemma:hilbert_to_ogrelugnn} to construct a reduction to local GNNs with $\relu$ activations, thus proving Theorem \ref{thm:hilbert_to_blrelugnn}.
% We state the reduction in the following lemma:
Before presenting the encoding of simple equation systems,
we first show that for every quantifier-free Presburger formula $\varphi(\bfx)$,
there exists a ``$\varphi$-test GNN fragment $\cA$:
one that does a Boolean test to check if the features of $\cA$ satisfy $\varphi(\bfx)$.
The proof idea is similar to Theorem~\ref{thm:logic_to_gnn},
with the main difference being that the input to the Presburger formula is encoded only in the features of the GNN,
rather than in the entire input graph.
Thus, the constructed GNN does not perform any aggregation.
For each subformula, there is a corresponding feature entry that represents the subformula.
The feature value is $1$ if the subformula holds and $0$ otherwise.

For a quantifier-free Presburger formula $\varphi(\bfx)$,
the depth of $\varphi(\bfx)$,
denoted by $\depth{\varphi}$,
is defined inductively as follows:
\begin{equation*}
    \depth{\varphi}\ =\ 
    \begin{cases}
        1, &
        \text{if $\varphi(\bfx)$ is $\sum_{t \in \intsinterval{n}} a_t x_t \ge c$} \\
        \depth{\psi} + 1, &
        \text{If $\varphi(\bfx)$ is $\neg\psi(\bfx)$} \\
        \max_{t \in \intsinterval{k}} \depth{\psi_{t}} + 1, &
        \text{if $\varphi(\bfx)$ is
        $\bigwedge_{t \in \intsinterval{k}} \psi_{t}(\bfx)$ or
        $\bigvee_{t \in \intsinterval{k}} \psi_{t}(\bfx)$}.
    \end{cases}
\end{equation*}
For every $\bfa \in \bbN^n$,
we denote the evaluation of $\varphi(\bfx)$ on $\bfa$ by $\eval{\varphi(\bfa)}$.
That is, if $\varphi(\bfa)$ holds, then $\eval{\varphi(\bfa)} = 1$.
Otherwise, $\eval{\varphi(\bfa)} = 0$.

Recall from the preliminaries that when we have a Presburger formula $\phi(x_1 \ldots x_n)$
with a canonical ordering of variables $x_1 \ldots x_n)$, and an
an $n$-dimensional vector $\vec v$, it makes sense to say that $\phi(\vec v)$ holds. We will often apply this where the vector $\vec v$ is a feature vector of a GNN below.

\begin{lemma}\label{lemma:fnn}
    For every quantifier-free Presburger formula $\Psi(\bfx)$ with $n$ variables,
    %and depth $k$ \michael{what is depth for a quantifier-free formula?},
    there exist matrices $\coefC{i}_\Psi$ and vectors $\coefb{i}_\Psi$ for $1 \le i \le 2\depth{\Psi}$,
    such that the following property holds:
    for every $\bgreluGNN$ $\cA$ with integer coefficients,
    if there exists $\ell_0$ satisfied that
    $\coefC{\ell_0 + i} = \coefC{i}_\Psi$;
    $\coefb{\ell_0 + i} = \coefb{i}_\Psi$;
    $\coefA{\ell_0 + i}{\out}$, $\coefA{\ell_0 + i}{\inc}$, and $\coefR{\ell_0 + i}$ are zero matrices for $1 \le i \le 2\depth{\Psi}$,
    then it holds that for every graph $\cG$ and vertex $v$ in $\cG$,
    \begin{equation*}
        \feat{\ell_0 + 2\depth{\Psi}}_i(v)\ =\ 
        \begin{cases}
            \feat{\ell_0}_i(v), &\text{for $1 \le i \le n$} \\
            \eval{\Psi\left(\feat{\ell_0}(v)\right)}, &\text{for $i = n+1$},
        \end{cases}
    \end{equation*}
    where $\feat{\ell_0}(v)$ is the $\ell_0^{th}$ feature of the vertex $v$.
    % \begin{itemize}
    %     \item for $1 \le i \le n$, $\feat{\ell_0 + 2\depth{\Psi}}_i(v) = \feat{\ell_0}_i(v)$, and
    %     \item if $\Psi\left(\feat{\ell_0}(v)\right)$ holds,
    %     then $\feat{\ell_0 + 2\depth{\Psi}}_{n+1}(v) = 1$.
    %     Otherwise, $\feat{\ell_0 + 2\depth{\Psi}}_{n+1}(v) = 0$.
    % \end{itemize}
\end{lemma}

\begin{proof}
    Let $L$ be the number of subformulas of $\Psi(\bfx)$ and
    $\set{\varphi_i(\bfx)}_{i \in \intsinterval{L}}$ be an enumeration of subformulas of $\Psi(\bfx)$ satisfying that $\varphi_1(\bfx)$ is $\Psi(\bfx)$.
    For $1 \le i \le L$, let $h_i := \depth{\varphi_i}$.

    For the first layer, $\coefC{1}_\Psi \in \bbZ^{n \times (n+L)}$ and $\coefb{1}_\Psi \in \bbZ^{n+L}$.
    For $2 \le \ell \le 2\depth{\ell}$,
    $\coefC{\ell}_\Psi \in \bbZ^{(n+L) \times (n+L)}$ and $\coefb{\ell}_\Psi \in \bbZ^{n+L}$.
    We define the coefficients in $\coefC{\ell}_\Psi$ and $\coefb{\ell}_\Psi$ row by row as follows:
    firstly, for $1 \le i \le n$, for $1 \le \ell \le 2\depth{\Psi}$,
    $\left(\coefC{\ell}_\Psi\right)_{i, i} = 1$.

    Next, for $1 \le i \le L$, for $2h_i < \ell \le 2\depth{\Psi}$,
    $\left(\coefC{\ell}_\Psi\right)_{n+i, n+i} = 1$.
    The definition of layers $2h_i$ and $2h_i- 1$ depends on the formula $\varphi_i(\bfx)$.
    \begin{itemize}
        \item Suppose that $\varphi_i(\bfx)$ is $\sum_{t \in \intsinterval{n}} a_t x_t \ge c$.
        Let
        $\left(\coefb{2h_i}_\Psi\right)_{n + i} = 1$,
        $\left(\coefC{2h_i}_\Psi\right)_{n + i, n + i} = -1$,
        $\left(\coefb{2h_i-1}_\Psi\right)_{n + i} = c$, and,
        for $1 \le t \le n$, $\left(\coefC{2h_i - 1}_\Psi\right)_{n + i, n + t} = -a_t$.
        
        \item Suppose that $\varphi_i(\bfx)$ is $\neg\varphi_{j}(\bfx)$.
        Let
        $\left(\coefb{2h_i}_\Psi\right)_{n + i} = 1$ and
        $\left(\coefC{2h_i}_\Psi\right)_{n + i, n + j} = -1$.

        \item Suppose that $\varphi_i(\bfx)$ is $\bigwedge_{t \in \intsinterval{k}} \varphi_{j_t}(\bfx)$.
        Let
        $\left(\coefb{2h_i}_\Psi\right)_{n + i} = 1$,
        $\left(\coefC{2h_i}_\Psi\right)_{n + i, n + i} = -1$,
        $\left(\coefb{2h_i-1}_\Psi\right)_{n + i} = k$, and,
        for $1 \le t \le k$, $\left(\coefC{2h_i - 1}_\Psi\right)_{n + i, n + j_t} = -1$.
        
        \item Suppose that $\varphi_i(\bfx)$ is $\bigvee_{t \in \intsinterval{k}} \varphi_{j_t}(\bfx)$.
        Let
        $\left(\coefb{2h_i}_\Psi\right)_{n + i} = 1$,
        $\left(\coefC{2h_i}_\Psi\right)_{n + i, n + i} = -1$,
        $\left(\coefb{2h_i-1}_\Psi\right)_{n + i} = 1$, and,
        for $1 \le t \le k$, $\left(\coefC{2h_i - 1}_\Psi\right)_{n + i, n + j_t} = -1$. %\michael{Why can't these last cases overlap: couldn't we have $j_t=i$?}
        Note that since $\varphi_{j_t}(x)$ is a strict subformula of $\varphi_i(x)$, $j_t \neq i$, so these case definitions cannot clash.
    \end{itemize}
    Finally, all other entries of $\coefC{\ell}_\Psi$ and $\coefb{\ell}_\Psi$ are 0.

    Let $\cA$ be a $\bgreluGNN$ with integer coefficients satisfying that,
    $\coefC{\ell_0 + i} = \coefC{i}_\Psi$;
    $\coefb{\ell_0 + i} = \coefb{i}_\Psi$;
    $\coefA{\ell_0 + i}{\out}$, $\coefA{\ell_0 + i}{\inc}$, and $\coefR{\ell_0 + i}$ are zero matrices for $1 \le i \le 2\depth{\Psi}$,
    For every graph $\cG$ and vertex $v$ in $\cG$,
    we first observe that for $1 \le i \le n$,
    \begin{equation*}
        \feat{\ell_0 + 2\depth{\Psi}}_i(v)\ =\ 
        \feat{\ell_0 + 2\depth{\Psi} - 1}_i(v)\ =\ 
        \cdots\ =\ 
        \feat{\ell_0}_i(v),
    \end{equation*}
    and, for $1 \le i \le L$,
    \begin{equation*}
        \feat{\ell_0 + 2\depth{\Psi}}_{n+i}(v)\ =\ 
        \feat{\ell_0 + 2\depth{\Psi} - 1}_{n+i}(v)\ =\ 
        \cdots\ =\ 
        \feat{\ell_0 + 2h_i}_{n+i}(v).
    \end{equation*}

    Next, we make a claim:
    for $1 \le i \le L$,
    if $\varphi_i\left(\feat{\ell_0}(v)\right)$ holds, then $\feat{\ell_0+2h_i}(v) = 1$.
    Otherwise, $\feat{\ell_0+2h_i}(v) = 0$.
    The lemma will follow the claim trivially.

    We prove the claim by induction on subformulas.
    \begin{itemize}
        \item If $\varphi_i(\bfx)$ is $\sum_{t \in \intsinterval{n}} a_t x_t \ge c$, then
        \begin{equation*}
            \feat{\ell_0 + 2h_i}_{n + i}(v)\ =\ 
            \relup{1 - \relup{c - \sum_{t \in \intsinterval{n}} a_t \feat{\ell_0 + 2h_i - 2}_{t}(v)}}.
        \end{equation*}

        If $\varphi_i\left(\feat{\ell_0}(v)\right)$ holds,
        then $c - \sum_{t \in \intsinterval{n}} a_t \feat{\ell_0}_{t}(v) \le 0$.
        Since $\feat{\ell_0 + 2h_i - 2}_{t}(v) = \feat{\ell_0}_{t}(v)$,
        it holds that $c - \sum_{t \in \intsinterval{n}} a_t \feat{\ell_0 + 2h_i - 2}_{t}(v) \le 0$,
        which implies that $\feat{\ell_0 + 2h_i}_{n + i}(v) = 1$.
        On the other hand, because all coefficients in $\cA$ are integers,
        $\feat{\ell_0}(v)$ is an integer vector.
        If $\varphi_i\left(\feat{\ell_0}(v)\right)$ doesn't hold,
        then $c - \sum_{t \in \intsinterval{n}} a_t \feat{\ell_0}_{t}(v) \ge 1$,
        which implies that
        $\feat{\ell_0 + 2h_i}_{n + i}(v) = 0$.
        
        \item If $\varphi_i(\bfx)$ is $\neg \varphi_{j}(\bfx)$, then
        \begin{equation*}
            \feat{\ell_0 + 2h_i}_{n + i}(v)\ =\ \relup{1 - \feat{\ell_0 + 2h_i-1}_{n + j}(v)}.
        \end{equation*}

        If $\varphi_i\left(\feat{\ell_0}(v)\right)$ holds,
        then $\varphi_{j}\left(\feat{\ell_0}(v)\right)$ doesn't hold.
        Since $h_i = h_j + 1$, $\ell_0 + 2h_{j} \le \ell_0 + 2h_i - 1$.
        By the induction hypothesis, $\feat{\ell_0 + 2h_i-1}_{n + j}(v) = 0$.
        Thus, $\feat{\ell_0 + 2h_i}_{n + i}(v) = 1$.
        Otherwise, if $\varphi_{i}\left(\feat{\ell_0}(v)\right)$ doesn't hold,
        then $\feat{\ell_0 + 2h_i-1}_{n + j}(v) = 1$ and $\feat{\ell_0 + 2h_i}_{n + i}(v) = 0$.

        \item If $\varphi_i(\bfx)$ is $\bigwedge_{t \in \intsinterval{k}} \varphi_{j_t}(\bfx)$, then
        \begin{equation*}
            \feat{\ell_0 + 2h_i}_{n + i}(v)\ =\ 
            \relup{1 - \relup{k - \sum_{t \in \intsinterval{k}} \feat{\ell_0 + 2h_i-2}_{n + j_t}(v)}}.
        \end{equation*}

        If $\varphi_i\left(\feat{\ell_0}(v)\right)$ holds,
        then, for $1 \le t \le k$, $\varphi_{j_t}\left(\feat{\ell_0}(v)\right)$ holds.
        Since $h_{j_t} < h_i$, $\ell_0 + 2h_{j_t} \le \ell_0 + 2h_i - 2$.
        By the induction hypothesis, $\varphi_{j_t}\left(\feat{\ell_0}(v)\right) = 1$.
        Thus, $\feat{\ell_0 + 2h_i}_{n + i}(v) = 1$.
        Otherwise, if $\varphi_i\left(\feat{\ell_0}(v)\right)$ doesn't hold,
        then there exists some $1 \le t \le k$
        such that $\varphi_{j_t}\left(\feat{\ell_0}(v)\right)$ doesn't hold and
        $\feat{\ell_0+2h_i}_{n + i}(v) = 0$.
        
        \item If $\varphi_i(\bfx)$ is $\bigvee_{t \in \intsinterval{k}} \varphi_{j_t}(\bfx)$, then
        \begin{equation*}
            \feat{\ell_0 + 2h_i}_{n + i}(v)\ =\ 
            \relup{1 - \relup{1 - \sum_{t \in \intsinterval{k}} \feat{\ell_0 + 2h_i-2}_{n + j_t}(v)}}.
        \end{equation*}
        
        If $\varphi_i\left(\feat{\ell_0}(v)\right)$ holds,
        then there exists some $1 \le t \le k$ such that $\varphi_{j_t}\left(\feat{\ell_0}(v)\right)$ holds.
        Since $h_{j_t} < h_i$, $\ell_0 + 2h_{j_t} \le \ell_0 + 2h_i - 2$.
        By the induction hypothesis, $\varphi_{j_t}\left(\feat{\ell_0}(v)\right) = 1$.
        Thus, $\feat{\ell_0 + 2h_i}_{n + i}(v) = 1$.
        Otherwise, if $\varphi_i\left(\feat{\ell_0}(v)\right)$ doesn't hold,
        then,  $1 \le t \le k$, $\varphi_{j_t}\left(\feat{\ell_0}(v)\right)$ doesn't hold and
        $\feat{\ell_0+2h_i}_{n + i}(v) = 0$.
    \end{itemize}
\end{proof}

We begin with the encoding for Theorem \ref{thm:hilbert_to_ogrelugnn}. 
We demonstrate the ability of $\ogreluGNN$s to check multiplication by the following example.

The idea will be that we construct a $\ogreluGNN$ $\cA$ such that the sizes of nodes satisfying certain unary predicates, within models of $\cA$,  will correspond to solutions of a set of equations.

For example, consider simulating a single equation  $x \cdot y=a$.
for $a \in \bbN$. Let $\cA_a$ be the $3$-layer 1-$\ogreluGNN$
defined as follows.
The dimensions are $d_0=d_1=d_3=1$ and $d_2=2$.
The coefficient matrices and bias vectors will  be chosen so that the features satisfy the following conditions,
\begin{equation*}
    \begin{aligned}
        &\feat{1}_1(v)\ =\ \relup{\sum_{u \in V} \feat{0}_1(u)},
        &\feat{2}_1(v)\ =\ \relup{\sum_{u \in V} \feat{1}_1(u) - a}, \\
        &\feat{2}_2(v)\ =\ \relup{a - \sum_{u \in V} \feat{1}_1(u)},
        &\feat{3}_1(v)\ =\ \relup{1 - \feat{2}_1(v) - \feat{2}_2(v)}.
    \end{aligned}
\end{equation*}

We claim that $\cA_a$ accepts $\tuple{\cG, v}$ if and only if $a = \abs{U_1} \cdot \abs{V}$, where $U_1$ is the unique input predicate.
To see this, note that $\feat{1}_1(v) = \abs{U_1}$ and $\sum_{u \in V} \feat{1}_1(u) = \abs{U_1} \cdot \abs{V}$.
For the last layer, $\feat{3}_1(v) = 1$ if and only if $\feat{2}_1(v) = \feat{2}_2(v) = 0$,
that is, $\sum_{u \in V} \feat{1}_1(u) = a$.

We now turn from intuition to the details of the construction.

\begin{lemma}\label{lemma:hilbert_to_ogrelugnn}
    For every simple equation system $\varepsilon$ with $n$ variables and $m$ equations,
    there exists an $(n+2m)$-$\ogreluGNN$ $\cA_\varepsilon$ 
    such that
    $\varepsilon$ has a solution in $\bbN$ if and only if
    $\cA_\varepsilon$ is satisfiable.
\end{lemma}

\begin{proof}
    We define the 15-layer $(n+2m)$-$\ogreluGNN$ $\cA_\varepsilon$ with 
    integer coefficient matrices and bias vectors as follows.
    % The input dimension $\gnndim{0}$ is $n + 2m$.
    % For $1 \le \ell \le X$, the dimension $\gnndim{\ell}$ is $n + 6m + 2$,
    % and $\gnndim{10} = 1$.
    The dimensions, coefficient matrices, and bias vectors will be chosen so that the features of $\cA_\varepsilon$ satisfy a set of conditions, which we now describe.
    Let 
    \begin{equation*}
        \begin{aligned}
            \Psi_1(\bfx)\ :=\ &\bigwedge_{t \in \intsinterval{m}} \Psi_{1, t}(\bfx) \\
            \Psi_2(\bfx)\ :=\ &
            \bigwedge_{t \in \intsinterval{m}} \Psi_{2, t}(\bfx)\ \land\ 
            \left(x_{n+2m+1} = x_{n+2m+2}\right),
        \end{aligned}
    \end{equation*}
    where, for $1 \le t \le m$,
    \begin{itemize}
        \item if the $t^{th}$ equation in $\varepsilon$ is $\upsilon_{t_1} = 1$,
        then $\Psi_{1, t}(\bfx) := \top$ and
        $\Psi_{2, t}(\bfx) := \left(x_{t_1} = 1\right)$.

        \item if the $t^{th}$ equation in $\varepsilon$ is $\upsilon_{t_1} =\upsilon_{t_2} + \upsilon_{t_3}$,
        then $\Psi_{1, t}(\bfx) := \top$ and
        $\Psi_{2, t}(\bfx) := \left(x_{t_1} = x_{t_2} + x_{t_3}\right)$.

        \item if the $t^{th}$ equation in $\varepsilon$ is $\upsilon_{t_1} =\upsilon_{t_2} \cdot \upsilon_{t_3}$,
        then
        \begin{equation*}
            \begin{aligned}
                \Psi_{1, t}(\bfx)\ :=\ &\left(\left(x_{n+t} = 0\right) \land \left(x_{n+m+t} = 0\right)\right) \lor \left(\left(x_{n+t} = x_{t_2}\right) \land \left(x_{n+m+t} = 1\right)\right) \\
                \Psi_{2, t}(\bfx)\ :=\ &\left(x_{n + t} = x_{t_1}\right) \land \left(x_{n + m + t} = x_{t_3}\right).
            \end{aligned}
        \end{equation*}
    \end{itemize}

    For the first layer,
    \begin{equation*}
        \feat{1}_i(v)\ =\ 
        \begin{dcases}
            \sum_{u \in V} \feat{0}_i(u), &\text{for $1 \le i \le n$} \\ 
            \sum_{u \in \nbr{\out}(v)} \feat{0}_i(u), &\text{for $n < i \le n + m$} \\
            \feat{0}_i(v), &\text{for $n + m < i \le n + 2m$} \\
            1, &\text{for $i = n + 2m + 1$}.
        \end{dcases}
    \end{equation*}

    Foe the second to the ninth layers, $\cA_\varepsilon$ test
    whether $\Psi_1\left(\feat{1}(v)\right)$ holds.
    That is
    \begin{equation*}
        \feat{9}_i(v)\ =\ 
        \begin{dcases}
            \feat{1}_i(v), &\text{for $1 \le i \le n+2m+1$} \\
            \eval{\Psi_1\left(\feat{1}(v)\right)}, &\text{for $i = n+2m+2$}.
        \end{dcases}
    \end{equation*}

    For the tenth layer,
    \begin{equation*}
        \feat{10}_i(v)\ =\ 
        \begin{dcases}
            \feat{9}_i(v), &\text{for $1 \le i \le n$} \\
            \sum_{u \in V} \feat{9}_{i}(u), &\text{for $n < i \le n + 2m + 2$}.
        \end{dcases}
    \end{equation*}

    Foe the eleventh to the fourteenth layers, $\cA_\varepsilon$ test
    whether $\Psi_2\left(\feat{10}(v)\right)$ holds.
    That is
    \begin{equation*}
        \feat{14}_i(v)\ =\ 
        \begin{dcases}
            \feat{10}_i(v), &\text{for $1 \le i \le n+2m+2$} \\
            \eval{\Psi_2\left(\feat{10}(v)\right)}, &\text{for $i = n+2m+3$}.
        \end{dcases}
    \end{equation*}

    Finally, for the last layer, $\feat{15}_{1}(v) = \feat{14}_{n+2m+3}(v)$.
    By Lemma~\ref{lemma:fnn}, it is easy to see that there are $\ogreluGNN$s that satisfy these conditions.
    % After these layers, we have a ``test GNN''  computing \michael{Issue: we did not technically say how to compose such a GNN. We should add a comment near the definition of test GNN}
    % $\feat{1}_{1} = \feat{14}_{n+2m+3} = \eval{\Psi_2\left(\feat{10}(v)\right)}$.
    % By Lemma~\ref{lemma:fnn}, it is easy to see that there are $\ogreluGNN$s that satisfy these conditions.
    
    We now show that this construction has the required properties.
    \textbf{\underline{Suppose $\cA_\varepsilon$ is satisfiable}} by the graph $\cG$ and vertex $v$ in $\cG$.
    Let $a_i := \feat{1}_i(v)$.
    We will show that $\set{\upsilon_i \gets a_i}_{i \in \intsinterval{n}}$ is a solution of $\varepsilon$.
    Note that for $1 \le i \le n$ and $u \in V$, $\feat{1}_{i}(u) = a_i$.

    We first make a claim.
    Fix $1 \le t \le m$,
    and suppose that the $t^{th}$ equation in $\varepsilon$ is $\upsilon_{t_1} = \upsilon_{t_2} \cdot \upsilon_{t_3}$.
    Then for every $u \in V$,
    \begin{equation*}
        \feat{9}_{n + t}(u)\ =\ a_{t_2} \cdot \feat{9}_{n + m + t}(u).
    \end{equation*}

    We now prove the claim.
    Because $\cA_\varepsilon$ is satisfied by $\tuple{\cG, v}$,
    $\feat{15}_{1}(v) = \feat{14}_{n+2m+3}(v) = \eval{\Psi_2\left(\feat{10}(v)\right)} = 1$.
    Since $\Psi_2\left(\feat{10}(v)\right)$ holds,
    it follows that
    \begin{equation*}
        \feat{10}_{n+2m+1}(v) = \feat{10}_{n+2m+2}(v).
    \end{equation*}
    By the definition of $\cA_\varepsilon$,
    we can rewrite both sides as follows:
    \begin{equation*}
        \begin{aligned}
            \feat{10}_{n+2m+1}(v)
            \ =\ &\sum_{u \in V} \feat{9}_{n+2m+1}(u)
            \ =\ \sum_{u \in V} \feat{1}_{n+2m+1}(u)
            \ =\ \sum_{u \in V} 1 \\
            \feat{10}_{n+2m+2}(v)
            \ =\ &\sum_{u \in V} \feat{9}_{n+2m+2}(u)
            \ =\ \sum_{u \in V}\eval{\Psi_1\left(\feat{1}(u)\right)},
        \end{aligned}
    \end{equation*}
    which implies that, for every $u \in V$, $\Psi_1\left(\feat{1}(u)\right)$ holds.
    %\michael{At the moment, I do not see why $\sum_{u \in V} \feat{9}_{n+2m+2}(u)$ equals $\sum_{u \in V} \feat{9}_{n+2m+1}(u)$. I am sure it is easy, but with so many definitions, why not give labels to some equations and explain as much as you can.}
   % \ch{rewritten}

    Fix $1 \le t \le m$,
    and suppose that the $t^{th}$ equation in $\varepsilon$ is $\upsilon_{t_1} = \upsilon_{t_2} \cdot \upsilon_{t_3}$.
    Since $\Psi_1\left(\feat{1}(u)\right)$ holds,
    $\Psi_{1, t}\left(\feat{1}(u)\right)$ also holds.
    Thus, one of the following must hold:
    \begin{equation*}
        \feat{1}_{n + t}(u) = 0 \land \feat{1}_{n + m + t}(u) = 0
        \quad\text{or}\quad
        \feat{1}_{n + t}(u) = \feat{1}_{t_2}(u) \land \feat{1}_{n + m + t}(u) = 1.
    \end{equation*}
    Note that $\feat{1}_{t_2}(u) = a_{t_2}$, and, by the definition of $\cA_\varepsilon$,
    \begin{equation*}
        \feat{9}_{n + t}(u)\ =\ \feat{1}_{n + t}(u)
        \quad\text{and}\quad
        \feat{9}_{n + m + t}(u)\ =\ \feat{1}_{n + m + t}(u).
    \end{equation*}
    Therefore, for both cases, it holds that $\feat{9}_{n + t}(u) = a_{t_2} \cdot \feat{9}_{n + m + t}(u)$.

    We now show that $\set{\upsilon_i \gets a_i}_{i \in \intsinterval{n}}$ is a solution of $\varepsilon$.
    Because $\cA_\varepsilon$ is satisfied by $\tuple{\cG, v}$,
    $\feat{15}_{1}(v) = \feat{14}_{n+2m+3}(v) = \eval{\Psi_2\left(\feat{10}(v)\right)} = 1$.
    % Because $\cA_\varepsilon$ is satisfied by $\tuple{\cG, v}$,
    % $\feat{15}_{1}(v) = \eval{\Psi_2\left(\feat{10}(v)\right)} = 1$.
    For $1 \le t \le m$,
    since $\Psi_2\left(\feat{10}(v)\right)$ holds,
    $\Psi_{2, t}\left(\feat{10}(v)\right)$ also holds.
  %  \michael{sentence does not quite parse, maybe because part was commented out. But I also don't recall why $\Psi_2\left(\feat{10}(v)\right)$ holds }
   % \ch{rewritten}
    We consider the $t^{th}$ equation in $\varepsilon$.
    Note that, for $1 \le i \le n$, $\feat{10}_{i}(v) = \feat{1}_{i}(v) = a_i$.
    \begin{itemize}
        \item If the $t^{th}$ equation is $\upsilon_{t_1} = 1$,
        then $\feat{10}_{t_1}(v) = 1$,
        which implies that
        $a_{t_1} = 1$.
        
        \item If the $t^{th}$ equation is $\upsilon_{t_1} = \upsilon_{t_2} + \upsilon_{t_3}$, then
        $\feat{10}_{t_1}(v) = \feat{10}_{t_2}(v) + \feat{10}_{t_3}(v)$,
        which implies that
        $a_{t_1} = a_{t_2} + a_{t_3}$,

        \item If the $t^{th}$ equation is $\upsilon_{t_1} = \upsilon_{t_2} \cdot \upsilon_{t_3}$, then
        \begin{equation*}
            \feat{10}_{n + t}(v)\ =\ \feat{10}_{t_1}(v)\quad \text{and} \quad
            \feat{10}_{n + m + t}(v)\ =\ \feat{10}_{t_3}(v).
        \end{equation*}
        By the definition of $\cA_\varepsilon$, we have
        \begin{equation*}
            \feat{10}_{n + t}(v)\ =\ \sum_{u \in V} \feat{9}_{n + t}(u)
            \quad\text{and}\quad
            \feat{10}_{n + m + t}(v)\ =\ \sum_{u \in V} \feat{9}_{n + m + t}(u).
        \end{equation*}
        Combining with the claim above,
        it holds that
        $\feat{10}_{n + t}(v) = a_{t_2}\cdot\feat{10}_{n + m + t}(v)$,
        which implies that
        $a_{t_1} = a_{t_2} \cdot a_{t_3}$. 
    \end{itemize}

    \textbf{\underline{Suppose $\varepsilon$ is solvable}}.
    Let $\set{\upsilon_i \gets a_i}_{i \in \intsinterval{n}}$ be a solution of $\varepsilon$ and
    $\am$ be the maximum value of $a_i$. 
    Note that in Definition~\ref{def:simple_eq_sys}, we define a system to be solvable if it has a solution in $\bbN$, thus the solution consists of natural numbers.
    We construct the following $(n+2m)$-graph $\cG$ depending on the solution of $\varepsilon$.
    For $1 \le j, j' \le \am$,
    let $v_j$ and $v_{j, j'}$ be fresh vertices and $V$ be the set of all such vertices.
    The colors of $\cG$ are defined as follows.
    \begin{itemize}
        \item For $1 \le i \le n$, $U_i := \setc{v_j}{1 \le j \le a_i}$.

        \item For $1 \le t \le m$, 
        if the $t^{th}$ equation in $\varepsilon$ is
        $\upsilon_{t_1} = \upsilon_{t_2} \cdot \upsilon_{t_3}$, then
        $U_{n + t} := \setc{v_{j, j'}}{1 \le j \le a_{t_3}, 1 \le j' \le a_{t_2}}$ and
        $U_{n + m + t} := \setc{v_j}{1 \le j \le a_{t_3}}$.
        Otherwise, 
        $U_{n + t} := \emptyset$ and $U_{n + m + t} := \emptyset$.
    \end{itemize}
    Finally, the edges of $\cG$ are $\setc{(v_j, v_{j, j'})}{1 \le j, j' \le \am}$.

    We now claim that $\cA_\varepsilon$ accepts $\tuple{\cG, v_1}$.
    We start from computing the first feature of vertices.
    \begin{itemize}
        \item For $1 \le j \le \am$,
        for $1 \le i \le n$, $\feat{1}_i(v_j) = a_i$.
        For $1 \le t \le m$,
        if the $t^{th}$ equation in $\varepsilon$ is
        $\upsilon_{t_1} = \upsilon_{t_2} \cdot \upsilon_{t_3}$ and $j \le a_{t_3}$,
        then $\feat{1}_{n+t}(v_j) = a_{t_2}$ and $\feat{1}_{n+m+t}(v_j) = 1$.
        Otherwise, $\feat{1}_{n+t}(v_j) = \feat{1}_{n+m+t}(v_j) = 0$.

        \item For $1 \le j, j' \le \am$,
        % for $1 \le i \le n$, $\feat{1}_i(v_{j,j'}) = a_i$, and,
        for $1 \le t \le m$, $\feat{1}_{n+t}(v_{j,j'}) = \feat{1}_{n+m+t}(v_{j,j'}) = 0$.
    \end{itemize}
    Then, it is clear that, for every vertex $u \in V$, 
    $\feat{9}_i(u) = \eval{\Psi_1\left(\feat{1}(u)\right)} = 1$.

    Next, we compute the tenth feature of $v_1$.
    For $1 \le i \le n$,
    \begin{equation*}
        \feat{10}_i(v_1) = \feat{9}_i(v_1) = \feat{1}_i(v_1) = a_i.
    \end{equation*}
    For $1 \le t \le n$,
    suppose that the $t^{th}$ equation in $\varepsilon$ is
    $\upsilon_{t_1} = \upsilon_{t_2} \cdot \upsilon_{t_3}$.
    We have
    \begin{equation*}
        \begin{aligned}
            \feat{10}_{n + t}(v_1)
            \ =\ &\sum_{u \in V} \feat{9}_{n + t}(u)
            \ =\ \sum_{j \in \intsinterval{a_{t_3}}} \feat{9}_{n + t}(v_j)
            \ =\ \sum_{j \in \intsinterval{a_{t_3}}} a_{t_2}
            \ =\ a_{t_2} \cdot a_{t_3} \\
            \feat{10}_{n + m + t}(v_1)
            \ =\ &\sum_{u \in V} \feat{9}_{n + m + t}(u)
            \ =\ \sum_{j \in \intsinterval{a_{t_3}}} \feat{9}_{n + m + t}(v_j)
            \ =\ \sum_{j \in \intsinterval{a_{t_3}}} 1
            \ =\ a_{t_3}
        \end{aligned}
    \end{equation*}
    Finally, it is straightforward to check that 
    $\feat{15}_1(v_1) = \feat{14}_{n+2m+3}(v_1) = \eval{\Psi_2\left(\feat{10}(v_1)\right)} = 1$,
    which implies that $\cA_\varepsilon$ accepts $\tuple{\cG, v_1}$.
\end{proof}

This completes the proof of Lemma~\ref{lemma:hilbert_to_ogrelugnn}, and
thus the proof of Theorem \ref{thm:hilbert_to_ogrelugnn}.

%\michael{Still missing the intuition for how bi-directional aggregation compensates for lack of global aggregation. Maybe redo the example with bi-directional aggregation?}
%\ch{added below}
We now turn to the proof of Theorem~\ref{thm:hilbert_to_blrelugnn}.
The idea is similar to Theorem~\ref{thm:hilbert_to_ogrelugnn}.
However, we now encode the solution of a set of equations in the number of children satisfying certain unary predicates.

For example, consider simulating a single equation  $x \cdot y=a$.
for $a \in \bbN$. Let $\cA_a$ be the $4$-layer 1-$\blreluGNN$
defined as follows.
The dimensions are $d_0=d_1=d_2=d_4=1$ and $d_3=2$.
The coefficient matrices and bias vectors will  be chosen so that the features satisfy the following conditions,
\begin{equation*}
    \begin{aligned}
        &\feat{1}_1(v)\ =\ \relup{\sum_{u \in \nbr{\out}(v)} \feat{0}_1(u)},
        &\feat{2}_1(v)\ =\ \relup{\sum_{u \in \nbr{\inc}(v)} \feat{1}_1(u)}, \\
        &\feat{3}_1(v)\ =\ \relup{\sum_{u \in \nbr{\out}(v)} \feat{2}_1(u) - a}, 
        &\feat{3}_2(v)\ =\ \relup{a - \sum_{u \in \nbr{\out}(v)} \feat{2}_1(u)}, \\
        &\feat{4}_1(v)\ =\ \relup{1 - \feat{3}_1(v) - \feat{3}_2(v)}.
    \end{aligned}
\end{equation*}

We claim that, for every $1$-level tree $\cT$ with root vertex $v_r$, $\cA_a$ accepts $\tuple{\cT, v_r}$ if and only if $a = \abs{\nbr{\out}(v_r)} \cdot \abs{U_1 \cap \nbr{\out}(v_r)}$, where $U_1$ is the unique input predicate.
To see this, note that $\feat{1}_1(v) = \abs{U_1\cap \nbr{\out}(v_r)}$, and,
for every $u \in \nbr{\out}(v_r)$, $\nbr{\inc}(u) = \set{v_r}$.
Thus, we have
\begin{equation*}
    \begin{aligned}
        \sum_{u \in \nbr{\out}(v_r)} \feat{2}_1(u)
        \ =\ &
        \sum_{u \in \nbr{\out}(v_r)}
        \sum_{u' \in \nbr{\inc}(u)}
        \feat{1}_1(u') \\
        \ =\ &
        \sum_{u \in \nbr{\out}(v_r)}
        \abs{U_1\cap \nbr{\out}(v_r)} \\
        \ =\ &
        \abs{\nbr{\out}(v_r)} \cdot
        \abs{U_1\cap \nbr{\out}(v_r)}
    \end{aligned}
\end{equation*}
For the last layer, $\feat{4}_1(v_r) = 1$ if and only if $\feat{3}_1(v_r) = \feat{3}_2(v_r) = 0$,
that is, $\sum_{u \in \nbr{\out}(v)} \feat{2}_1(u) = a$.

The correctness of the above example is based on the fact that the children of the root vertex in the $1$-level tree have only one incoming neighbor -- in this case, the root vertex.
Using separate incoming and outgoing aggregations, we can enforce the ``single incoming neighbor constraint'' on an arbitrary vertex, which will allow us encode iterative multiplication.

The general reduction is given in the following lemma:

\begin{lemma}\label{lemma:hilbert_to_blrelugnn}
    For every simple equation system $\varepsilon$ with $n$ variables and $m$ equations,
    there exists an $(n+2m)$-$\blreluGNN$ $\cA_\varepsilon$ 
    such that
    $\varepsilon$ has a solution in $\bbN$ if and only if
    $\cA_\varepsilon$ is satisfiable.
\end{lemma}

\begin{proof}
    We define the 16-layer $(n+2m)$-$\blreluGNN$ $\cA_\varepsilon$ with 
    integer coefficient matrices and bias vectors as follows.
    % The input dimension $\gnndim{0}$ is $n + 2m$.
    The dimensions, coefficient matrices, and bias vectors will be chosen so that the features of $\cA_\varepsilon$ satisfy certain conditions, to be described next.
    Let 
    \begin{equation*}
        \begin{aligned}
            \Psi_1(\bfx)\ :=\ &
            \bigwedge_{t \in \intsinterval{m}} \Psi_{1, t}(\bfx)\ \land\ 
            \left(x_{2n+2m+1} = 1\right) \\
            \Psi_2(\bfx)\ :=\ &
            \bigwedge_{t \in \intsinterval{m}} \Psi_{2, t}(\bfx)\ \land\ 
            \left(x_{n+2m+1} = x_{n+2m+2}\right),
        \end{aligned}
    \end{equation*}
    where, for $1 \le t \le m$,
    \begin{itemize}
        \item if the $t^{th}$ equation in $\varepsilon$ is $\upsilon_{t_1} = 1$,
        then $\Psi_{1, t}(\bfx) := \top$ and $\Psi_{2, t}(\bfx) := \left(x_{t_1} = 1\right)$.

        \item if the $t^{th}$ equation in $\varepsilon$ is $\upsilon_{t_1} =\upsilon_{t_2} + \upsilon_{t_3}$,
        then $\Psi_{1, t}(\bfx) := \top$ and
        $\Psi_{2, t}(\bfx) := \left(x_{t_1} = x_{t_2} + x_{t_3}\right)$.

        \item if the $t^{th}$ equation in $\varepsilon$ is $\upsilon_{t_1} =\upsilon_{t_2} \cdot \upsilon_{t_3}$,
        then 
        \begin{equation*}
            \begin{aligned}
                \Psi_{1, t}(\bfx)\ :=\ &\left(\left(x_{n+t} = 0\right) \land \left(x_{n+m+t} = 0\right)\right) \lor \left(\left(x_{n+t} = x_{n+2m+t_2}\right) \land \left(x_{n+m+t} = 1\right)\right) \\
                \Psi_{2, t}(\bfx)\ :=\ &\left(x_{n+t} = x_{t_1}\right) \land \left( x_{n+m+t} = x_{t_3}\right).
            \end{aligned}
        \end{equation*}
    \end{itemize}

    For the first layer,
    \begin{equation*}
        \feat{1}_i(v)\ =\ 
        \begin{dcases}
            \sum_{u \in \nbr{\out}(v)} \feat{0}_i(u), &\text{for $1 \le i \le n+m$} \\ 
            % \sum_{u \in \nbr{\out}(v)} \feat{0}_i(u), &\text{for $n < i \le n+m$} \\
            \feat{0}_i(v), &\text{for $n+m < i \le n+2m$} \\
            1, &\text{for $i = n + 2m + 1$}.
        \end{dcases}
    \end{equation*}

    For the second layer,
    \begin{equation*}
        \feat{2}_i(v)\ =\ 
        \begin{dcases}
            \feat{1}_i(v), &\text{for $1 \le i \le n+2m$} \\
            \sum_{u \in \nbr{\inc}(v)} \feat{1}_{i - n - 2m}(u), &\text{for $n+2m < i \le 2n+2m$} \\
            \sum_{u \in \nbr{\inc}(v)} \feat{1}_{i - n}(u), &\text{for $i = 2n + 2m + 1$} \\
            1, &\text{for $i = 2n + 2m + 2$}.
        \end{dcases}
    \end{equation*}

    For the third to the tenth layers,
    $\cA_\varepsilon$ test whether $\Psi_1\left(\feat{2}(v)\right)$ holds.
    That is
    \begin{equation*}
        \feat{10}_i(v)\ =\ 
        \begin{dcases}
            \feat{2}_i(v), &\text{for $1 \le i \le 2n+2m+2$} \\
            \eval{\Psi_1\left(\feat{2}(v)\right)}, &\text{for $i = 2n+2m+3$}.
        \end{dcases}
    \end{equation*}

    For the eleventh layer,
    \begin{equation*}
        \feat{11}_i(v)\ =\ 
        \begin{dcases}
            \feat{10}_i(v), &\text{for $1 \le i \le n$} \\
            \sum_{u \in \nbr{\out}} \feat{10}_{i}(u), &\text{for $n < i \le n + 2m$} \\
            \sum_{u \in \nbr{\out}} \feat{10}_{i+n+1}(u), &\text{for $n + 2m < i \le n + 2m + 2$}.
        \end{dcases}
    \end{equation*}

    For the twelfth to the fifteeneth layers,
    $\cA_\varepsilon$ test whether $\Psi_2\left(\feat{11}(v)\right)$ holds.
    That is
    \begin{equation*}
        \feat{15}_i(v)\ =\ 
        \begin{dcases}
            \feat{11}_i(v), &\text{for $1 \le i \le n+2m+2$} \\
            \eval{\Psi_2\left(\feat{11}(v)\right)}, &\text{for $i = n+2m+3$}.
        \end{dcases}
    \end{equation*}

    Finally, for the last layer, $\feat{16}_{1}(v) = \feat{15}_{n+2m+3}(v)$.
    By Lemma~\ref{lemma:fnn}, it is easy to see that there are $\blreluGNN$s that satisfy these conditions.
  %  \michael{Similar to my question above: aren't we jumping from layer $11$ to layer $15$?}
  %  \ch{rewritten} michael: thanks
    
    We now show that this construction has the required properties for the lemma.
    \textbf{\underline{Suppose $\cA_\varepsilon$ is satisfiable}} by the graph $\cG$ and vertex $v$ in $\cG$.
    Let $a_i := \feat{1}_i(v)$.
    We will show that $\set{\upsilon_i \gets a_i}_{i \in \intsinterval{n}}$ is a solution of $\varepsilon$.

    The proof idea is similar to Lemma~\ref{lemma:hilbert_to_ogrelugnn},
    with the only difference being the following variant of the main claim used in that lemma.
    For every $u \in \nbr{\out}(v)$, both properties hold.
    \begin{itemize}
        \item $\nbr{\inc}(u) = \set{v}$.
        \item For $1 \le t \le m$,
        if the $t^{th}$ equation in $\varepsilon$ is $\upsilon_{t_1} = \upsilon_{t_2} \cdot \upsilon_{t_3}$,
        then
        \begin{equation*}
            \feat{10}_{n+t}(u)\ =\ a_{t_2} \cdot \feat{10}_{n+m+t}(u).
        \end{equation*}
    \end{itemize}

    We now prove the claim.
    By the argument similar to Lemma~\ref{lemma:hilbert_to_ogrelugnn}.
    For every $u \in \nbr{\out}(v)$,
    $\Psi_1\left(\feat{2}(u)\right)$ holds.
    \begin{itemize}
        \item
        By the definition of $\cA_\varepsilon$, we have 
        \begin{equation*}
            \feat{2}_{2n+2m+1}(u)
            \ =\ \sum_{u' \in \nbr{\inc}(u)} \feat{1}_{n+2m+1}(u')
            \ =\ \sum_{u' \in \nbr{\inc}(u)} 1
            \ =\ \abs{\nbr{\inc}(u)}.
        \end{equation*}
        Because $\Psi_1\left(\feat{2}(u)\right)$ holds, it follows that $\feat{2}_{2n+2m+1}(u) = 1$,
        which implies that $u$ has exactly one incoming neighbor.
        Since $u$ is a outgoing neighbor of $v$, $v$ is an incoming neighbor of $u$.
        Thus, it holds that $\nbr{\inc}(u) = \set{v}$.

        \item Fix $1 \le t \le n$, and suppose that 
        the $t^{th}$ equation in $\varepsilon$ is $\upsilon_{t_1} = \upsilon_{t_2} \cdot \upsilon_{t_3}$.
        Since $\Psi_1\left(\feat{2}(u)\right)$ holds,
        $\Psi_{1, t}\left(\feat{2}(u)\right)$ also holds.
        Thus, one of the following must hold:
        \begin{equation*}
            \feat{2}_{n + t}(u) = 0 \land \feat{2}_{n + m + t}(u) = 0
            \quad\text{or}\quad
            \feat{2}_{n + t}(u) = \feat{2}_{n + 2m + t_2}(u) \land \feat{2}_{n + m + t}(u) = 1.
        \end{equation*}
        Note that by the definition of $\cA_\varepsilon$, we have
        \begin{equation*}
            \feat{10}_{n + t}(u)\ =\ \feat{2}_{n + t}(u)
            \quad\text{and}\quad
            \feat{10}_{n + m + t}(u)\ =\ \feat{2}_{n + m + t}(u).
        \end{equation*}
        Combining with the previous claim, $\nbr{\inc}(u) = \set{v}$,
        it follows that
        \begin{equation*}
            \feat{2}_{n + 2m + t_2}(u)
            \ =\ \sum_{u' \in \nbr{\inc}(u)} \feat{1}_{t_2}(u')
            % \ =\ \sum_{u' \in \set{v}} \feat{1}_{t_2}(u')
            \ =\ \feat{1}_{t_2}(v)
            \ =\ a_{t_2}.
        \end{equation*}
        Therefore, for both cases, it holds that
        $\feat{10}_{n + t}(u) = a_{t_2} \cdot \feat{10}_{n + m + t}(u)$.
    \end{itemize}

    \textbf{\underline{Suppose $\varepsilon$ is solvable}}.
    Let $\set{\upsilon_i \gets a_i}_{i \in \intsinterval{n}}$ be a solution of $\varepsilon$ and
    $\am$ be the maximum value of $a_i$.
    We construct the following $(n+2m)$-graph $\cG$ depending on the solution of $\varepsilon$.
    For $0 \le j \le \am$ and $1 \le j' \le \am$,
    let $v_j$ and $v_{j, j'}$ be fresh vertices and $V$ be the set of all such vertices.
    The colors of $\cG$ are defined as follows.
    \begin{itemize}
        \item For $1 \le i \le n$, $U_i = \setc{v_{0, j}}{1 \le j \le a_i}$.

        \item For $1 \le t \le m$, 
        if the $t^{th}$ equation in $\varepsilon$ is
        $\upsilon_{t_1} = \upsilon_{t_2} \cdot \upsilon_{t_3}$, then
        $U_{n + t} := \setc{v_{j, j'}}{1 \le j \le a_{t_3}, 1 \le j' \le a_{t_2}}$ and
        $U_{n + m + t} := \setc{v_j}{1 \le j \le a_{t_3}}$.
        Otherwise, 
        $U_{n + t} := \emptyset$ and $U_{n + m + t} := \emptyset$.
    \end{itemize}
    Finally, the edges of $\cG$ are
    \begin{equation*}
        E\ :=\ \setc{(v_j, v_{j, j'})}{0 \le j \le \am, 1 \le j' \le \am}\ \cup\ 
        \setc{(v_0, v_j)}{1 \le j \le \am}.
    \end{equation*}

    We can show that $\cA_\varepsilon$ accepts $\tuple{\cG, v_0}$
    via an argument similar to the one used in Lemma~\ref{lemma:hilbert_to_ogrelugnn}.
\end{proof}

Since the solvability (over $\bbN$) of simple equation systems is undecidable,
we obtain undecidability of the satisfiability problems for $\blreluGNN$ 
that is, Theorem~\ref{thm:hilbert_to_blrelugnn}.
%michael: restructured slightly

\subsection{Undecidability of the universal satisfiability problem for GNNs with unbounded activations} \label{subsec:unbounded_undec_univ_sat}

Next, we turn to the universal satisfiability problem for GNNs with $\relu$ activations.
We first show that the universal satisfiability for $\ogreluGNN$s is undecidable.
The proof is also by reducing Hilbert's tenth problem to it.

\begin{lemma}\label{lemma:hilbert_to_ogrelugnn_univ}
    For every simple equation system $\varepsilon$ with $n$ variables and $m$ equations,
    let $\cA_\varepsilon$ be the $(n+2m)$-$\ogreluGNN$ defined in Lemma~\ref{lemma:hilbert_to_ogrelugnn}.
    For every $(n+2m)$-graph $\cG$ and vertices $v, v'$ in $\cG$,
    $\cA_\varepsilon$ accepts $\tuple{\cG, v}$
    if and only if it accepts $\tuple{\cG, v'}$.
\end{lemma}

\begin{proof}
    We prove the lemma by analyzing the dependency of features across layers.
    For every $v \in V$, 
    by the definition of $\cA_\varepsilon$ in Lemma~\ref{lemma:hilbert_to_ogrelugnn},
    it holds that
    \begin{equation*}
        \feat{15}_1(v)
        \ =\ \feat{14}_{n+2m+3}(v)
        \ =\ \eval{\Psi_2\left(\feat{10}(v)\right)}.
    \end{equation*}
    For $1 \le i \le n$,
    \begin{equation*}
        \feat{10}_i(v)\ =\ \feat{9}_i(v)\ =\ \sum_{u \in V} \feat{0}_i(u),
    \end{equation*}
    and for $n < i \le n+2m+2$,
    \begin{equation*}
        \feat{10}_i(v)\ =\ \sum_{u \in V} \feat{9}_i(v).
    \end{equation*}
    Observe that the value of $\feat{15}_1(v)$ is independent of the neighbors of $v$.
    Therefore, for every pair of vertices $v, v' \in$, $\feat{15}_1(v) = \feat{15}_1(v')$,
    which implies that
    $\cA_\varepsilon$ accepts $\tuple{\cG, v}$
    if and only if it accepts $\tuple{\cG, v'}$.
\end{proof}

Note that it is obvious that if $\cA_\varepsilon$ is universal satisfiable,
then it is satisfiable.
Therefore, by Lemma~\ref{lemma:hilbert_to_ogrelugnn},
the simple equation system $\varepsilon$ has a solution in $\bbN$
if and only if $\cA_\varepsilon$ is universally satisfiable.
Thus we have the following undecidability result.

\begin{theorem}\label{thm:oggnn_unbounded_undec}
    The universal satisfiability problem of $\ogreluGNN$s is undecidable.
\end{theorem}

For $\blreluGNN$s,
note that we do not claim an expressive equivalence with a logic in Section~\ref{subsec:sep_unbounded}.
Nevertheless this containment of the logic $\blMtwoPtwo$ %\michael{which logic?}\ch{fix} 
in the GNN class is useful, since
we can show undecidability of the logic 
by reduction from the halting problem of two-counter machines,
which is known to be 
undecidable~\cite{minsky}.

\begin{definition}
    A \emph{two-counter machine} $\cM$ is a finite list of $n$ instructions having one of the forms \emph{$\tcminc{i}$}, \emph{$\tcmif{i}{q}$}, or \emph{$\tcmhalt$},
    where $i \in \set{0, 1}$ and $1 \le q \le n$.

    A \emph{configuration} is a tuple $\tuple{q, c_0, c_1}$, where $1 \le q \le n$ and $c_0, c_1 \in \bbN$. 
    We say $\tuple{q', c'_0, c'_1}$ is the \emph{successor configuration} of $\tuple{q, c_0, c_1}$
    depending on the $q^{th}$ instruction of the machine.
    \begin{itemize}
        \item If the $q^{th}$ instruction is \emph{$\tcminc{i}$},
        then $q' = q+1$, $c'_i = c_i+1$, and $c'_{1-i} = c_{1-i}$.
        \item If the $q^{th}$ instruction is \emph{$\tcmif{i}{j}$},
        if $c_i = 0$, then $q' = j$, $c'_0 = c_0$, and $c'_1 = c_1$;
        otherwise, $q' = q+1$, $c'_i = c_i-1$, and $c'_{1-i} = c_{1-i}$.
        \end{itemize}
        Note that if the $q^{th}$ instruction is \emph{$\tcmhalt$},
        there is no successor.
        This configuration is called a halting configuration. 
  
    The \emph{computation} of the machine is a (possibly infinite) sequence of configurations where the first is $\tuple{1, 0, 0}$,
    consecutive pairs are in the succcessor relationship above,
    the last configuration is a halting configuration.
    The machine \emph{halts} if its computation is a finite sequence.
\end{definition}

\begin{lemma}\label{lemma:tcm_to_mtwoptwo}
    For every two-counter machine $\cM$ with $n$ instructions, there exists an $(n+5)$-$\blMtwoPtwo$ formula $\Psi_\cM(x)$ such that $\cM$ halts if and only if the formula $\forall x\ \Psi_\cM(x)$ is finitely satisfiable. 
\end{lemma}

\begin{figure}[h]
    \centering
    \begin{tikzpicture}[thick, main/.style = {draw, circle, minimum size=30, scale=0.6}]
        \node[main] (c0) at (0,  0)  {$S$};
        \node[main] (c1) at (2,  0)  {$Q_1$};
        \node[main] (c2) at (4,  0)  {$Q_2$};
        \node[main] (c3) at (6,  0)  {$Q_3$};
        \node[main] (c4) at (8,  0)  {$Q_8$};
        \node[    ] (c5) at (9,  0)  {$\cdots$};
        \node[main] (c6) at (10,  0) {$Q_7$};
        \node[main] (c7) at (12, 0)  {$T$};
        \draw[->] (c0) -- (c1);
        \draw[->] (c1) -- (c2);
        \draw[->] (c2) -- (c3);
        \draw[->] (c3) -- (c4);
        \draw[->] (c4) -- (c5);
        \draw[->] (c5) -- (c6);
        \draw[->] (c6) -- (c7);
        \node[main] (c1b1) at (2,  1)         {$C_0$};
        \node[main] (c1g1) at (2.342, -0.940) {$C_1$};
        \node[main] (c1g2) at (1.658, -0.940) {$C_1$};
        \draw[->] (c1) -- (c1b1);
        \draw[->] (c1) -- (c1g1);
        \draw[->] (c1) -- (c1g2);
        \node[main] (c2b1) at (4.342, 0.940)  {$C_0$};
        \node[main] (c2b2) at (3.658, 0.940)  {$C_0$};
        \node[main] (c2g1) at (4.342, -0.940) {$C_1$};
        \node[main] (c2g2) at (3.658, -0.940) {$C_1$};
        \draw[->] (c2) -- (c2b1);
        \draw[->] (c2) -- (c2b2);
        \draw[->] (c2) -- (c2g1);
        \draw[->] (c2) -- (c2g2);
        \node[main] (c3b1) at (6.342, 0.940)  {$C_0$};
        \node[main] (c3b2) at (5.658, 0.940)  {$C_0$};
        \node[main] (c3g1) at (6, -1) {$C_1$};
        \draw[->] (c3) -- (c3b1);
        \draw[->] (c3) -- (c3b2);
        \draw[->] (c3) -- (c3g1);
        \node[main] (c4b1) at (8.342, 0.940)  {$C_0$};
        \node[main] (c4b2) at (7.658, 0.940)  {$C_0$};
        \draw[->] (c4) -- (c4b1);
        \draw[->] (c4) -- (c4b2);
        \node[main] (c6b1) at (10.342, 0.940)  {$C_0$};
        \node[main] (c6b2) at (9.658, 0.940)   {$C_0$};
        \draw[->] (c6) -- (c6b1);
        \draw[->] (c6) -- (c6b2);

        \draw [dashed, gray] (1.1, 1.5) -- (2.9, 1.5) -- (2.9, -1.5) -- (1.1, -1.5) -- (1.1, 1.5);
        \draw [dashed, gray] (3.1, 1.5) -- (4.9, 1.5) -- (4.9, -1.5) -- (3.1, -1.5) -- (3.1, 1.5);
        \draw [dashed, gray] (5.1, 1.5) -- (6.9, 1.5) -- (6.9, -1.5) -- (5.1, -1.5) -- (5.1, 1.5);
        \draw [dashed, gray] (7.1, 1.5) -- (8.9, 1.5) -- (8.9, -1.5) -- (7.1, -1.5) -- (7.1, 1.5);
        \draw [dashed, gray] (9.1, 1.5) -- (10.9, 1.5) -- (10.9, -1.5) -- (9.1, -1.5) -- (9.1, 1.5);
    \end{tikzpicture}
    \centering
    \caption{An example of the encoding of the computation of the two-counter machines to directed graphs.}
    \label{fig:two_counter_machine}
\end{figure}

The reduction is by encoding the computation of a two-counter machine into the graph directly.
We have illustrated it in Fig.~\ref{fig:two_counter_machine}.
Each vertex is labeled with exactly one of $S, T, C_0, C_1, Q_1, \cdots, Q_n$.
Informally, $S$ and $T$ will be indicators for the beginning and the end of the computation.
Each $Q_i$ will represent the particular state of the configuration.
The value of the counters $i$ is encoded by the number of outgoing neighbors with label $C_i$.

\begin{proof}
    Let $\tau$ be the vocabulary consisting of unary predicates
    $\set{S, T, C_0, C_1} \cup \set{Q_q}_{q \in \intsinterval{n}}$
    and a binary predicate $E$.

    We first define some useful gadgets.
    \begin{equation*}
        \begin{aligned}
            \psi^{\itdiff}_{i, \delta}(x)\ :=\ &
                \left(
                    \presbzy{E(x, z)\land E(z, y) \land C_i(y)} - 
                    \presby{E(x, y) \land C_i(y)}  = \delta
                \right) \\
            \psi^{\itzero}_{i}(x)\ :=\ &
                \left(
                    \presby{E(x, y) \land C_i(y)}  = 0
                \right) \\
            \psi^{\itsucc}_{j}(x)\ :=\ &
                \left(
                    \presby{E(x, y) \land Q_j(y)} = 1
                \right) \\
            Q(x)\ :=\ & \bigvee_{q \in \intsinterval{n}} Q_q(x) \\
            C(x)\ :=\ & C_1(x) \lor C_2(x)
        \end{aligned}
    \end{equation*}

    Next, for $1 \le q \le n$, we define the $\blMtwoPtwo$ formula $\psi_q(x)$ depending on the $q^{th}$ instruction of $\cM$.
    \begin{itemize}
        \item If the $q^{th}$ instruction is $\tcminc{i}$,
        then $\psi_q(x) :=
            \psi^{\itsucc}_{q+1}(x) \land
            \psi^{\itdiff}_{i, 1}(x) \land
            \psi^{\itdiff}_{1-i, 0}(x)$.
       
        \item If the $q^{th}$ instruction is $\tcmif{i}{j}$,
        then 
        \begin{equation*}
            \begin{aligned}
                \psi_q(x)\ :=\ &
                    \left(\psi^{\itzero}_i(x) \to \psi^{\itif}_q(x)\right) \land
                    \left(\neg \psi^{\itzero}_i(x) \to \psi^{\itelse}_q(x)\right) \\
                \psi^{\itif}_q(x)\ :=\ &
                    \psi^{\itsucc}_{j}(x) \land
                    \psi^{\itdiff}_{0, 0}(x) \land
                    \psi^{\itdiff}_{1, 0}(x) \\
                \psi^{\itelse}_q(x)\ :=\ &
                    \psi^{\itsucc}_{q+1}(x) \land
                    \psi^{\itdiff}_{i, -1}(x) \land
                    \psi^{\itdiff}_{1-i, 0}(x).
            \end{aligned}
        \end{equation*}
        
        \item If the $q^{th}$ instruction is $\tcmhalt$,
        then $\psi_q(x) := \left(\presby{E(x, y) \land T(y)} = 1\right)$. 
    \end{itemize}
    Finally, we define $\Psi_\cM(x)$ as follows.
    \begin{equation*}
        \begin{aligned}
            \varphi_1(x)\ :=\ & 
                \left(
                    \bigvee_{U \in \set{S, T, C_0, C_1} \cup \set{Q_i}_{i \in \intsinterval{n}}}U(x)
                \right) \land
                \left(
                    \bigwedge_{\substack{U_1, U_2 \in \set{S, T, C_0, C_1} \cup \set{Q_i}_{i \in \intsinterval{n}}\\U_1 \neq U_2}} \!\!\!\!\!\!\!\!
                    \neg U_1(x) \lor \neg U_2(x)
                \right) \\
            \varphi_2(x)\ :=\ &
                % \varphi^\itout_1(x) 
                \left(\presby{E(x, y)\land\top} = 1\right) \land
                \left(\presby{E(x, y) \land Q_1(y) \land \psi^{\itzero}_0(y) \land \psi^{\itzero}_1(y)} = 1\right) \\
            \varphi_3(x)\ :=\ &
                \left(\presby{E(x, y)\land S(y)} = 1\right) \land
                \left(\presby{E(x, y)\land \left(T(y) \lor C(y) \lor Q(y)\right)} = 0\right) \\
            \varphi_4(x)\ :=\ &
                \left(\presby{E(x, y)\land (Q(y)\lor T(y))} = 1\right) \land
                \left(\presby{E(x, y)\land S(y)} = 1\right) \land \\
                &
                \left(\presby{E(y, x)\land \left(S(y) \lor Q(y)\right)} = 1\right) \land 
                \left(\presby{E(y, x)\land \left(T(y) \lor C(y)\right)} = 0\right) \\
            \Psi_{\cM}(x)\ :=\ &
                \varphi_1(x) \land
                \left(S(x) 
                    \to \varphi_2(x)\right) \land
                \left(\left(C(x) \lor T(x)\right) 
                    \to \varphi_3(x)\right)\land 
                \left(Q(x) 
                    \to \varphi_4(x)\right)\land \\
                &\bigwedge_{q \in \intsinterval{n}} \left(Q_q(x) \to \psi_q(x)\right)
        \end{aligned}
    \end{equation*}
    
    \textbf{\underline{If $\forall x\ \Psi_{\cM}(x)$ is satisfied by the finite graph $\cG$}},
    then we will argue that $\cM$ halts.
    For every vertex $v \in V$,
    because $\cG \models \varphi_1(v)$, $v$ is in exactly one of $S$, $T$, $C_0$, $C_1$, or $Q_q$ for $1 \le q \le n$.
    We say that $v$ is an $S$ vertex if it is in $S$, and
    $u$ is an $S$ outgoing neighbor of $v$ if $u$ is an $S$ vertex and $u$ is an outgoing neighbor of $v$.
    We use a similar naming convention in the other cases.

    Recall that a configuration of $\cM$ is a tuple $\tuple{q, c_0, c_1}$, where $1 \le q \le n$ and $c_0, c_1 \in \bbN$. 
    For a $Q_q$ vertex $v$, where $1 \le q \le n$,
    the configuration realized by $v$, denoted by $\conf(v)$, is defined as $\tuple{q, c_0, c_1}$,
    where $c_0 = \abs{\setc{u \in \cN_{\out}(v)}{\cG \models C_0(u)}}$ and
    $c_1 = \abs{\setc{u \in \cN_{\out}(v)}{\cG \models C_1(u)}}$. 
    %\michael{Maybe move closer to the definition of configuration, or move into a definition environment}\ch{rewritten}
    We first make a claim about the number of incoming and outgoing neighbors of $v$, giving conditions on this number which will depend on the type of $v$.
    \begin{itemize}
        \item
        If $v$ is an $S$ vertex,
        then $\cG \models \varphi_2(v)$,
        which implies that $v$ has only one outgoing neighbor $u$, where $u$ is a $Q$ vertex.
        Moreover, the configuration of $u$ is $\tuple{1, 0, 0}$.

        \item
        If $v$ is a $T$, $C_1$, or $C_2$ vertex,
        then $\cG \models \varphi_3(v)$,
        which implies that $v$ has only one outgoing neighbor $u$, where $u$ is an $S$ vertex.

        \item
        If $v$ is a $Q_q$ vertex, for $1 \le q \le n$,
        then $\cG \models \varphi_4(v)$,
        which implies that $v$ has exactly one $Q$ or $T$ outgoing neighbor,
        and exactly one $S$ outgoing neighbor.
        Additionally, $v$ has only one incoming neighbor $u$, where $u$ is either a $Q$ or $S$ vertex.

        Let $u$ be the $Q$ or $T$ outgoing neighbor of $v$,
        $\conf(v) = \tuple{q, c_1, c_2}$, and
        $\conf(u) = \tuple{q', c_1', c_2'}$.
        Note that
        if $\cG \models \psi^{\itzero}_{i}(v)$,
        then $c_i = 0$;
        if $\cG \models \psi^{\itdiff}_{i, \delta}(v)$,
        then $c'_i - c_i = \delta$;
        if $\cG \models \psi^{\itsucc}_j(v)$,
        then $q' = j$.
        We now establish a relationship between $\conf(v)$ and $\conf(u)$, depending on the $q^{th}$ instruction of $\cM$. 
        Note that $\cG \models \psi_q(v)$.
        \begin{itemize}
            \item If the $q^{th}$ instruction is $\tcminc{i}$,
            then 
            $\cG \models \psi^{\itsucc}_{q+1}(v)$,
            $\cG \models \psi^{\itdiff}_{i, 1}(v)$, and
            $\cG \models \psi^{\itdiff}_{1-i, 0}(v)$.
            Thus
            $q' = q+1$, $c'_i = c_i+1$, and $c'_{1-i} = c_{1-i}$.
            Hence $\conf(u)$ is the successor configuration of $\conf(v)$.
            
            \item If the $q^{th}$ instruction is $\tcmif{i}{j}$,
            if $c_i = 0$,
            then
            $\cG \models \psi^{\itsucc}_{j}(v)$,
            $\cG \models \psi^{\itdiff}_{0, 0}(v)$, and
            $\cG \models \psi^{\itdiff}_{1, 0}(v)$.
            Thus $q' = j$, $c'_0 = c_0$, and $c'_1 = c_1$.
            Otherwise, if $c_i > 0$,
            then
            $\cG \models \psi^{\itsucc}_{q+1}(v)$,
            $\cG \models \psi^{\itdiff}_{i, -1}(v)$, and
            $\cG \models \psi^{\itdiff}_{1-i, 0}(v)$.
            Thus $q' = q+1$, $c'_i = c_i-1$, and $c'_{1-i} = c_{1-i}$.
            Hence $\conf(u)$ is the successor configuration of $\conf(v)$.
            
            \item If the $q^{th}$ instruction is $\tcmhalt$,
            then $\cG \models \left(\presby{E(v, y) \land T(y)} = 1\right)$,
            which implies that $u$ is a $T$ vertex.
        \end{itemize}

    \end{itemize}
    
    We now claim that $\cM$ halts.
    Our definition of graph requires the vertex set 
    $V$ to be nonempty.
    For a vertex in $V$, 
    if it is not a $S$ vertex,
    then it has an $S$ outgoing neighbor.
    Therefore there is at least one $S$ vertex in $V$.
    Let $v_s$ be an $S$ vertex in $V$.
    By the claim above,
    it has exactly one outgoing neighbor $v_1$,
    and for this node we have $\conf(v_1) = \tuple{1, 0, 0}$. 
    %michael: not "such that", since then it seems like we are saying
    %it has only one with this property
    
    We demonstrate that there exists a finite sequence $v_1, v_2, \ldots, v_{\ell+1}$,
    where for $2 \le i \le \ell$, $v_{i}$ is a $Q$ outgoing neighbor of $v_{i-1}$,
    and $v_{\ell+1}$ is a $T$ outgoing neighbor of $v_\ell$.
    Note that each $Q$ vertex has exactly one $Q$ or $T$ outgoing neighbor.
    If there were no such finite sequence, then the length of the maximal sequence of $Q$ nodes  starting from $v_1$ is infinite. Then, since $\cG$ is a finite graph,
    there must exist vertices $v_j$ and $v_{j'}$ in the sequence such that $j \neq j'$ and $v_j = v_{j'}$. 
    Without loss of generality, let $j$ and $j'$ be the smallest such pair with $j < j'$.
    If $j = 1$, then $v_1$ would have two incoming neighbors $v_s$ and $v_{j-1}$,
    but $v_1$ has only one incoming neighbor. Hence we have a contradiction.
    If $j > 1$, then $v_j$ would have two incoming neighbors $v_{j-1}$ and $v_{j'-1}$,
    which again leads to a contradiction. Thus we can always find such $T$ vertex $v_{\ell+1}$.

    By the claim above, for $2 \le i \le \ell$,
    $\conf(v_{i})$ is the successor configuration of $\conf(v_{i-1})$.
    In addition, $\conf(v_1) = \tuple{1, 0, 0}$ and $\conf(v_\ell)$ is a halt configuration.
    Therefore $\conf(v_1), \conf(v_2), \ldots, \conf(v_\ell)$
    is a computation of $\cM$. Because the length of the computation is finite, $\cM$ halts.
    
    \textbf{\underline{Suppose that $\cM$ halts}}.
    Let $\tuple{q_0, c_{0, 0}, c_{1, 0}}, \cdots, \tuple{q_\ell, c_{0, \ell}, c_{1, \ell}}$ be the computation of $\cM$.
    We define the graph $\cG$ as follows.
    Let $S := \set{v_s}$, $S := \set{v_t}$,
    $C_0 := \setc{v_{0, i, j}}{1 \le i \le \ell, 1 \le j \le c_{0, i}}$,
    $C_1 := \setc{v_{1, i, j}}{1 \le i \le \ell, 1 \le j \le c_{1, i}}$, and
    $Q_j := \setc{v_i}{1 \le i \le \ell, q_i = j}$ for $1 \le j \le n$.
    The set of vertices is $S \cup C_0 \cup C_1 \bigcup_{j \in \intsinterval{n}} Q_j$, and
    the set of edges is
    \begin{equation*}
        \begin{aligned}
            E\ :=\ &\set{(v_s, v_1), (v_\ell, v_t)} \cup
            \setc{(v_i, v_{i+1})}{1 \le i \le \ell-1} \cup \\
            &\setc{(v_i, v_{0, i, j})}{1 \le i \le \ell, 1 \le j \le c_{0, i}} \cup
            \setc{(v_i, v_{1, i, j})}{1 \le i \le \ell, 1 \le j \le c_{1, i}}.
        \end{aligned}
    \end{equation*}
    It is then straightforward to check that $\cG$ is a finite graph,
    as depicted in Fig.~\ref{fig:two_counter_machine},
    and that it is a model of $\forall x\ \Psi_\cM(x)$.
\end{proof}

Since the halting problem of two-counter machines is undecidable,
and $\blMtwoPtwo$ formulas can be translted to $\blreluGNN$s,
we obtain the undecidability of the universal satisfibility problem of $\blreluGNN$, by reduction from $\blMtwoPtwo$.

\begin{theorem}\label{thm:blgnn_unbounded_undec}
    The universal satisfiability problem of $\blreluGNN$s is undecidable.
\end{theorem}

% \ch{it is trivial}
% \begin{proof}
%     For every two-counter machine $\cM$, 
%     by Lemma~\ref{lemma:tcm_to_mtwoptwo},
%     there exists a $\blMtwoPtwo$ formula $\Psi_\cM(x)$ such that
%     the machine halts if and only if $\forall x\ \Psi_\cM(x)$ is finitely satisfiable.
%     By Theorem~\ref{thm:logic_to_unbounded_gnn},
%     there exists a $\blreluGNN$ $\cA_{\Psi_\cM}$ such that $\Psi_\cM(x)$ and $\cA_{\Psi_\cM}$ are equivalent. We claim that $\cM$ halts if and only if $\cA_{\Psi_\cM}$ is universally satisfiable. From this it would immediately follow that the universal satisfiability problem of $\blreluGNN$s is undecidable.

%     It is sufficient to show that $\forall x\ \Psi_\cM(x)$ is finitely satisfiable if and only if $\cA_{\Psi_\cM}$ is universally satisfiable.
%     If $\forall x\ \Psi_\cM(x)$ is finitely satisfiable, then there exists a graph $\cG$ such that for every vertex $v \in V$, $\tuple{\cG, v} \models \Psi_\cM(x)$.
%     By the equivalence between $\Psi_\cM(x)$ 
%     and $\cA_{\Psi_\cM}$, each $\tuple{\cG, v}$ satisfies $\cA_{\Psi_\cM}$. Then, by definition, $\cA_{\Psi_\cM}$ is universally satisfiable by $\cG$.
%     If $\cA_{\Psi_\cM}$ is universally satisfiable, then there exists a graph $\cG$ such that for every vertex $v \in V$,  $\tuple{\cG, v} \models \cA_{\Psi_\cM}$.
%     By the equivalence, $\tuple{\cG, v} \models\Psi_\cM(x)$. %michael: equivalent is an adjective, equivalence is a noun
%     Hence by the definition, $\cG \models \forall x\ \Psi_\cM(x)$,
%     which implies that $\forall x\ \Psi_\cM(x)$ is finitely satisfiable.
% \end{proof}

\subsection{Decidability of satisfiability for ``modal'' GNNs with unbounded activation functions} \label{subsec:dec_unbounded}

% Thus the situation for universal satisfiability contrasts with the eventually constant case. What about the \emph{satisfiability problem}?
We can see that even simple unbounded activation functions produce unbounded spectra,
so the proof technique in the truncated case certainly will not work.

\begin{remark}
    Even simple GNNs may have unbounded spectra. 
    For example, let $\cA$ be a $1$-layer $1$-$\olreluGNN$ defined as follows:
    the dimensions are $\gnndim{0} = \gnndim{1} = 1$;
    the coefficient matrix $\coefC{1}$ is a zero matrix;
    $\left(\coefA{1}{\out}\right)_{1,1} = 1$;
    the bias vector $\coefb{1}$ is a zero vector.
    It is not difficult to see that $\feat{1}_1(v)$ is the number of out-neighbors of $v$.
    Hence, the $1$-spectrum of $\cA$ is the set of natural numbers.
\end{remark}

We present a decidability result for the ``modal version'': aggregation over nodes connected by outgoing edges only, within a directed graph:
\begin{theorem} \label{thm:outputonlydecidability}
    The satisfiability problem of $\olpwGNN$s is $\nexp$-complete,
    and it becomes $\np$-complete when the number of layers is fixed.
\end{theorem}

Analogously to what we did in the eventually constant case, we describe all the possible values of a given activation function.
Unlike in the eventually constant case, this will not be a finite set, but it will be \emph{semi-linear}: that is, describable using a formula of Presburger arithmetic.
Furthermore, the this set is definable by a polynomial size \emph{existential Presburger arithmetic with stars} ($\epastar$) formula, which was studied by~\cite{kleenestar}.

\begin{definition}
    We give the syntax of \emph{existential Presburger arithmetic with stars} ($\epastar$).
    \begin{itemize}
        \item For $m, n \in \bbN^+$, $\bfx$ be a tuple of $n$ variables,
        $\bfA \in \bbZ^{m \times n}$, and $\bfb \in \bbZ^{m}$,
        $\bfA\bfx \ge \bfb$ is an $\epastar$ formula.
        \item If $\varphi(\bfx)$ and $\psi(\bfy)$
        are $\epastar$ formulas,
        then so is
        $\varphi(\bfx) \land \psi(\bfy)$ and 
        $\varphi(\bfx) \lor \psi(\bfy)$.
        \item If $\varphi(\bfx)$ is an $\epastar$ formula and $\bfy \subseteq \bfx$,
        then so is $\exists \bfy\ \varphi(\bfx)$.
        \item If $\varphi(\bfx)$ is an $\epastar$ formula,
        then so is $\psi(\bfx) := \mstar{\varphi(\bfx)}$.
    \end{itemize}
    The semantics of the linear inequalities, Boolean connectives, and existential quantifiers is as usual,
    and the semantics of the Kleene star is defined as follows.
    Considering $\epastar$ formula $\psi(\bfx) := \mstar{\varphi(\bfx)}$ and $\bfa \in \bbZ^n$,
    we say that $\psi(\bfa)$ holds if there exists
    $\bfa_1, \ldots, \bfa_k \in \bbZ^n$ such that
    $\varphi\left(\bfa_t\right)$ holds and
    $\bfa = \sum_{t \in \intsinterval{k}} \bfa_i$.
\end{definition}

\begin{definition} \label{def:starheight}
    The \emph{star height} of an $\epastar$ formula
    is the depth of nesting of the Kleene star operator. It is zero for formulas that
    have no Kleene star, and the star height of the formula $\mstar{\varphi(\bfx)}$ is the star height
    of the formula $\varphi(\bfx)$ plus 1.
\end{definition}

%\ch{TODO: rational coef, def of star height, and $\circledast$}
%\michael{I do not understand the reference to $\circledast$: in Presburger logics it is already defined} ch:yes, thanks
\begin{theorem}[\cite{kleenestar}, Theorem III.1.] \label{thm:epastar}
    The satisfiability problem of $\epastar$ is $\nexp$-complete,
    and it becomes $\np$-complete for fixed star height.
\end{theorem}

% Let $\cA \subseteq \bbN^n$ and $\bfx$ be a tuple of $n$ variables.
We say that a set $\cC \subseteq \bbQ^n$ is defined by the $\epastar$ formula $\varphi(\bfx)$ with common denominator $D \in \bbN$ if the following properties hold.
\begin{itemize}
    \item For every $\bfa \in \cC$, $D\bfa \in \bbZ^n$.
    % \michael{modified, please check}\ch{fixed, thanks}
    \item $\bfa \in \cC$ if and only if $\varphi(D\bfa)$ holds.
\end{itemize}

\begin{definition}
    A \emph{piecewise linear functions} is defined
    by a finite sequence $(I_1, f_1), \ldots, (I_k,f_k)$
    where $\bigcup_{t \in \intsinterval{k}} I_t$ is a partition of $\bbQ$ into $k$ intervals
    and each $f_t: \bbQ \to \bbQ$ is an affine transformation.
    % The sequence $((I_1, f_1), \cdots, (I_p,f_p))$ defines a function
    % where $x$ is mapped to $f_i(x)$ if $x$ is in the interval $I_i$.
    For every $q \in \bbQ$,
    $q$ is mapped to $f_t(q)$ if $q$ is in the interval $I_t$.
\end{definition}
For every piecewise linear function $f$ and $D \in \bbN$,
it is not difficult to show that there exist an $\epastar$ formula $\Phi_{f, D}(x, y)$,
called the \emph{characteristic formula of $f$ with common denominator $D$},
such that for every $p, q \in \bbZ$, $\Phi_{f, D}(p, q)$ holds if and only if $q / D = f(p / D)$.

We now state our representation theorem, which immediately implies 
the upper bound Theorem~\ref{thm:outputonlydecidability}.
Recall that given a GNN $\cA$,
for graph $\cG$ and vertex $v$ in $\cG$,
for $0 \leq \ell \leq L$,
the $\ell$-history of $v$ (w.r.t. $\cA$), defined in Definition~\ref{def:hist},
is the tuple that collects the first $\ell+1$ feature vectors of $v$.
The \emph{$\ell$-history-space} of $\cA$,
denoted by $\hsp{\ell}$,
is the set that collects $\ell$-history over all possible graph and vertices.

We contrast the theorem with Theorem \ref{thm:computespectrum}. There we could only overapproximate the spectrum, because we could not determine which numbers from previously layers were simultaneously realizable.
By inductively maintaining the entire history at each node, we have enough information to resolve these questions of consistency, and compute an \emph{exact} representation of the semantic object, not just an overapproximation.

\begin{definition}
    For an $\olpwGNN$ $\cA$, the \emph{common denominator of $\cA$}, denoted by $D_\cA$,
    is the least common denominator of denominator of coefficients matrices and bias vectors in $\cA$ and coefficients in piecewise linear activation functions.
\end{definition}

\begin{theorem} \label{thm:semilinearoutgoingonlyunbounded}
    For every $\olpwGNN$ $\cA$ and $0 \le \ell \le L$,
    the $\ell$-history-space $\hsp{\ell}$ of $\cA$ is $\epastar$ definable with common denominator $D_\cA$.
    Moreover, the size of the formula is polynomial in the description of $\cA$,
    and the star height of the formula is the number of layers of $\cA$.
\end{theorem}

\begin{proof}
    We define the $\epastar$ formula $\varphilp{\ell}{\varx{0}, \ldots, \varx{\ell}}$ inductively.
    For the base case $\ell = 0$,
    \begin{equation*}
        \varphilp{0}{\varx{0}}\ :=\ 
        \bigwedge_{i \in \intsinterval{\gnndim{0}}} \left(\varx{0}_i = 0\ \lor\ \varx{0}_i = D_\cA\right).
    \end{equation*}
    For $1 \le \ell \le L$, let
    \begin{equation*}
        \begin{aligned}
            \tlp{t}{\varx{t}, \varx{t-1}, \vary{t}}\ :=\ 
            &\exists \varz{t-1}\ \bigwedge_{i \in \intsinterval{\gnndim{t}}} \Phi_{\act{t}, D_\cA}\left(\varz{t}_i, \varx{t}_i\right)
            \ \land\ \\
            &\left(\frac{\varz{t}}{D_\cA} = \coefC{t}\left(\frac{\varx{t-1}}{D_\cA}\right) + \coefA{t}{\out}\left(\frac{\vary{t-1}}{D_\cA}\right) + \coefb{t}\right) \\
            \varphilp{\ell}{\varx{0}, \ldots, \varx{\ell}}\ :=\ 
            &\exists \vary{0} \cdots \vary{\ell-1}\ 
            \bigwedge_{t \in \intsinterval{\ell}} \tlp{t}{\varx{t}, \varx{t-1}, \vary{t-1}}
            \ \land\ \\
            &\varphilp{0}{\varx{0}} \land \mstar{\varphil{\ell-1}\left(\vary{0}, \ldots, \vary{\ell-1}\right)},
        \end{aligned}
    \end{equation*}

    %\michael{To me $\varz{\ell} / D_\cA$ looks like substitution (and similarly elsewhere) Can
   % we use frac to make clear it is division $\frac{\varz{\ell}}{D_\cA}$}
  %  \ch{rewritten}
    where $\Phi_{\act{\ell}, D_\cA}(x, z)$ is the characteristic formula of the piecewise linear function $\act{\ell}$ with common denominator $D_\cA$.
    Note that the size of $\varphilp{\ell}{\varx{0}, \ldots, \varx{\ell}}$ is polynomial in the description of $\cA$,
    and its star height is the number of layers of $\cA$.
    
    We now prove that $\varphilp{\ell}{\varx{0}, \ldots, \varx{\ell}}$
    defines $\hsp{\ell}$ with common denominator $D_\cA$, by induction on $\ell$.
    The base case $\ell = 0$ is trivial.
    For the induction hypothesis,
    we assume that the lemma holds for $\ell-1$.
    \textbf{\underline{Suppose that $\bfh \in \hsp{\ell}$}}.
    By the definition of $\hsp{\ell}$,
    there exists a graph $\cG$ and vertex $v \in V$ such that 
    $\bfh = \hist{\ell}(v)$.
    For $0 \le t \le \ell$, let
    \begin{equation*}
        \begin{aligned}
            \bfal{t}\ :=\ &D_\cA\bfh[t] \\
            \bfal{t}_\out\ :=\ &\sum_{u \in \nbr{\out}(v)} D_\cA\hist{\ell}(u)[t] \\
            \bfal{t}_\pre\ :=\ &
            D_\cA\left(\coefC{t}\bfh[t-1] + \coefA{t}{\out}\sum_{u \in \nbr{\out}(v)} \hist{\ell}(u)[t-1] + \coefb{t}\right).
        \end{aligned}
    \end{equation*}
    It is obvious that the following relation between
    $\bfal{t}$, $\bfal{t}_\out$, and $\bfal{t}_\pre$ holds:
    \begin{equation*}
        \frac{\bfal{t}_\pre}{D_\cA}\ =\ \coefC{t}\left(\frac{\bfal{t-1}}{D_\cA}\right) + \coefA{t}{\out}\left(\frac{\bfal{t-1}_\out}{D_\cA}\right) + \coefb{t}.
    \end{equation*}
    By the semantics of these GNNs, for $1 \le t \le \ell$, it holds that
    \begin{equation*}
        \bfh[t]\ =\ 
        \actp{t}{\coefC{t}\bfh[t-1] + \coefA{t}{\out}\sum_{u \in \nbr{\out}(v)} \hist{\ell}(u)[t-1] + \coefb{t}},
    \end{equation*}
    which implies that
    $\frac{\bfal{t}}{D_\cA} = \actp{t}{\bfal{t}_\pre / D_\cA}$.
    Since $\Phi_{\act{\ell}, D_\cA}(x, z)$ is the characteristic formula of $\act{\ell}$ with common denominator $D_\cA$,
    it follows that $\Phi_{\act{\ell}, D_\cA}\left(\bfal{t}_\pre, \bfal{t}\right)$ holds.
    Thus, it holds that $\tlp{t-1}{\bfal{t}, \bfal{t-1}, \bfal{t-1}_\out}$.
    % Furthermore, for every $1 \le i \le \ell-1$, $\feat{i}(u) = \left(\hist{\ell-1}(u)\right)[i]$.

    By the induction hypothesis, $\hsp{\ell-1}$ is defined by
    $\varphilp{\ell-1}{\varx{0}, \ldots, \varx{\ell-1}}$ with common denominator $D_\cA$.
    For every $u \in \nbr{\out}(v)$, note that $\hist{\ell-1}(u) \in \hsp{\ell-1}$.
    Therefore, we have
    $\varphilp{\ell-1}{D_\cA\feat{0}(u), \ldots, D_\cA\feat{\ell-1}(u)}$,
    which implies that
    $\mstar{\varphilp{\ell-1}{\bfbl{0}, \ldots, \bfbl{\ell-1}}}$ holds.
    Thus, it holds that $\varphilp{\ell}{D_\cA\bfh[0], \ldots, D_\cA\bfh[\ell]}$.

    \textbf{\underline{For every $\bfh$ satisfing $\varphilp{\ell}{D_\cA\bfh[0], \ldots, D_\cA\bfh[\ell]}$}},
    there exist $\bfh_1, \ldots, \bfh_k$,
    such that, for $1 \le t \le k$,
    $\varphilp{\ell-1}{D_\cA\bfh_t[0], \ldots, D_\cA\bfh_t[\ell-1]}$ holds.
    By the induction hypothesis, there exists graph $\cG_t$ and vertex $v_t$,
    such that $\hist{\ell-1}_{\cG_t}(v_t) = \bfh_t$.
    Let $\cG$ be the graph obtained by taking the disjoint union of $\cG_1, \ldots, \cG_k$
    and adding a fresh vertex $v_r$.
    We set the colors of $v_r$ so that $\feat{0}_\cG(v_r) = \bfh[0]$.
    For $1 \le t \le k$, there exists an edge from $v_r$ to $v_t$.
    It is routine to verify that the $\ell$-history of $v_r$ is precisely $\bfh$.
    Note that \emph{because $\cA$ is outgoing-only, the edge from $v_r$ to $v_t$ has no effect on the $(\ell-1)$-history of $v_t$}.
    Thus $\hist{\ell-1}_{\cG}(v_t) = \hist{\ell-1}_{\cG_t}(v_t)$.
\end{proof}

\begin{corollary}
    The satisfiability problem of $\olpwGNN$s is in $\nexp$,
    and it becomes $\np$ when the number of layers is fixed.
\end{corollary}

\begin{proof}
    Let $\cA$ be an $\olpwGNN$.
    By Throrem~\ref{thm:semilinearoutgoingonlyunbounded},
    there exists the $\epastar$ formula $\varphilp{L}{\varx{0}, \ldots, \varx{L}}$
    define the $L$-history-space of $\cA$ with common denominator of $\cA$.
    Combining the acceptance condition of GNNs,
    $\cA$ is satisfiable if and only if the $\epastar$ sentence
    \begin{equation*}
        \Phi\ :=\ \exists \varx{0} \cdots \varx{L}, \varphilp{L}{\varx{0}, \ldots, \varx{L}}
        \ \land\ \left(2\varx{L}_1 \ge D_\cA\right)
    \end{equation*}
    is satisfiable.
    Note that we can construct $\varphilp{L}{\varx{0}, \ldots, \varx{L}}$ in polynomial time in the description of $\cA$, and thus the same is true of $\Phi$. 
    Since the satisfiability problem of $\epastar$ is in $\nexp$ by Theorem \ref{thm:epastar},
    the satisfiability problem of $\olpwGNN$s is also in $\nexp$.
    The star height of $\varphilp{L}{\varx{0}, \ldots, \varx{L}}$ is the number of layers of $\cA$, and thus the same holds for  $\Phi$. 
    %\michael{As above, I don't understand "and so is $\Phi$''}
    Therefore the satisfiability problem of $\olpwGNN$s is in $\np$ when the number of layers is fixed.
\end{proof}

The hardness of the satisfiability problem for $\olpwGNN$s is by reducing
from the satisfiability problem of $\epastar$ formulas,
which is known to be $\nexp$-hard from Theorem~\ref{thm:epastar}.
% \michael{which is known to be $\nexptime$-hard from ... }
Note that fixed layer $\oltrreluGNN$s can be simulated by fixed layer $\olreluGNN$s with a linear blowup.
By Theorem~\ref{thm:pspace}, the satisfiability problem of fixed layer $\oltrreluGNN$s is $\np$-hard, and thus the same holds for fixed layer $\olreluGNN$s.

%\michael{the $\np$-hardness is inherited from the results about truncated relu?}
%\ch{yes}
\begin{theorem}\label{thm:epastar_to_gnn}
    The satisfiability problem of $\olreluGNN$s is $\nexp$-hard.
\end{theorem}

The proof works by considering $\epastar$ formulas in normal form.

\begin{definition}
    We say that a quantifier-free $\epastar$ formula $\varphi(\bfx)$ is in \emph{normal form}
    if it is either star-free or it is in the form $\psi_1(\bfx) \land \mstar{\psi_2(\bfx)}$,
    where $\psi_1(\bfx)$ is quantifier-free and star-free while
    $\psi_2(\bfx)$ is a quantifier-free $\epastar$ formula in normal form

    We say that an $\epastar$ formula $\varphi(\bfx)$ is in \emph{normal form}
    if it is in the form $\exists \bfy\ \psi(\bfx, \bfy)$,
    where $\psi(\bfx, \bfy)$ is a quantifier-free $\epastar$ formula in \emph{normal form}.
\end{definition}

Informally, normal form requires the occurrences of Kleene star to be layered one on top of the other. So, for example, a formula of the form $\mstar{\phi_1(x)} \land \mstar{\phi_2(y)}$ is \emph{not} in
normal formula. We can rewrite it in normal form as:
\begin{equation*}
    \mstar{(x = 0 \lor \phi_1(x)) \land (y = 0 \lor \phi_2(y))}.
\end{equation*}
For translating to GNNs, this normal form allows us to conveniently map occurrences
of GNN to aggregations at different levels.

\begin{lemma}
    There is a polynomial time algorithm that converts an $\epastar$ formula
    into a semantically equivalent $\epastar$ formula in normal form.
\end{lemma}

We say that a $\epastar$ formula $\mstar{\phi(\bfz)}$ is a top-level star subformula of the $\epastar$ formula $\varphi(\bfx)$,
if there is no $\epastar$ formula $\mstar{\psi(\bfy)}$ satisfied that
$\mstar{\phi(\bfz)}$ is a strict subformula $\mstar{\psi(\bfy)}$ and
$\mstar{\psi(\bfy)}$ is a subformula of $\varphi(\bfx)$.

\begin{proof}
    We first observe that for every $\epastar$ formula $\varphi(\bfx, \bfy)$,
    the formulas $\exists \bfy\ \mstar{\varphi(\bfx, \bfy)}$ and
    $\mstar{\exists \bfy\ \varphi(\bfx, \bfy)}$ 
    are semantically equivalent.
    Thus, with routine procedure, we assume that the formula is in the prefix normal form in any time.

    Next, for a quantifier-free $\epastar$ formula $\varphi(\bfx)$ with $m$ top-level star subformulas
    $\mstar{\phi_1(\bfx)}, \ldots, \mstar{\phi_m(\bfx)}$.
    We claim that $\varphi(\bfx)$ and $\exists \bfx_1 \cdots \bfx_m\ \psi(\bfx, \bfx_1, \ldots, \bfx_m)$
    are semantically equivalent, where
    $\bfx_1, \ldots, \bfx_m$ are fresh variables and
    \begin{equation*}
        \begin{aligned}
            \psi(\bfx, \bfx_1, \ldots, \bfx_m)\ :=\ &
            \psi_1(\bfx, \bfx_1, \ldots, \bfx_m)\ \land\ 
            \mstar{\psi_2(\bfx, \bfx_1, \ldots, \bfx_m)} \\
            \psi_1(\bfx, \bfx_1, \ldots, \bfx_m)\ :=\ &
            \varphi\left[\bfx=\bfx_1 / \mstar{\phi_1(\bfx)}\right]\ldots\left[\bfx=\bfx_m / \mstar{\phi_m(\bfx)}\right] \\
            \psi_2(\bfx, \bfx_1, \ldots, \bfx_m)\ :=\ &
            \bigwedge_{i \in \intsinterval{m}}
            \left(\bfx_i = \bfzero \lor \phi_i(\bfx_i)\right).
        \end{aligned}
    \end{equation*}
    Suppose that $\varphi(\bfx)$ is satisfiable by $\bfa$.
    For $1 \le i \le m$, let
    \begin{equation*}
        \bfa_i\ :=\ 
        \begin{cases}
            \bfa, &\text{if $\mstar{\psi(\bfa)}$ holds} \\
            \bfzero, &\text{otherwise}.
        \end{cases}
    \end{equation*}
    It is not difficult to check that $\psi(\bfa, \bfa_1, \ldots, \bfa_m)$ holds.

    On the other hand,
    suppose that $\psi(\bfa, \bfa_1, \ldots, \bfa_m)$ holds.
    For $1 \le i \le m$,
    by the definition of $\psi$,
    $\mstar{\phi_i(\bfa_i)}$ holds.
    Therefore, if $\bfa = \bfa_i$, then $\mstar{\psi_i(\bfa)}$ holds.
    Note that negations are not allowed in $\epastar$ formula.
    Since $\psi(\bfa, \bfa_1, \ldots, \bfa_m)$ holds, $\varphi(\bfa)$ also holds.

    Finally, note that the star height of the formula $\varphi_2(\bfx, \bfx_1, \ldots, \bfx_m)$ is one less than the origin formula $\varphi(\bfx)$.
    We can repeat this procedure and obtain the desired formula in the normal form.
\end{proof}

It is worth noting that the $\epastar$ formula $\exists \bfy\ \psi(\bfx, \bfy)$ is satisfiable
if and only if $\psi(\bfx, \bfy)$ is satisfiable.
Thus, it is sufficient to only consider quantifier-free $\epastar$ formulas in the normal form.

We are now ready to prove the hardness of the satisfiability problem of $\olreluGNN$s.
Let $\Psi(\bfx)$ be a quantifier-free $\epastar$ formula in normal form with $n$ variables and star height $k$, 
\begin{equation*}
    \Psi(\bfx) := \psi_k(\bfx) \land \mstar{\psi_{k-1}(\bfx) \land \mstar{\cdots \land \mstar{\psi_{0}(\bfx)}}},
\end{equation*}
where $\psi_{t}(\bfx)$ are quantifier-free and star-free,
Let $\Psi_{0}(\bfx) := \psi_{0}(\bfx)$, and,
for $1 \le t \le k$,
\begin{equation*}
    \Psi_{t}(\bfx)\ :=\ \psi_{t}(\bfx) \land \mstar{\Psi_{t-1}(\bfx)}.
\end{equation*}
Note that $\Psi(\bfx) = \Psi_k(\bfx)$.
For $0 \le t \le k$,
let $\ell_t := \sum_{i \in \intsinterval{t}} \left(2\depth{\psi_i} + 1\right)$.

We define the $(\ell_k + 1)$-layer $n$-$\olreluGNN$ $\cA_\Psi$ as follows.
The coefficient matrices and bias vectors are chosen so that the features of $\cA_\Psi$ will satisfy certain conditions.

For the first layer,
\begin{equation*}
    \feat{1}_i(v)\ =\ 
    \begin{dcases}
        \sum_{u \in \nbr{\out}(v)}\feat{0}_i(u), &\text{for $1 \le i \le n$} \\
        0, &\text{for $n < i \le n+2$} \\
        1, &\text{for $i = n+3$}.
    \end{dcases}
\end{equation*}

For layers between $2$ and $\ell_0$,
$\cA_\Psi$ test whether $\psi_0\left(\bfal{0}(v)\right)$ holds,
where $\bfal{0}(v)$ is the first $n$ entries of $\feat{1}(v)$.
That is
\begin{equation*}
    \feat{\ell_0}_i(v)\ =\ 
    \begin{cases}
        \feat{1}_i(v), &\text{for $1 \le i \le n+3$} \\
        \eval{\psi_0\left(\bfal{0}(v)\right)}, &\text{for $i = n + 4$}.
    \end{cases}
\end{equation*}

For $1 \le t \le k$,
for the $(\ell_{t} - 2\depth{\psi_t})$ layer,
\begin{equation*}
    \feat{\ell_{t} - 2\depth{\psi_t}}_i(v)\ =\ 
    \begin{dcases}
        \sum_{u \in \nbr{\out}(v)}\feat{\ell_{t} - 2\depth{\psi_t}-1}_i(u), &\text{for $1 \le i \le n$} \\
        \sum_{u \in \nbr{\out}(v)}\feat{\ell_{t} - 2\depth{\psi_t}-1}_{i+2}(u), &\text{for $n < i \le n+2$} \\
        1, &\text{for $i = n+3$}.
    \end{dcases}
\end{equation*}
For layers between $\ell_{t} - 2\depth{\psi_t} + 1$ and $\ell_t$,
$\cA_\Psi$ test whether $\psi_t\left(\bfal{t}(v)\right)$ holds,
where $\bfal{t}(v)$ is the first $n$ entries of $\feat{\ell_{t} - 2\depth{\psi_t}}(v)$.
That is
\begin{equation*}
    \feat{\ell_t}_i(v)\ =\ 
    \begin{cases}
        \feat{\ell_{t} - 2\depth{\psi_t}}_i(v), &\text{for $1 \le i \le n+3$} \\
        \eval{\psi_t\left(\bfal{t}(v)\right)\ \land\ 
        \feat{\ell_{t} - 2\depth{\psi_t}}_{n+1}(v) = \feat{\ell_{t} - 2\depth{\psi_t}}_{n+2}(v)},
        &\text{for $i = n + 4$}.
    \end{cases}
\end{equation*}

Finally, for the last layer, $\feat{\ell_k+1}_1(v) = \feat{\ell_k}_{n+4}(v)$.
Note that, for $1 \le t \le k$,
$\ell_t - \ell_{t-1} = 2\depth{\psi_t} + 1$.
Thus, these definitions cannot clash.
By Lemma~\ref{lemma:fnn}, it is easy to see that there are $\olreluGNN$s that satisfy these conditions.

One direction of the theorem will follows once we have shown the following property of $\cA_\Psi$:
\begin{lemma}
    For every quantifier-free $\epastar$ formula $\Psi(\bfx)$,
    let $\cA_\Psi$ be the $\olreluGNN$ defined above.
    For every $n$-graph $\cG$ and vertex $v$ in $\cG$,
    for $0 \le t \le k$,
    if $\feat{\ell_{t}}_{n+4}(v) = 1$,
    then $\Psi_t\left(\bfal{t}(v)\right)$ holds,
    where $\bfal{t}(v)$ is the first $n$ entries of $\feat{\ell_{t} - 2\depth{\psi_t}}(v)$.
\end{lemma}

\begin{proof}
    We prove the lemma by induction on the star height of the formula.
    For the base case $t = 0$, note that $\Psi_0(\bfx) = \psi_0(\bfx)$ is quantifier-free and star-free.
    Thus the property easily follows the definition of $\cA_\Psi$.
    
    For the induction step $1 \le t \le k$,
    we first observe that, for $1 \le i \le n$, 
    by the definition of $\cA_\Psi$,
    \begin{equation*}
        \feat{\ell_t - 2\depth{\psi_t}}_i(v)
        \ =\ \sum_{u \in \nbr{\out}(v)}\feat{\ell_{t-1}}_i(u)
        \ =\ \sum_{u \in \nbr{\out}(v)}\feat{\ell_{t-1} - 2\depth{\psi_{t-1}}}_i(u),
    \end{equation*}
    which implies that $\bfal{t}(v) = \sum_{u \in \nbr{\out}(v)} \bfal{t-1}(u)$.
    Next, by the definition of $\cA_\Psi$,
    if $\feat{\ell_{t}}_{n+4}(v) = 1$,
    then following properties follow.
    \begin{itemize}
        \item $\psi_t\left(\bfal{t}(v)\right)$ holds.
        \item $\feat{\ell_t - 2\depth{\psi_t}}_{n+1}(v) = \feat{\ell_t - 2\depth{\psi_t}}_{n+2}(v)$.
        Note that, by the definition of $\cA_\Psi$,
        \begin{equation*}
            \begin{aligned}
                \feat{\ell_t - 2\depth{\psi_t}}_{n+1}(v)
                \ =\ &\sum_{u \in \nbr{\out}(v)}\feat{\ell_{t-1}}_{n+3}(u)
                \ =\ \sum_{u \in \nbr{\out}(v)} 1 \\
                \feat{\ell_t - 2\depth{\psi_t}}_{n+2}(v)
                \ =\ &\sum_{u \in \nbr{\out}(v)}\feat{\ell_{t-1}}_{n+4}(u),
            \end{aligned}
        \end{equation*}
        which implies that, for every $u \in \nbr{\out}(v)$,
        $\feat{\ell_{t-1}}_{n+4}(u) = 1$.
        By the induction hypothesis, it follows that
        $\Psi_{t-1}\left(\bfal{t-1}(u)\right)$ holds.
        Since $\bfal{t}(v) = \sum_{u \in \nbr{\out}(v)} \bfal{t-1}(u)$,
        by the semantics of Kleene star, %michael: semantics, not semantic
        $\mstar{\Psi_{t-1}\left(\bfal{t}(v)\right)}$ holds.
    \end{itemize}
    Combining above two results, we obtain that
    $\Psi_{t}\left(\bfal{t}(v)\right)$ holds.
\end{proof}

We now prove the other direction.

\begin{lemma}
    For every quantifier-free $\epastar$ formula $\Psi(\bfx)$ in normal form,
    let $\cA_\Psi$ be the $\olreluGNN$ defined above.
    If $\Psi(\bfx)$ is satisfiable, then $\cA_\Psi$ is satisfiable.
\end{lemma}

\begin{proof}
    For $0 \le t \le k$,
    for every valid assignment $\bfa \in \bbN^n$ of $\Psi_t(\bfx)$,
    we define the $n$-tree $\rtree{t, \bfa}$ inductively.
    For the base case $t = 0$,
    the set of vertices in $\rtree{0, \bfa}$ is $\set{v_r} \cup \setc{v_i}{1 \le i \le \am}$, 
    where $\am$ is the maximum entry of $\bfa$.
    For $1 \le i \le n$, $U_{t} = \setc{v_j}{1 \le j \le a_i}$,
    where $a_i$ is the $i^{th}$ entries of $\bfa$.
    The edges are $E = \setc{(v_r, v_i)}{1 \le i \le \am}$.

    For the inductive step $1 \le t \le k$,
    recall that $\Psi_t(\bfx) = \psi_t(\bfx) \land \mstar{\Psi_{t-1}(\bfx)}$.
    Since $\bfa$ is a valid assignment of $\Psi_t(\bfx)$,
    there exists $\bfa_1, \ldots, \bfa_\ell$ satisfying that
    $\bfa = \sum_{i \in \intsinterval{\ell}} \bfa_i$ and
    $\Psi_{t-1}\left(\bfa_i\right)$ holds.
    The $n$-tree $\rtree{t, \bfa}$ is a disjoint union of
    $\rtree{t-1, \bfa_i}$ and a fresh vertex $v_r$.
    For $1 \le i \le \ell$, there exists an edge from $v_r$ to $v_i$, where
    $v_i$ is the root of $\rtree{t-1, \bfa_i}$.

    The correctness of the construction follows from the following claim:
    For $0 \le t \le k$
    for every valid assignment $\bfa \in \bbN^n$ of $\Psi_t(\bfx)$,
    it holds that 
    \begin{equation*}
        \feat{\ell_t}_{\rtree{t, \bfa}, n+4}(v_r) = t-1
        \quad \text{and} \quad
        \bfal{t}(v_r) = \bfa,
    \end{equation*}
    where $v_r$ is the root of $\rtree{t, \bfa}$
    and $\bfal{t}(v_r)$ consists of the first $n$ entries of $\feat{\ell_t}_{\rtree{t, \bfa}, n+4}(v_r)$.

    We prove the claim by induction on the star height of $\Psi(\bfx)$.
    For the base case $t = 0$,
    by the definition of $\rtree{0, \bfa}$,
    the first $n$ entries of $\feat{1}_{\rtree{0, \bfa}}(v_r)$
    are exactly $\bfa$.
    Since $\bfa$ is a valid assignment of $\Psi_0(\bfx)$,
    it holds that
    \begin{equation*}
        \feat{\ell_0}_{\rtree{0, \bfa}, n+4}(v_r)
        \ =\ \eval{\Psi_0\left(\bfal{0}(v_r)\right)}
        \ =\ \eval{\Psi_0\left(\bfa\right)}
        \ =\ 1.
    \end{equation*}

    For the inductive step $1 \le t \le k$,
    recall that $\Psi_t(\bfx) = \psi_t(\bfx) \land \mstar{\Psi_{t-1}(\bfx)}$.
    Let $v_r$ be the root of $\rtree{t, \bfa}$.
    Since $\bfa$ is a valid assignment of $\Psi_t(\bfx)$,
    there exists $\bfa_1, \ldots, \bfa_\ell$ satisfying that
    $\bfa = \sum_{i \in \intsinterval{\ell}} \bfa_i$ and
    $\Psi_{t-1}\left(\bfa_i\right)$ holds.
    Let $v_i$ be the root of $\rtree{t-1, \bfa_i}$.
    Note that since $\cA_\Psi$ is outgoing-only,
    the edge from $v_r$ to $v_i$ did not affect the feature of $v_i$.
    Thus, it holds that
    \begin{equation*}
        \feat{\ell_{t-1}}_{\rtree{t-1, \bfa_i}}(v_i)
        \ =\ 
        \feat{\ell_{t-1}}_{\rtree{t, \bfa_i}}(v_i).
    \end{equation*}
    \begin{itemize}
        \item For $1 \le i \le \ell$,
        by the induction hypothesis,
        the first $n$ entries of $\feat{\ell_{t-1}}_{\rtree{t-1, \bfa_i}}(v_i)$ are exactly $\bfa_i$.
        By the definition of $\cA_\Psi$,
        for $1 \le j \le n$, it follows that
        \begin{equation*}
            \feat{\ell_t}_{\rtree{t, \bfa}, j}(v_r)
            \ =\ 
            \sum_{u \in \nbr{\out}(v_r)}\feat{\ell_{t-1}}_{\rtree{t, \bfa}, j}(u)
            % \ =\ 
            % \sum_{i \in \intsinterval{\ell}}\feat{\ell_{t-1}}_{\rtree{t, \bfa}, j}(v_i)
            \ =\ 
            \sum_{i \in \intsinterval{\ell}}\feat{\ell_{t-1}}_{\rtree{t-1, \bfa_i}, j}(v_i),
        \end{equation*}
        which implies that the first $n$ entries of 
        $\feat{\ell_t}_{\rtree{t, \bfa}}(v_r)$ are $\bfa$.

        % Since $\bfa$ is a valid assignment of $\Psi_t(\bfx)$,
        % \begin{equation*}
        %     \feat{\ell_t}_{\rtree{\Psi_{t}\left(\bfa\right)}, n+4}(v_r)
        %     \ =\ \eval{\Psi_t\left(\bfa\right)}
        %     \ =\ 1.
        % \end{equation*}

        \item
        By the definition of $\cA_\Psi$, it holds that
        \begin{equation*}
            \begin{aligned}
                \feat{\ell_t}_{\rtree{t, \bfa}, n+1}(v_r)
                \ =\ &
                \sum_{u \in \nbr{\out}(v_r)}\feat{\ell_{t-1}}_{\rtree{t, \bfa}, n+3}(u)
                % \ =\ 
                % \sum_{i \in \intsinterval{\ell}}\feat{\ell_{t-1}}_{\rtree{t, \bfa}, n+3}(v_i)
                \ =\ 
                \sum_{i \in \intsinterval{\ell}}\feat{\ell_{t-1}}_{\rtree{t-1, \bfa_i}, n+3}(v_i)
                \ =\ 
                \sum_{i \in \intsinterval{\ell}} 1 \\
                \feat{\ell_t}_{\rtree{t, \bfa}, n+2}(v_r)
                \ =\ &
                \sum_{u \in \nbr{\out}(v_r)}\feat{\ell_{t-1}}_{\rtree{t, \bfa}, n+4}(u)
                % \ =\ 
                % \sum_{i \in \intsinterval{\ell}}\feat{\ell_{t-1}}_{\rtree{t, \bfa}, n+4}(v_i)
                \ =\ 
                \sum_{i \in \intsinterval{\ell}}\feat{\ell_{t-1}}_{\rtree{t-1, \bfa_i}, n+4}(v_i).
            \end{aligned}
        \end{equation*}
        For $1 \le i \le \ell$,
        by the induction hypothesis,
        ite holds that $\feat{\ell_{t-1}}_{\rtree{t-1, \bfa_i}, n+4}(v_i) = 1$.
        Thus, we have
        \begin{equation*}
            \feat{\ell_t}_{\rtree{t, \bfa}, n+1}(v_r)
            \ =\ 
            \feat{\ell_t}_{\rtree{t, \bfa}, n+2}(v_r).
        \end{equation*}
    \end{itemize}
\end{proof}

\section{Discussion} \label{sec:discuss}

This work extends the exploration of the relationship between aggregate-combine GNNs and logic, with exact characterizations of expressiveness for GNNs with eventually constant activation functions, and embedding a logic into the GNNs with standard $\relu$ activations. 
We also obtain  both decidability and undecidability results, some using the logical characterizations and some by porting the techniques used for decidability of the logics to apply directly on the GNNs. Perhaps the main take-away, echoing the theme
of \cite{barceloetallogical}, is that Presburger logics and the techniques for analyzing them can be relevant to  GNNs.

We give a fairly comprehensive picture of the complex of satisfiability and universal satisfiability:
see Table \ref{tab:sat} and \ref{tab:univsat}. We leave open a couple cases, such as truncated $\relu$, outgoing-only, with global readout.

In our work we leave open the question of getting an exact logical characterization of the expressiveness
of GNNs with regular $\relu$. We presented a logic that was contained in this class, which was sufficient
to obtain negative results on verification.
Also note that in this work we considered GNNs with bounded input features. We believe our techniques extend to the case of graphs where the features are attributes with unbounded precision: we are examining this in ongoing work.

%michael:
%We have left open one major technical problem: the decidability of satisfiability for standard GNNs using the $\relu$ activation function. Here we have proven decidability only for the ``outgoing-only'' variant. We also do not know whether the undecidability results we have proven  -- e.g. for standard GNNs with global readout -- still hold for the variants with outgoing-only aggregation. Thus, for all we know,  the most crucial dividing line for decidability could revolve around  outgoing-only vs bidirectional aggregation, rather than (e.g.) local vs global aggregation or truncation vs non-truncation in the activation function.

Lastly, note that our  complexity analysis is rather limited. We provide complexity bounds for the outgoing-only case with truncated $\relu$, and for GNNs based on truncated $\relu$ and local aggregation, we have shown satisfiability is $\pspace$-complete, and is $\np$-complete for a fixed number of layers.  Of course, for the broad class of GNNs with eventually constant activation functions, it is difficult to talk about complexity bounds.  But for specific eventually constant functions one could hope to say more.

Looking at broader open issues, we focused here on some very basic verification problems on GNNs: can a certain classification be achieved?
But it is clear that our techniques apply to many other logic-based verification problems; for example, it can be applied to determine whether a GNN can achieve a certain classification on a graph satisfying a certain sentence -- provided that the sentence is also in one of our decidable logics.
In ongoing work, we are developing verification tools based on the ideas in the paper, via reduction to problems that can be decided using existing solvers.

Our work provides motivation for exploring the properties of Presburger logics over  relational structures and their connections with GNNs beyond the setting here, which considers only graphs with discrete feature values from a fixed set. In our ongoing work we are  adapting our techniques to deal with GNNs whose feature values are \emph{unbounded integers}, specified by an initial semi-linear set.

%\michael{Discuss what we know about complexity -- the upper bounds from \cite{localpresburgerbartosztony} and the issue of blow up in the translation}

%\michael{Mention status of decidability of sat for unbounded}

%\michael{If we don't have results about unbounded with global readout, mention this as open}

In the body of the paper, we have worked with a directed graph model. As mentioned earlier, all of our
 results which will deal with  bidirectional GNNs also apply to undirected graphs: we explain this in the appendix.
 
\bibliography{references}

\newpage
\appendix
\section{Results for Undirected Graphs}

We summarize the results of the complexity and decidability of verification problems for GNNs over undirected graphs in Table~\ref{tab:sat_undirected} and Table~\ref{tab:univsat_undirected}.
In cases where we have only written ``decidable', we have not computed precise bounds: and in the ``eventually constant'' case we would require more hypotheses in order to compute such bounds.

\begin{table}[h!]
    \tbl{Complexity and decidability results for the satisfiability problem for fragments of GNNs over undirected graphs.}
    {\begin{tblr}{
        colspec={l|c|c},
        hline{2, 3, 4, 5},
        cell{2}{3} = {r=4}{m},
        cell{4}{2} = {r=2}{m},
        }
        & Local ($\cL$) & Global \\
        $\trrelu$ &
        \makecell[c]{$\pspace$-complete \\(Theorem~\ref{thm:pspace_undirected})} &
        \makecell[c]{Undecidable \\(Theorem~\ref{thm:global_gnn_undirected_undecidable})} \\
        Eventually constant ($\cC$) &
        \makecell[c]{Decidable \\(Theorem~\ref{thm:local_gnn_undirected_decidable})} \\
        $\relu$ &
        \makecell[c]{Undecidable \\(Theorem~\ref{thm:hilbert_to_blrelugnn_undirected})} \\
        Piecewise linear ($\pw$) &
    \end{tblr}} \label{tab:sat_undirected}
\end{table}

\begin{table}[h!]
    \tbl{Decidability results for the universal satisfiability problem for fragments of GNNs over undirected graphs.}
    {\begin{tblr}{
        colspec={l|c|c},
        hline{2, 3, 4, 5},
        cell{2}{3} = {r=4}{m},
        cell{2}{2} = {r=2}{m},
        cell{4}{2} = {r=2}{m},
        }
        & Local ($\cL$) & Global \\
        $\trrelu$ &
        \makecell[c]{Decidable \\(Theorem~\ref{thm:local_gnn_universal_undirected_decidable})} &
        \makecell[c]{Undecidable \\(Theorem~\ref{thm:global_gnn_universal_undirected_undecidable})} \\
        Eventually constant ($\cC$) \\
        $\relu$ &
        \makecell[c]{Undecidable \\(Theorem~\ref{thm:gnn_unbounded_undecun})} \\
        Piecewise linear ($\pw$) 
    \end{tblr}} \label{tab:univsat_undirected}
\end{table}

%\michael{The table is said to be about ''complexity results'', but then it just %mentions ''decidable'' vs ''undecidable''}
%\ch{fixed} %michael: thanks
\subsection{Undirected decidability results for Section \ref{subsec:spectrum}}
%\michael{Typo in the section header -- should be 4.1}

In the body of the paper, we mentioned that the decidability results for GNNs with eventually constant activations also apply to undirected graphs. Recall that we consider an undirected graph as a directed graph where the edge relation is symmetrical. 
We now explain the modifications needed to adapt the results in Section \ref{subsec:spectrum} to the undirected case.

We first prove the undirected version of Corollary~\ref{corollary:mp2_decidable}.
The main idea is that we can enforce undirectedness within the larger decidable logic $\GPtwo$ to obtain decidability:

\begin{corollary}\label{corollary:local_mp2_undirected_decidable}
    The finite satisfiability problem of $\blMPtwo$ over undirected graphs is decidable.
\end{corollary}

\begin{proof}
    The proof is similar to that of Corollary~\ref{corollary:mp2_decidable}.
    Let $\varphi(x)$ be an $n$-$\blMPtwo$ formula and $U_{n+1}$ be a fresh unary predicate.
    We claim that $\varphi(x)$ is finitely satisfiable over undirected graphs if and only if the $\GPtwo$ sentence
    \begin{equation*}
        \psi'\ :=\ \psi \land \left( \forall x\ \forall y\ E(x, y) \to E(y, x)\right)
    \end{equation*}
    is also finitely satisfiable,
    where $\psi := \exists x\ U_{n+1}(x) \land \varphi(x)$ is the $\GPtwo$ sentence defined in Corollary~\ref{corollary:mp2_decidable}.
    Then the corollary follows from the decidability of the finite satisfiability problem of $\GPtwo$ by Theorem~\ref{thm:gp2_decidabel}.
    
    If $\varphi(x)$ is finitely satisfiable by the undirected $n$-graph $\cG$ and vertex $v \in G$,
    let $\cG'$ be the $(n+1)$-graph that extended $\cG$ with $U_{n+1} := \set{v}$.
    By the argument similar to Corollary~\ref{corollary:mp2_decidable},
    we have $\cG' \models \psi$.
    Since $\cG$ is a undirected graph, for every vertices $v, u \in V$,
    if $(v, u) \in E$, then $(u, v) \in E$,
    which implies that $\cG' \models \forall x\ \forall y\ E(x, y) \to E(y, x)$.
    Hence, $\cG' \models \psi'$.

    If $\psi$ is finitely satisfiable by the $(n+1)$-graph $\cG$,
    let $\cG'$ be the $n$-graph that restricted $\cG$ by removing $U_{n+1}$.
    By the argument similar to Corollary~\ref{corollary:mp2_decidable},
    there exists a vertex $v \in V$ such that $\cG' \models \varphi(v)$.
    Since $\cG \models \forall x\ \forall y\ E(x, y) \to E(y, x)$,
    for every vertices $v, u \in V$,
    if $(v, u) \in E$, then $(u, v) \in E$.
    Therefore $\cG$ is a undirected graph, and so is $\cG'$.
    % Hence $\varphi(x)$ is finitely satisfiable by the undirected graph $\cG'$.
\end{proof}

Recall that in Section \ref{subsec:spectrum} we obtained decidability results by reducing to satisfiability of the logic
over directed graphs.
It is now straightforward to show that by reducing to decidability in the logic $\blMPtwo$ over undirected graphs, shown above, 
we extend the results to the undirected case. That is,   satisfiability and universal satisfiability problem for local GNNs with eventually constant activations over undirected graphs -- that is, the standard notion of GNN -- is decidable.
\begin{theorem}\label{thm:local_gnn_undirected_decidable}
    The satisfiability problem of $\blcGNN$s over undirected graphs is decidable.
\end{theorem}

%\ch{it is trivial}
% \begin{proof}
%     The proof is similar to Theorem~\ref{thm:local_gnn_decidable}.
%     For every $\blcGNN$ $\cA$, by Theorem~\ref{thm:gnn_to_logic},
%     there exists a $\blMPtwo$ formula $\Psi_\cA(x)$ such that $\cA$ and $\Psi_\cA(x)$ are equivalent.
%     We claim that $\cA$ is satisfiable over undirected graphs if and only if $\Psi_\cA(x)$ is finitely satisfiable over undirected graphs.
%     Since the finite satisfiability problem over undirected graphs of $\blMPtwo$ is decidable by Corollary~\ref{corollary:local_mp2_undirected_decidable},
%     we conclude that satisfiability over undirected graphs of $\blcGNN$s is decidable.
% \end{proof}

\begin{theorem}\label{thm:local_gnn_universal_undirected_decidable}
    The universal satisfiability problem of $\blcGNN$s over undirected graphs is decidable.
\end{theorem}

% \begin{proof}
%     The proof is similar to Theorem~\ref{thm:local_gnn_universal_decidable}
%     with the argument for undirected graphs in Corollary~\ref{corollary:local_mp2_undirected_decidable}.
%     For every $\blcGNN$ $\cA$, by Theorem~\ref{thm:gnn_to_logic},
%     there exists a $\blMPtwo$ formula $\varphi_\cA(x)$ such that $\cA$ and $\varphi_\cA(x)$ are equivalent.
%     We can show that $\cA$ is universally satisfiable over undirected graphs if and only if the $\GPtwo$ sentence
%     \begin{equation*}
%         \psi\ :=\ \left(\forall x\ (x=x) \to \varphi_\cA(x)\right)
%         \land 
%         \left(\forall x\ \forall y\ E(x, y) \to E(y, x)\right)
%     \end{equation*}
%     is finitely satisfiable.
%     Since the finite satisfiability of $\GPtwo$ is decidable by Theorem~\ref{thm:gp2_decidabel},
%     the theorem follows.
% \end{proof}

\subsection{Undirected variants of Section~\ref{subsec:global_undecidable}: undecidability in the eventually constantt case}

We can also revise our undecidability results for global GNNs with eventually constant activations to the undirected case,
thus giving undecidability for the usual notion of GNN with global readout.
This is done with the same reduction from solvability of simple equation systems to the finite satisfiability of $\bgMPtwo$ formulas, which we can show works over undirected graphs.
The main difference is the following stronger version of Lemma~\ref{lemma:hilbert_to_mptwo}.

\begin{lemma}\label{lemma:hilbert_to_mptwo_undirected}
    For every simple equation system $\varepsilon$ with $n$ variables and $m$ equations,
    there exists an $(n+m)$-$\MPtwo$ formula $\Psi_\varepsilon(x)$
    such that the following are equivalent,
    \begin{enumerate}
        \item The system $\varepsilon$ has a solution in $\bbN$.
        \item There exists an undirected $(n+m)$-graph $\cG$ such that for every vertex $v$ in $\cG$, $\cG \models \Psi_\varepsilon(v)$.
        \item The formula $\Psi_\varepsilon(x)$ is finitely satisfiable over undirected graphs.
    \end{enumerate}
\end{lemma}

\begin{proof}
    Let $\Psi_\varepsilon(x)$ be the $(n+m)$-$\MPtwo$ formula which is defined in Lemma~\ref{lemma:hilbert_to_mptwo}.
    The proof for (2) $\Rightarrow$ (3) and (3) $\Rightarrow$ (1) are the same.
    We only need to revise the proof for (1) $\Rightarrow$ (2).
    Let $\cG'$ be the graph obtained by adding symmetric edges to $\cG$ in the proof.
    Clearly, $\cG'$ is undirected.
    It can be verified that all arguments for $\cG$ also holds for $\cG'$.
    % It is worth noting that the graph $\cG$ defined in the proof of Lemma~\ref{lemma:hilbert_to_mptwo}
    % is actually a undirected graph.
    Thus, the property holds trivially.

    % It is almost the same, the only difference is that now we change the construction of edges to
    % \begin{equation*}
    %     E := \bigcup_{i \in \intsinterval{m}} \setc{(v, u), (u, v)}{(v, u) \in E_i}
    % \end{equation*}
    % It is clear that the graph $\cG$ is now undirected,
    % and it is routine to check that for every $v \in \cG$, $\cG \models \Psi_\varepsilon(v)$. 
\end{proof}

The following undecidability results over undirected graphs can be proven based on Lemma~\ref{lemma:hilbert_to_mptwo_undirected}.
% The proofs are similar to Theorem~\ref{thm:global_mptwo_undecidable}, Theorem~\ref{thm:global_gnn_undecidable}, and Theorem~\ref{thm:global_gnn_universal_undecidable}.
\begin{theorem}\label{thm:global_mptwo_undirected_undecidable}
    The finite satisfiability problem of $\bgMPtwo$ over undirected graphs is undecidable.
\end{theorem}

\begin{theorem}\label{thm:global_gnn_undirected_undecidable}
    The satisfiability problem of $\bgtrreluGNN$s over undirected graphs is undecidable.
\end{theorem}

\begin{theorem}\label{thm:global_gnn_universal_undirected_undecidable}
    The universal satisfiability problem of $\bgtrreluGNN$s over undirected graphs is undecidable.
\end{theorem}

\subsection{Undirected $\pspace$-completeness results for Section \ref{sec:pspace}}

We can also revise our $\pspace$-completeness results for local GNNs with $\trrelu$ to the undirected case:

\begin{theorem}\label{thm:pspace_undirected}
    The satisfiability problem of $\bltrreluGNN$s is $\pspace$-complete,
    and it becomes $\np$-complete when the number of layers is fixed.
\end{theorem}

The revisions for $\pspace$ algorithm and $\np$-complete results are obvious.
Here we only show the $\pspace$ lower bound to the undirected case.

The idea of the proof is similar to Lemma~\ref{lemma:alc}.
We establishe the $\pspace$ lower bound
by reducing the concept satisfiability of the description logic $\alc$
into the satisfiability problem of $\blMPtwo$ over undirected graphs.
We introduce fresh unary predicates $P_0, \ldots, P_L$, and
enforce that the vertex with depth $\ell$ only realizes $P_\ell$.
% To guarantee the structure of the tree.
Thus, it is routine to convert models of the $\alc$ concept to the $\blMPtwo$ formula and vice versa.
% Since the concept satisfiability problem of $\alc$ with one role is  $\pspace$-hard~\cite{alc-pspace},
% it will follow from the embedding that the finite satisfiability problem of $\olMPtwo$ is also $\pspace$-hard.
% Since the reduction from $\olMPtwo$ to $\oltrreluGNN$ mentioned in Theorem~\ref{thm:gnn_to_logic} is polynomial,
% the satisfiability problem of $\oltrreluGNN$ is also $\pspace$-hard, and so it $\bltrreluGNN$.

\begin{lemma}
    There exists a polynomial time translation $\pi$ from $\alc$ concepts with one role $R$ to $\blMPtwo$ formulas such that the $\alc$ concept $C$ is satisfiable if and only if the $\blMPtwo$ formula $\pi(C)$ is finitely satisfiable over undirected graphs.
\end{lemma}

\begin{proof}
    Suppose that the quantifier depth of the $\alc$ concept $C$ is $L$.
    The vocabulary of the $\blMPtwo$ formula $\pi_x(C)$ is the vocabulary of the $\alc$ concept $C$ extended with fresh unary predicates $P_0, \ldots, P_L$.
    For $0 \le \ell \le L$, let
    \begin{equation*}
        \psil{\ell}(x)\ :=\ 
        \left( \bigwedge_{i \in \intsinterval{0, \ell-1}} \neg P_i(x) \right) \land 
        P_\ell(x) \land
        \left( \bigwedge_{i \in \intsinterval{\ell-1, L}} \neg P_i(x) \right).
    \end{equation*}

    We will define translations $\trans{\ell}{x}$ and $\trans{\ell}{y}$ inductively, which are the standard translation from $\alc$ concepts to first-order logic formulas, except with some slight modification to quantifiers to fit our logic and guarantee the depth of elements in its tree models.
    Note that the base case of the inductive definition is $\trans{L}{x}$ and $\trans{L}{x}$.
    Since the quantifier depth of the $\alc$ concept is $L$,
    for the translations $\trans{L}{x}$ and $\trans{L}{x}$,
    there is no case for quantifiers. 
    The rules for the translations are
    \begin{equation*}
        \begin{aligned}
            \trans{\ell}{x}(A)\ =\ &A(x)&
            \trans{\ell}{y}(A)\ =\ &A(y)
            \\
            \trans{\ell}{x}(\neg C)\ =\ &\neg \trans{\ell}{x}(C)&
            \trans{\ell}{y}(\neg C)\ =\ &\neg \trans{\ell}{y}(C)
            \\
            \trans{\ell}{x}(C \sqcap D)\ =\ &\trans{\ell}{x}(C) \land \trans{\ell}{x}(D)&
            \trans{\ell}{y}(C \sqcap D)\ =\ &\trans{\ell}{y}(C) \land \trans{\ell}{y}(D)
            \\
            \trans{\ell}{x}(C \sqcup D)\ =\ &\trans{\ell}{x}(C) \lor \trans{\ell}{x}(D)&
            \trans{\ell}{y}(C \sqcup D)\ =\ &\trans{\ell}{y}(C) \lor \trans{\ell}{y}(D)
        \end{aligned}
    \end{equation*}
    and
    \begin{equation*}
        \begin{aligned}
            \trans{\ell}{x}(\exists R.C)\ =\ &\presby{E(x, y) \land \psil{\ell+1}(y) \land \trans{\ell+1}{y}(C)} \ge 1
            \\
            \trans{\ell}{y}(\exists R.C)\ =\ &\presbx{E(y, x) \land \psil{\ell+1}(x) \land \trans{\ell+1}{x}(C)} \ge 1
            \\
            \trans{\ell}{x}(\forall R.C)\ =\ &\presby{E(x, y) \land \psil{\ell+1}(y) \land \neg \trans{\ell+1}{y}(C)} = 0
            \\
            \trans{\ell}{y}(\forall R.C)\ =\ &\presbx{E(y, x) \land \psil{\ell+1}(x) \land \neg \trans{\ell+1}{x}(C)} = 0.
        \end{aligned}
    \end{equation*}
    Finally, we define $\pi(C) := \psil{0}(x) \land \trans{0}{x}(C)$.

    % Suppose that the $\alc$ concept $C$ is satisfiable by the finite tree model $\cT$ with depth $L$.
    % Let $\cT'$ be the undirected version of $\cT$ extended with following componments.
    % For $0 \le \ell \le L$,
    % $P_i$ is the set of vertices with depth $\ell$.
    % It is routine to check that $\psil{0}(x) \land \trans{0}{x}(C)$ is satisfied by the root of $\cT'$.
\end{proof}

\subsection{Undirected results for Section \ref{subsec:sep_unbounded}: undecidability and separation for GNNs with unbounded activations}

We now briefly explain why the undecidability and expressiveness separation results for GNNs with unbounded activation functions also apply to undirected graphs.

For the expressiveness results, recall that the key idea relies on the property:
``the number of two-hop paths from the vertex $v$ to the green vertices is the same as the number of two-hop paths from the vertex $v$ to the blue vertices.''
The separation is proved by constructing a sequence of pairs of bipolar graphs.
To extend this to undirected graphs, we define an “undirected” version of bipolar graphs by adding symmetric edges to the standard definition of bipolar graphs.
It is then routine to verify that all arguments hold for this undirected version of bipolar graphs.
Therefore, we obtain the following results for undirected graphs.

% in the two-hop logic. To show that no $\blcGNN$ can express it, we construct a sequence of pairs of graphs, each with a special node, such that the property holds
% we define ``undirected'' version of bipolar graph
% note that the graphs that we constructed in the proof of Lemma \ref{lemma:mtwoptwo_gap1} are undirected. Hence the expressiveness gap between $\blreluGNN$ and $\blcGNN$ still exists for the undirected case.

\begin{lemma}
    $\blMtwoPtwo$ is strictly more expressive over undirected graphs than $\blcGNN$s.
\end{lemma}

\begin{corollary}
    $\blMtwoPtwo$ is strictly more expressive over undirected graphs than $\blMPtwo$.
\end{corollary}

\begin{corollary}
    $\blreluGNN$s  are strictly more expressive over undirected graphs than $\blcGNN$s.
\end{corollary}

% Recall that the bipolar graphs are undirected graphs.
% \ch{need to rewrite}
% Hence by the same argument as in Lemma~\ref{lemma:mtwoptwo_gap1}
% we obtain the same reduction for undirected graphs.

% \begin{lemma}
%     There exists a $\blMtwoPtwo$ formula $\Psi(x)$ such that 
%     for every $\blcGNN$ $\cA$, $\Psi(x)$ and $\cA$ are not equivalent \emph{over undirected graphs}.
% \end{lemma}
% Theorem \ref{thm:unbounded_gnn_undirected} is a direct consequence of the lemma above.

\subsection{Undirected results for Section \ref{subsec:unbounded_undec_sat}}

Next, we turn to the satisfiability problem for $\blreluGNN$s.
For a simple equation system $\varepsilon$, we say that $\varepsilon$ has a non-trivial solution
if it has a solution where not all variables are zero.
For an undirected graph $\cG$ and a vertex $v$ in $\cG$, let $\nbr{}(v)$ be the set of neighbors of $v$.

\begin{lemma}
    For every simple equation system $\varepsilon$ with $n$ variables and $m$ equations,
    there exists an $(n+2m)$-$\blreluGNN$ $\cA_\varepsilon$ 
    such that
    $\varepsilon$ has a non-trivial solution in $\bbN$ if and only if
    $\cA_\varepsilon$ is satisfiable over undirected graphs.
\end{lemma}

\begin{proof}
    The proof is similar to Lemma~\ref{lemma:hilbert_to_blrelugnn}.

    We define the 20-layer $(n+2m)$-$\blreluGNN$ $\cA_\varepsilon$ with 
    \emph{integer} coefficient matrices and bias vectors as follows.
    % The input dimension $\gnndim{0}$ is $n + 2m$.
    The dimensions, coefficient matrices, and bias vectors will be chosen so that the features of $\cA_\varepsilon$ satisfy the following conditions.
    Let 
    \begin{equation*}
        \begin{aligned}
            \Psi_0(\bfx)\ :=\ &
            \bigvee_{t \in \intsinterval{m}} x_t \ge 1 \\
            \Psi_1(\bfx)\ :=\ &
            \bigwedge_{t \in \intsinterval{m}} \Psi_{1, t}(\bfx)\ \land\ 
            \left(x_{2n+2m+1} = 1\right) \\
            \Psi_2(\bfx)\ :=\ &
            \bigwedge_{t \in \intsinterval{m}} \Psi_{2, t}(\bfx)\ \land\ 
            \left(x_{n+2m+1} = x_{n+2m+2}\right)\ \land\ \Psi_0(\bfx),
        \end{aligned}
    \end{equation*}
    where, for $1 \le t \le m$,
    \begin{itemize}
        \item if the $t^{th}$ equation in $\varepsilon$ is $\upsilon_{t_1} = 1$,
        then $\Psi_{1, t}(\bfx) := \top$ and $\Psi_{2, t}(\bfx) := \left(x_{t_1} = 1\right)$.

        \item if the $t^{th}$ equation in $\varepsilon$ is $\upsilon_{t_1} =\upsilon_{t_2} + \upsilon_{t_3}$,
        then $\Psi_{1, t}(\bfx) := \top$ and
        $\Psi_{2, t}(\bfx) := \left(x_{t_1} = x_{t_2} + x_{t_3}\right)$.

        \item if the $t^{th}$ equation in $\varepsilon$ is $\upsilon_{t_1} =\upsilon_{t_2} \cdot \upsilon_{t_3}$,
        then 
        \begin{equation*}
            \begin{aligned}
                \Psi_{1, t}(\bfx)\ :=\ &\left(\left(x_{n+t} = 0\right) \land \left(x_{n+m+t} = 0\right)\right) \lor \left(\left(x_{n+t} = x_{n+2m+t_2}\right) \land \left(x_{n+m+t} = 1\right)\right) \\
                \Psi_{2, t}(\bfx)\ :=\ &\left(x_{n+t} = x_{t_1}\right) \land \left( x_{n+m+t} = x_{t_3}\right).
            \end{aligned}
        \end{equation*}
    \end{itemize}

    For the first layer,
    \begin{equation*}
        \feat{1}_i(v)\ =\ 
        \begin{dcases}
            \sum_{u \in \nbr{}(v)} \feat{0}_i(u), &\text{for $1 \le i \le n+m$} \\ 
            \feat{0}_i(v), &\text{for $n+m < i \le n+2m$}.
        \end{dcases}
    \end{equation*}

    For the second to the fifth layers,
    $\cA_\varepsilon$ test whether $\Psi_0\left(\feat{1}(v)\right)$ holds.
    That is
    \begin{equation*}
        \feat{5}_i(v)\ =\ 
        \begin{dcases}
            \feat{1}_i(v), &\text{for $1 \le i \le n+2m$} \\
            \eval{\psi_0\left(\feat{1}(v)\right)}, &\text{for $i = n+2m+1$}.
        \end{dcases}
    \end{equation*}

    For the sixth layer,
    \begin{equation*}
        \feat{6}_i(v)\ =\ 
        \begin{dcases}
            \feat{5}_i(v), &\text{for $1 \le i \le n+2m$} \\
            \sum_{u \in \nbr{}(v)} \feat{5}_{i - n - 2m}(u), &\text{for $n+2m < i \le 2n+2m$} \\
            \sum_{u \in \nbr{}(v)} \feat{5}_{i - n}(u), &\text{for $i = 2n + 2m + 1$} \\
            1, &\text{for $i = 2n + 2m + 2$}.
        \end{dcases}
    \end{equation*}

    For the seventh to the fourteenth layers,
    $\cA_\varepsilon$ test whether $\Psi_1\left(\feat{6}(v)\right)$ holds.
    That is,
    \begin{equation*}
        \feat{14}_i(v)\ =\ 
        \begin{dcases}
            \feat{6}_i(v), &\text{for $1 \le i \le 2n+2m+2$} \\
            \eval{\psi_1\left(\feat{6}(v)\right)}, &\text{for $i = 2n+2m+3$}.
        \end{dcases}
    \end{equation*}

    For the fifteenth layer,
    \begin{equation*}
        \feat{15}_i(v)\ =\ 
        \begin{dcases}
            \feat{14}_i(v), &\text{for $1 \le i \le n$} \\
            \sum_{u \in \nbr{}(v)} \feat{14}_{i}(u), &\text{for $n < i \le n + 2m$} \\
            \sum_{u \in \nbr{}(v)} \feat{14}_{i+n+1}(u), &\text{for $n + 2m < i \le n + 2m + 2$}.
        \end{dcases}
    \end{equation*}

    For the sixteenth to the nineteenth layers,
    $\cA_\varepsilon$ test whether $\Psi_2\left(\feat{15}(v)\right)$ holds.
    That is
    \begin{equation*}
        \feat{19}_i(v)\ =\ 
        \begin{dcases}
            \feat{19}_i(v), &\text{for $1 \le i \le n+2m+2$} \\
            \eval{\Psi_2\left(\feat{15}(v)\right)}, &\text{for $i = n+2m+3$}.
        \end{dcases}
    \end{equation*}

    Finally, for the last layer, $\feat{20}_{1}(v) = \feat{19}_{n+2m+3}(v)$.
    By Lemma~\ref{lemma:fnn}, it is easy to see that there are $\blreluGNN$s that satisfy these conditions.
    
    We now show that this construction has the required properties for the lemma.
    \textbf{\underline{Suppose $\cA_\varepsilon$ is satisfiable}} by the graph $\cG$ and vertex $v$ in $\cG$.
    Let $a_i := \feat{1}_i(v)$.
    We will show that $\set{\upsilon_i \gets a_i}_{i \in \intsinterval{n}}$ is a solution of $\varepsilon$.

    The proof idea is similar to Lemma~\ref{lemma:hilbert_to_blrelugnn},
    with the only difference being the following variant of the main claim used in that lemma:

    % \medskip
    The following properties hold for every $u \in \nbr{}(v)$.
    \begin{itemize}
        \item For $1 \le i \le n$,
        $\feat{6}_{n+2m+i}(u) = a_i$.
        \item For $1 \le t \le m$,
        if the $t^{th}$ equation in $\varepsilon$ is $\upsilon_{t_1} = \upsilon_{t_2} \cdot \upsilon_{t_3}$,
        then
        \begin{equation*}
            \feat{14}_{n+t}(u)\ =\ a_{t_2} \cdot \feat{14}_{n+m+t}(u).
        \end{equation*}
    \end{itemize}
    % \medskip

    We now prove the claim.
    By the argument similar to Lemma~\ref{lemma:hilbert_to_blrelugnn}.
    For every $u \in \nbr{}(v)$,
    $\Psi_1\left(\feat{6}(u)\right)$ holds.
    \begin{itemize}
        \item By the definition of $\cA_\varepsilon$, it holds that
        \begin{equation*}
            \feat{6}_{2n+2m+1}(u)
            \ =\ \sum_{u' \in \nbr{}(u)} \feat{5}_{n+2m+1}(u)
            \ =\ \sum_{u' \in \nbr{}(u)} \eval{\Psi_0\left(\feat{1}(u')\right)}.
        \end{equation*}
        Because $\Psi_1\left(\feat{6}(u)\right)$ holds,
        it follows that $\feat{6}_{2n+2m+1}(u) = 1$.
        This implies that
        there is only one neighbor $\widetilde{u}$ of $u$
        for which $\Psi_0\left(\feat{1}(\widetilde{u})\right)$ holds.
        Therefore,
        for $1 \le i \le n$,
        for every $u' \in \nbr{}(u)$ with $u' \neq \widetilde{u}$,
        $\feat{1}_i(u') = 0$.
        Thus, for $1 \le i \le n$,
        by the definition of $\cA_\varepsilon$,
        it holds that
        \begin{equation*}
            \begin{aligned}
                \feat{6}_{n+2m+i}(u)
                \ =\ \sum_{u' \in \nbr{}(u)} \feat{5}_{i}(u')
                \ =\ \sum_{u' \in \nbr{}(u)} \feat{1}_{i}(u')
                \ =\ \feat{1}_{i}(\widetilde{u}).
            \end{aligned}
        \end{equation*}
        Once we show that $\widetilde{u} = v$,
        it follows that
        $\feat{6}_{n+2m+i}(u) = \feat{1}_{i}(\widetilde{v}) = a_i$.

        Since $v$ is a neighbor of $u$ and
        $\widetilde{u}$ is the only neighbor of $u$
        for which $\Psi_0\left(\feat{1}(\widetilde{u})\right)$ holds.
        To show that $\widetilde{u} = v$,
        it is sufficient to demonstrate that $\Psi_0\left(\feat{1}(v)\right)$ holds.

        Because $\cA_\varepsilon$ is satisfied by $\tuple{\cG, v}$,
        $\feat{20}_{1}(v) = \feat{19}_{n+2m+3}(v) = \eval{\Psi_2\left(\feat{15}(v)\right)} = 1$.
        Since $\Psi_2\left(\feat{15}(v)\right)$ holds,
        $\Psi_0\left(\feat{15}(v)\right)$ also holds.
        Note that, for $1 \le i \le n$,
        by the definition of $\cA_\varepsilon$,
        $\feat{15}_i(v) = \feat{1}_i(v)$.
        Hence $\Psi_0\left(\feat{1}(v)\right)$ holds.

        \item Fix $1 \le t \le n$, and suppose that 
        the $t^{th}$ equation in $\varepsilon$ is $\upsilon_{t_1} = \upsilon_{t_2} \cdot \upsilon_{t_3}$.
        Since $\Psi_1\left(\feat{6}(u)\right)$ holds,
        $\Psi_{1, t}\left(\feat{6}(u)\right)$ also holds.
        Thus, one of the following must hold:
        \begin{equation*}
            \feat{6}_{n + t}(u) = 0 \land \feat{6}_{n + m + t}(u) = 0
            \quad\text{or}\quad
            \feat{6}_{n + t}(u) = \feat{6}_{n + 2m + t_2}(u) \land \feat{6}_{n + m + t}(u) = 1.
        \end{equation*}
        Note that by the definition of $\cA_\varepsilon$,
        \begin{equation*}
            \feat{14}_{n + t}(u)\ =\ \feat{6}_{n + t}(u)
            \quad\text{and}\quad
            \feat{14}_{n + m + t}(u) = \feat{6}_{n + m + t}(u).
        \end{equation*}
        By the previous claim,
        $\feat{6}_{n+2m+t_2}(u) = a_{t_2}$.
        Therefore, for both cases, 
        $\feat{14}_{n + t}(u) = a_{t_2} \cdot \feat{14}_{n + m + t}(u)$.
    \end{itemize}

    \textbf{\underline{Suppose $\varepsilon$ is solvable}}.
    Let $\cG'$ be the graph obtained by adding symmetric edges to the directed graph $\cG$ constructed in the proof of Lemma~\ref{lemma:hilbert_to_ogrelugnn}.
    We can show that $\cA_\varepsilon$ accepts $\tuple{\cG, v_0}$
    via an argument similar to the one used in Lemma~\ref{lemma:hilbert_to_ogrelugnn}.
\end{proof}

Since the solvability (over $\bbN$) of simple equation systems is undecidable,
we obtain undecidability of the satisfiability problems for $\blreluGNN$.

\begin{theorem}\label{thm:hilbert_to_blrelugnn_undirected}
    The satisfiability problem for $\blreluGNN$s over undirected graphs is undecidable.
\end{theorem}

\subsection{Undirected results for Section \ref{subsec:unbounded_undec_univ_sat}}

To obtain undecidability of the universal satisfiability problem for $\blreluGNN$ over undirected graphs,
we again reduce from two-counter machines, but now with a modification to guarantee the direction of the transition of the counter machine.

\begin{lemma}\label{lemma:tcm_to_mtwoptwoun}
    For every two-counter machine $\cM$ with $n$ instructions, there exists an $(n+8)$-$\blMtwoPtwo$ formula $\Psi_\cM(x)$ such that $\cM$ halts if and only if $\forall x\ \Psi_\cM(x)$ is finitely satisfiable over undirected graphs. 
\end{lemma}

We cannot apply the exact encoding from Figure \ref{fig:two_counter_machine} and  Lemma~\ref{lemma:tcm_to_mtwoptwo} here,
because in that encoding we distinguished the predecessor and successor configurations using the direction of edges.
Here we sketch the trick used to overcome the lack of direction in the edges.
We will utilize the predicates from the proof of Lemma~\ref{lemma:tcm_to_mtwoptwo}. In particular
we will have a predicate $Q$ and an associated notion of $Q$ vertex as in that proof.

We introduce three fresh unary predicates $I_0$, $I_1$, and $I_2$ to label the configuration's index modulo $3$. 
We add an extra clause to the formula to guarantee that each element has exactly one of these three index labels.
Each $C_0$ or $C_1$ vertex has exactly the same index as the $Q$ vertex that connected to it.
Finally, for each $Q$ vertex $v$ with index $i$, 
there exists at most one $Q$ vertex $v'$ with index $(i + 1 \bmod 3)$,
such that $v$ and $v'$ are connected;
there exists at most one $Q$ vertex $v''$ with index $(i - 1 \bmod 3)$,
such that $v$ and $v'$ are connected.
Therefore we can modify the formula which identifies the successor and predecessor
based on the index, rather than the direction of the edges, and show that the two-counter machine halts if and only if the modified formula is finitely satisfiable over undirected graphs.

\begin{figure}[h]
    \centering
    \begin{tikzpicture}[thick, main/.style = {draw, circle, minimum size=30, scale=0.6}]
        \node[main] (c0) at (0,  0)  {$S$};
        \node[main] (c1) at (2,  0)  {$Q_1$};
        \node[main] (c2) at (4,  0)  {$Q_2$};
        \node[main] (c3) at (6,  0)  {$Q_3$};
        \node[main] (c4) at (8,  0)  {$Q_8$};
        \node[    ] (c5) at (9,  0)  {$\cdots$};
        \node[main] (c6) at (10,  0) {$Q_7$};
        \node[main] (c7) at (12, 0)  {$T$};
        \draw[] (c0) -- (c1);
        \draw[] (c1) -- (c2);
        \draw[] (c2) -- (c3);
        \draw[] (c3) -- (c4);
        \draw[] (c4) -- (c5);
        \draw[] (c5) -- (c6);
        \draw[] (c6) -- (c7);
        \node[main] (c1b1) at (2,  1)         {$C_0$};
        \node[main] (c1g1) at (2.342, -0.940) {$C_1$};
        \node[main] (c1g2) at (1.658, -0.940) {$C_1$};
        \draw[] (c1) -- (c1b1);
        \draw[] (c1) -- (c1g1);
        \draw[] (c1) -- (c1g2);
        \node[main] (c2b1) at (4.342, 0.940)  {$C_0$};
        \node[main] (c2b2) at (3.658, 0.940)  {$C_0$};
        \node[main] (c2g1) at (4.342, -0.940) {$C_1$};
        \node[main] (c2g2) at (3.658, -0.940) {$C_1$};
        \draw[] (c2) -- (c2b1);
        \draw[] (c2) -- (c2b2);
        \draw[] (c2) -- (c2g1);
        \draw[] (c2) -- (c2g2);
        \node[main] (c3b1) at (6.342, 0.940)  {$C_0$};
        \node[main] (c3b2) at (5.658, 0.940)  {$C_0$};
        \node[main] (c3g1) at (6, -1) {$C_1$};
        \draw[] (c3) -- (c3b1);
        \draw[] (c3) -- (c3b2);
        \draw[] (c3) -- (c3g1);
        \node[main] (c4b1) at (8.342, 0.940)  {$C_0$};
        \node[main] (c4b2) at (7.658, 0.940)  {$C_0$};
        \draw[] (c4) -- (c4b1);
        \draw[] (c4) -- (c4b2);
        \node[main] (c6b1) at (10.342, 0.940)  {$C_0$};
        \node[main] (c6b2) at (9.658, 0.940)   {$C_0$};
        \draw[] (c6) -- (c6b1);
        \draw[] (c6) -- (c6b2);

        \draw [dashed, gray] (1.1, 1.5) -- (2.9, 1.5) -- (2.9, -1.5) -- (1.1, -1.5) -- (1.1, 1.5);
        \draw [dashed, gray] (3.1, 1.5) -- (4.9, 1.5) -- (4.9, -1.5) -- (3.1, -1.5) -- (3.1, 1.5);
        \draw [dashed, gray] (5.1, 1.5) -- (6.9, 1.5) -- (6.9, -1.5) -- (5.1, -1.5) -- (5.1, 1.5);
        \draw [dashed, gray] (7.1, 1.5) -- (8.9, 1.5) -- (8.9, -1.5) -- (7.1, -1.5) -- (7.1, 1.5);
        \draw [dashed, gray] (9.1, 1.5) -- (10.9, 1.5) -- (10.9, -1.5) -- (9.1, -1.5) -- (9.1, 1.5);
        \node[] () at (2, -1.8)  {$I_0$};
        \node[] () at (4, -1.8)  {$I_1$};
        \node[] () at (6, -1.8)  {$I_2$};
        \node[] () at (8, -1.8)  {$I_0$};
        \node[] () at (10, -1.8)  {$I_{(i \bmod 3)}$};
        
    \end{tikzpicture}
    \centering
    \caption{An example of the encoding of the computation of a two-counter machines in undirected graphs.}
    \label{fig:two_counter_machine_un}
\end{figure}

Formally the reduction is as follows.
Given a two-counter machine $\cM$, 
we construct the formula $\Psi_\cM(x)$ as in the proof of 
Lemma~\ref{lemma:tcm_to_mtwoptwo}.
Here we demonstrate only the main idea of the modification;
other formulas can be treated analogously.
\begin{equation*}
    \begin{aligned}
        \psi^{\itdiff}_{t, \delta}(x)\ :=\ 
            \bigwedge_{i \in \intinterval{0}{2}} I_{i}(x) \to 
            (
                &\presbzy{E(x, z) \land E(z, y) \land C_t(y) \land I_{(i+1\bmod 3)}(y)} - \\
                &\presby{E(x, y) \land C_t(y) \land I_{i}(y)}  = \delta
            ) \\
        \psi^{\itsucc}_{j}(x)\ :=\ 
            \bigwedge_{i \in \intinterval{0}{2}} I_{i}(x) \to
            (
                &\presby{E(x, y) \land Q_j(y) \land I_{(i+1\bmod 3)}(y)} = 1
            )
    \end{aligned}
\end{equation*}
% Then the difference are the following formulas.
% \begin{equation*}
%     \begin{aligned}
%         \\
%         \varphi_4(x) \ := \ &
%         \left(\bigvee_{1\le i \le n}Q_i(x)\right) \land \left(\bigwedge_{1 \le i < j \le n} \neg Q_i(x) \lor \neg Q_j(x)\right)\land
%         \\
%         & 
%         \left(\bigvee_{0\le i \le 2}I_i(x)\right) \land \left(\bigwedge_{0 \le i < j \le 2} \neg I_i(x) \lor \neg I_j(x)\right)
%     \end{aligned}
% \end{equation*}

The proof that $\cM$ halts if and only if $\forall x \Psi_\cM(x)$
is finitely satisfiable over undirected graphs is similar to
the one for Lemma~\ref{lemma:tcm_to_mtwoptwo},
thus, omitted.

\begin{theorem}\label{thm:gnn_unbounded_undecun}
    The universal satisfiability problem of $\blreluGNN$s over undirected graphs is undecidable.
\end{theorem}

The theorem follows from  Lemma \ref{lemma:tcm_to_mtwoptwoun}, using the undecidability result for two-counter machines, as in Theorem \ref{thm:blgnn_unbounded_undec}.

\end{document}